\documentclass[12pt]{article}
\usepackage{amsmath, amsthm, amssymb, bbm, setspace,bigints}
\usepackage[margin=1 in]{geometry}
\usepackage{caption, enumitem}
\usepackage{subcaption}
\usepackage[toc,page]{appendix}
\usepackage{graphicx, amsfonts, tikz, multirow}

\usepackage{natbib}
\usepackage{mathtools}
\usetikzlibrary{bayesnet}
\usepackage{pdfpages}
\usepackage{epstopdf}
\usepackage{booktabs}
\usepackage{floatrow, lscape}
\usepackage[ruled]{algorithm2e}
\usepackage[colorlinks,citecolor=blue,urlcolor=blue]{hyperref}
\usepackage[utf8]{inputenc}
\usepackage{authblk}
\usepackage{comment}
\usepackage{xr}
\newfloatcommand{capbtabbox}{table}[][\FBwidth]

\def\Gauss{{ \mathrm{N} }}
\def\T{\mathrm{\scriptscriptstyle{T}}}

\def\ind{\mathbbm{1}}
\newcommand{\norm}[1]{\Vert#1\Vert}
\DeclareMathOperator*{\argmax}{arg\,max}

\allowdisplaybreaks
\newtheorem{theorem}{Theorem}[section]
\newtheorem{lemma}[theorem]{Lemma}

\newtheorem{rem}{Remark}[section]

\newtheorem{assumption}[theorem]{Assumption}

\newcommand\RR{\mathbb{R}}

\newcommand{\lnjbar}[1]{\bar{\ell}_{n j}(#1)}
\newcommand{\ljbar}[1]{\bar{\ell}_{j}(#1)}

\newcommand{\sigmatru}{\sigma^*_{12}}

\newcommand{\sndpartial}[3]{\frac{\partial^2{#1}}{\partial{#2}\partial{#3}}}
\newcommand{\overeq}[1]{\stackrel{(#1)}{=}}

\DeclareMathOperator{\tr}{tr}

\newcommand{\Var}{\mathrm{Var}}
\newcommand{\Cov}{\mathrm{Cov}}

\usepackage{mathtools}
\usepackage{cleveref}

\date{}
\title{Bayesian inference on high-dimensional multivariate binary responses}
\author[1]{Antik Chakraborty\thanks{antik015@purdue.edu}}
\author[2]{Rihui Ou\thanks{rihui.ou@duke.edu}}
\author[2]{David B. Dunson \thanks{dunson@duke.edu}}
\affil[1]{Department of Statistics, Purdue University}
\affil[2]{Department of Statistical Science, Duke University}
\begin{document}
\maketitle
\begin{abstract}
It has become increasingly common to collect high-dimensional binary response data; for example, with the emergence of new sampling techniques in ecology.  In smaller dimensions, multivariate probit (MVP) models are routinely used for inferences.  However, algorithms for fitting such models face issues in scaling up to high dimensions due to the intractability of the likelihood, involving an integral over a multivariate normal distribution having no analytic form.  Although a variety of algorithms have been proposed to approximate this intractable integral, these approaches are difficult to implement and/or inaccurate in high dimensions. Our main focus is in accommodating high-dimensional binary response data with a small to moderate number of covariates.
We propose a two-stage approach for inference on model parameters while taking care of uncertainty propagation between the stages. We use the special structure of latent Gaussian models to reduce the highly expensive computation involved in joint parameter estimation to focus inference on marginal distributions of model parameters. This essentially makes the method embarrassingly parallel for both stages. We illustrate performance in simulations and applications to joint species distribution modeling in ecology. 
\end{abstract}

\noindent%

\section{Introduction}
\label{sec:intro}

High-dimensional multivariate binary response data are routinely collected in many application areas.  We are particularly motivated by joint species distribution modeling in ecology \citep{warton2015so,ovaskainen2017make}.  In this setting, data consist of a high-dimensional vector of indicators of occurrences of different species in $n$ locations.  Interest focuses on inferences on the dependence structure across $q$ species, as well as the effects of $p$ covariates on the marginal occurrence probabilities. While $p$ tends to be small to moderate, automated sampling and species identification methods have lead to routine collection of $q=1,000 - 100,000+$ species in a single study.  
Very similar data are collected in many other application areas,
including studies of pathogens 
\citep{zhang2021large} and the microbiome \citep{zhao2021sigmoid}. In genomics, genetic variants are often represented as massive-dimensional binary response data \citep{lee2010sparse, davenport2018powerful}. There are numerous other examples.  Unfortunately, most statistical methods for multivariate binary response data with an unstructured dependence structure cannot be implemented in the huge $q$ regression setting even with small to moderate $p$.


As a canonical model that is easily interpretable and routinely used in related contexts, we focus on the multivariate probit model (MVP) \citep{ashford1970multi,cox1972analysis}. The latent Gaussian formulation of the model leads to straightforward interpretation of the regression coefficients and also provides flexibility in modeling the dependence structure of binary responses. A key computational challenge for fitting such models lies in the evaluation of multivariate Gaussian orthant probabilities \citep{bock1996high}.
\cite{chib1998analysis} developed a data augmentation scheme simulating the latent variables from truncated multivariate Gaussian distributions for maximum likelihood estimation and Bayesian inference. Unfortunately, generating samples from high dimensional truncated Gaussian distributions is computationally prohibitive and remains an active area of research; see \cite{pakman2014exact,botev2017normal} for developments. Moreover, Markov chain Monte Carlo (MCMC) algorithms based on simulating latent variables often suffer from poor mixing.  This was shown formally in imbalanced binary data models by 
\cite{johndrow2019mcmc}. 
 
 Due to practical challenges with MCMC, approximate posterior inference algorithms have become popular for latent Gaussian models.  A very successful example is the Integrated Nested Laplace Approximation (INLA) \citep{rue2009approximate}.
 However, current implementations of INLA available at \href{https://www.r-inla.org/home}{https://www.r-inla.org/home} do not allow for multivariate binary outcomes.  One major issue is that INLA only allows low-dimensional parameters integrating out the latent Gaussian process, but in our setting we have a high-dimensional unknown correlation matrix and high-dimensional regression coefficients. A popular alternative is to use variational Bayes 
  \citep{blei2017variational}; however, the resulting posterior approximations have no guarantees in terms of accurate uncertainty quantification, and indeed are well known to badly under-estimate posterior covariance in general.  
  

Recently, \cite{chen2018end} proposed a fast computational algorithm to approximate multivariate Gaussian orthant probabilities for deep MVP models. The method is appealing, especially in high dimensions, as it can be parallelized over the Monte Carlo samples and dimensions. However, there are two important issues. The method is very sensitive to the underlying correlation structure and as the dimension increases, exponentially more Monte Carlo samples are needed to produce the same level of accuracy. \cite{pichler2020new} build on this approximation technique but regularize the high-dimensional correlation matrix.  Their approach inherits the problem with approximation inaccuracy and they do not address uncertainty quantification in statistical inferences.

In this article, we develop a computationally scalable two-stage method, bigMVP, for inference for the MVP model.  
bigMVP is
motivated by scaling up to large $q$ for any choice of $n$ with $p$ small to moderate $p$, but can be applied broadly.
Our over-arching goals are to maintain accuracy in terms of estimation, uncertainty quantification and prediction under a limited computational budget. To achieve this, we focus on marginal inferences for model parameters in the MVP model; namely, regression coefficients for each outcome and correlation coefficients measuring pairwise dependence between outcomes. 
Two-stage inference methods have been popular for copula models in the frequentist literature \citep{shih1995inferences}. Several authors including \cite{joe2005asymptotic, ko2019model} studied asymptotic properties of the resulting estimators. \cite{joe2005asymptotic} studied asymptotic relative efficiency of the two-stage method compared to full maximum likelihood estimation and provided examples where the two-stage method achieves full efficiency. Adjusting for uncertainty in the first step is typically addressed by the two-stage variance estimator from \cite{murphy2002estimation}. Building on this line of work, \cite{ting2022fast} recently proposed a two stage maximum likelihood method for MVP models where the regression parameters for each outcome are estimated marginally in the first stage and correlation parameters are estimated in the second stage for each pair of outcomes by plugging in estimates obtained in the first stage. Although the method is conceptually related to the work presented here, simply plugging in maximum likelihood estimates of the regression coefficients for the second stage estimates of the correlation structure can incur large bias in finite samples; see Section \ref{sec:simulations} for detailed comparisons. This in turn results in significant under/over coverage of confidence intervals for the correlation coefficients, with the performance getting worse as dimension increases. 


Our initial motivation was to provide a rapid approximation to marginal posterior distributions in Bayesian MVP models with large $q$ and small to moderate $p$, while showing frequentist asymptotic guarantees to provide a methodology with broad appeal. Indeed, when a prior is available for the covariance having closed form marginals, then our approach can be used to obtain rapid approximations to marginal posterior distributions; we provide examples including the popular LKJ prior \citep{lewandowski2009generating}.  However, we found it too limiting to restrict attention to such cases; given the focus of inference is on the marginals, it is appealing to directly specify priors for these marginals.
When a joint prior does not exist that is consistent with these marginals, then we are targeting a generalized Bayes posterior for the marginals.  This simplifies prior elicitation and design of shrinkage priors for the marginal parameters.  We prove that the resulting procedure 
achieves optimal rates in estimating both the regression parameters and the correlation coefficients. We also develop a hierarchical extension in which the regression coefficients for the different outcomes are drawn from a common Gaussian distribution to borrow information. This is especially useful when many of the binary outcomes are observed rarely - a typical scenario in species sampling data and other motivating applications areas mentioned above. An approximation to the predictive distribution is provided in the supplementary materials.

\section{bigMVP}\label{sec:method}
Multivariate binary outcome data consist of a vector $y_i = (y_{i1},\ldots,y_{iq})^T$ for samples $i=1,\ldots,n$, with $y_{ij} \in \{0,1\}$.  In our motivating application, $y_{ij}=1$ if the $j$th species is present in the $i$th sample, with $y_{ij}=0$ otherwise.  We also have covariate information $x_i = (x_{i1},\ldots,x_{ip})^T$ for each sample.  In the species sampling application, $q$ is large while $p$ contains a small to moderate number of attributes of the sample; with this motivation, we focus on the problem of high-dimensional outcomes (large $q$) and moderate-dimensional covariates (moderate $p$).  

A challenge with multivariate binary data is how to define the dependence structure.  Two of the most common approaches are (1) define a generalized linear model (GLM) (e.g., logistic regression) for each outcome and then include common sample-specific latent factors in these models to induce dependence; and (2) define an underlying continuous variable model and induce dependence in the binary outcomes through dependence in these underlying variables.  Although strategy (1) is common in the ecology literature, there are disadvantages that motivate our focus on the multivariate probit (MVP) model and strategy (2).  A particularly concerning issue with GLM latent factor models is that the latent factor structure impacts the interpretation of the outcome-specific models, so that how we interpret covariate effects on the $j$th outcome depends on which other outcomes are included in the model.  

The MVP model does not have this issue, and is appealing in separating the marginal regression models for each outcome from the dependence structure between outcomes.  This is accomplished with an underlying Gaussian variable model in which $y_{ij} = \ind(z_{ij} > 0)$, $z_i = (z_{i1},\ldots,z_{iq})' \in \Re^q$ and 
these underlying variables have a simple multivariate normal linear model structure,  
\begin{eqnarray}
z_{ij} = x_i^T\beta_j + \epsilon_{ij},\quad \epsilon_i=(\epsilon_{i1},\ldots,\epsilon_{iq})^T \sim \Gauss(0, \Sigma), \label{eq:base}
\end{eqnarray}
where $\beta_j = (\beta_{j1},\ldots,\beta_{jp})^T$ are regression coefficients specific to outcome $j$, and 
$\Sigma$ is a positive definite correlation matrix defining the dependence structure across outcomes.  Marginally, a simple probit regression model is induced for each of the outcomes, with 
\begin{eqnarray}
\mbox{pr}( y_{ij}= 1 | x_i ) = 
\Phi( x_i^T \beta_j ), \label{eq:probit1}
\end{eqnarray}
where $\Phi(\cdot)$ is the cumulative distribution function of a standard Gaussian random variable.  Hence, we can interpret the $\beta_j$s based on (\ref{eq:probit1}), while the correlation coefficient $\sigma_{jk}$ in element $(j,k)$ of matrix $\Sigma$ controls the degree of dependence between $y_{ij}$ and $y_{ik}$.


We follow standard practice for multivariate probit models in assuming the data in the different samples, $y_i$ and $y_{i'}$, are independent, so that the likelihood under \eqref{eq:base} is 
\begin{equation}\label{eq:full_lik}
\prod_{i=1}^n \mbox{pr}(y_i \mid x_i, B, \Sigma) = \prod_{i=1}^n \mbox{pr}(z_i \in E_i\mid x_i, B, \Sigma),\quad z_i \sim \Gauss(B^\T x_i, \Sigma),
\end{equation}
where $B = (\beta_1, \ldots, \beta_q)$ is the $p \times q$ matrix of regression coefficients and $E_i = E_{i1} \times \ldots \times E_{ij}\times \ldots \times E_{iq}\subset \Re^q$ with $E_{ij} = \{z: z>0\} \mbox{ if } y_{ij} =1$ and $E_{ij} = \{z: z>0\} \mbox{ if } y_{ij} =0$. Working within a Bayesian framework, one endows the coefficient matrix with the prior $\Pi_B(\cdot)$ and the correlation matrix $\Sigma$ with the prior $\Pi_\Sigma(\cdot)$. Initially, we shall assume that regression vectors $\beta_j$ are independent {\it apriori}, so that $\Pi_B = \prod_{j=1}^q \Pi_j(\beta_j)$; extensions to hierarchical priors are considered in Section \ref{sec:hierarchical}. The full posterior distribution of the model parameters is obtained as
\begin{equation}\label{eq:full_posterior}
\Pi(B, \Sigma \mid y, X) \propto \Pi(\Sigma)\Big\{ \prod_{j=1}^q \Pi_j(\beta_j)\Big\}\Big\{ \prod_{i=1}^n \mbox{pr}(z_i \in E_i)\Big\},\quad z_i \sim \Gauss(B^\T x_i, \Sigma)
\end{equation}
Clearly, evaluating \eqref{eq:full_posterior} becomes highly expensive as $q$ increases because of the high dimensional integral involved in computing the likelihood for a given value of the parameters. This is true even when the focus is on inferences based on  marginal posterior distributions
 $\Pi_j(\beta_j \mid y, X)$ or $\Pi(\sigma_{jk} \mid y, X)$. 
 
 To avoid computing the marginal likelihood integrating out $\{ z_i \}$, one can instantiate the latent data in a data augmentation (DA) algorithm.  
  \cite{chib1998analysis} develop a DA Gibbs sampler for the MVP model, which relies on a parameter-expanded version of the model replacing the correlation matrix $\Sigma$ with a covariance matrix $\Sigma^*$ and coefficients $\beta_j$ with $\beta_j^*$.  One alternates between sampling the $z_i$ vectors from their truncated multivariate normal conditional posteriors, sampling $\beta_j^*$s from their multivariate normal conditional, and sampling $\Sigma^*$ from an inverse-Wishart conditional under an inverse-Wishart prior. 
  For moderate $q$, one can alternatively choose different notions of ``non-informative'' priors for the correlation matrix $\Sigma$; for example, the LKJ prior \citep{lewandowski2009generating} lets $\Pi(\Sigma) \propto |\Sigma|^{\nu - 1}, \, \nu> 0$. For $\nu = 1$, the prior distribution is uniform over the set of correlation matrices. Another possibility is  the marginally non-informative prior of \cite{huang2013simple} which lets $\Sigma\mid a_1, \ldots, a_q \sim \text{IW}(\nu + p -1), 2\nu \text{diag}(a_1, \ldots, a_q))$ and $a_j \sim \text{inverse-Gamma}(1/2, 1/A_j^2), \, j=1, \ldots, q$. When $\nu = 2$, this prior implies a Uniform[-1,1] prior on correlations. For both of these choices, the full conditional posterior of $\Sigma$ is inverse-Wishart and the sampler of \cite{chib1998analysis} can be trivially adapted.
  In a post-processing step, $\beta_j$ is set to $\beta_j^*$ divided by the square root of the $j$th diagonal element of $\Sigma^*$ and $\Sigma$ to the correlation matrix corresponding to covariance $\Sigma^*$ to obtain posterior samples for the MVP parameters.  While this approach can work well in low dimensions (small $q$), as $q$ increases three problems arise: (1) inefficiency of sampling from a high-dimensional truncated multivariate normal; (2) poor performance of the above priors
 for high-dimensional correlation matrices, and (3) worsening mixing, 
  particularly when some binary outcomes are imbalanced ($\mbox{pr}(y_{ij}=1) \approx 0$ or $\approx 1$) \citep{johndrow2019mcmc}.
  
Our focus is on obtaining a much more computationally efficient and scalable alternative for approximating marginal posteriors of $\beta_j$ and $\sigma_{jk}$; in practice inference based on such posteriors is almost always the focus.  For example, in our motivating ecology applications to studies of species biodiversity, the focus is on interpreting the covariate effects and correlations among species, and all such inferences can be based on marginal posteriors.
In the next subsection we introduce such an approximation $\Pi^*_j(\beta_j \mid y, X)$ for the marginal posterior of $\beta_j$, while in the subsequent subsection we propose an approach to approximate the posteriors of $\sigma_{jk}$.
\begin{rem}\label{rm:generalized_bayes}
Our posterior approximation does not require a joint prior for all the MVP parameters, but only a prior for the marginals.  It is convenient to directly specify this marginal prior to simplify prior elicitation and design of default shrinkage approaches.  When the marginal prior is consistent with a coherent joint prior, our approach targets a Bayesian posterior but otherwise the target is generalized Bayes; we will include illustrations of both cases.
\end{rem}

In what follows, we write the density function of the standard Gaussian distribution as $\phi(\cdot)$. For a two dimensional covariance matrix $Q$, let $\phi_Q(\cdot)$ and $\Phi_{Q}(\cdot)$ denote the density and distribution function of a bivariate Gaussian distribution with mean $(0, 0)^\T$; i.e. $\phi_Q (x) = (2\pi)^{-1/2} |Q|^{-1/2} \exp\{-x^\T Q^{-1}x/2\}$ and $\Phi_Q(x) = \int_{-\infty}^x \phi_{Q}(u) du$ for $x, u \in \mathbb{R}^2$.  For two vectors $x, y$ we write $x \odot y$ for their Hadamard product.

\subsection{First-stage inference}\label{sec:first_stage}
We start in the first stage by approximating the marginal posterior distributions of the regression coefficients $\beta_j$ obtained from the joint posterior defined in \eqref{eq:full_posterior}.  Our approximation $\Pi^*(\beta \mid y, X)$ to the marginal for $\beta_j$ under \eqref{eq:full_posterior} is obtained by using a purposely misspecified likelihood. 
In particular, we replace the likelihood in \eqref{eq:full_lik} by the product of marginal likelihoods $\prod_{i=1}^n \prod_{j = 1}^q \mbox{pr}(z_{ij}\in E_{ij})$; that is, we set $\Sigma = \mathrm{I}$. 
This is equivalent to fitting univariate probit models to each outcome and allows for parallelization over the $q$ outcomes. That is, in the first stage, we fit the model $y_{ij} \sim \text{Ber}\{\Phi(x_i^\T \beta_j)\}$ for each $j = 1, \ldots, q$. Letting $\ell_j(\beta_j)$ to be the log-likelihood of the $j$th binary outcome under this misspecified model, we have
 \begin{equation}\label{eq:marginal_likelihood_beta_j}
\ell_j(\beta_j) = \sum_{i=1}^n  y_{ij}\log\{\Phi(x_i^T\beta_{j})\}+(1-y_{ij})\log\{1-\Phi(x_i^T\beta_j)\}.
\end{equation}
Set $\bar{\ell}_{nj}(\beta_j) = -\ell_j(\beta_j)/n$. Combining this with the prior $\prod_{j=1}^q \Pi_j(\beta_j)$ as in \eqref{eq:full_posterior}, we construct the approximate marginal $\Pi_j^*(\beta_j \mid y, X) = \exp\{-n\lnjbar{\beta_j}\}\Pi_j(\beta_j)/d_n$, where $d_n = \int \exp\{-n\lnjbar{\beta_j}\}\Pi_j(\beta_j) d\beta_j$.

The distribution $\Pi_j^*(\beta_j \mid y, X)$ does not have a closed form expression. However, since $\beta_j$ is moderate dimensional by assumption, we set
\begin{equation}\label{eq:first_stage}
\Pi^*_j(\beta_j \mid y, X) \approx \Gauss(\hat{\beta_j}, H_j); \,\, 
 \hat{\beta_j} = \argmax_{\beta_j}\left\lbrace -n \bar{\ell}_{nj}(\beta_j) + \log \Pi_j(\beta_j)\right\rbrace,
\end{equation} 
using Laplace's method \citep{tierney1986accurate} where $H_j$ is the corresponding inverse Hessian. In Section \ref{sec:first_stage_anal}, we show that in the limit $n\to \infty$, the marginal likelihood $d_n$ can be suitably approximated by Laplaces's method \citep{tierney1986accurate} and that the marginal posterior itself approaches a Gaussian distribution after suitable scaling and centering.
The computational complexity of obtaining $\hat{\beta_j}$ is $\mathcal{O}(M_j np^2)$, where $M_j$ is the number of iterations until convergence of the Newton-Raphson algorithm. Hence, the complexity is linear in the number of outcomes $q$.


\subsection{Second stage inference}\label{sec:second_stage}
Having obtained $\Pi_j^*(\beta_j \mid y, X)$ we move on to the more challenging problem of inference on the correlation matrix $\Sigma$. We focus on approximating the marginal posterior distribution $\Pi(\sigma_{jk}\mid y, X)$ of the correlation between pairs of outcomes obtained under joint posterior \eqref{eq:full_posterior}. This is accomplished by considering a bivariate probit model between the pairs $(j,k)$. That is, in the second stage, we consider the likelihood $\prod_{i=1}^n \prod_{j<k}\text{pr}(z_{ij} \in E_{ij}, z_{ik} \in E_{ik})$. To approximate the posterior distribution of $\sigma_{jk}$, we combine this likelihood with the marginal prior for 
$\sigma_{jk}$ derived from the joint prior $\Pi(\Sigma)$; refer to Remark \ref{rm:marginal_prior} for examples. Alternatively, one can focus on more convenient product 
 marginal priors $\prod_{j<k} \Pi_{jk}(\sigma_{jk})$. 
 

Our construction of this approximate marginal depends on the approximations obtained in the first stage, $\Pi_j^*(\beta_j \mid y, X)$ and $\Pi_j^*(\beta_k \mid y, X)$. While frequentist analogues of two stage estimation generally plug in maximum likelihood estimates of $\beta_j$ and $\beta_k$ obtained in the first stage \citep{joe2005asymptotic, yi2011robust, ko2019model,ting2022fast}, we include $\Pi_j^*(\beta_j \mid y, X)$ and $\Pi_k^*(\beta_k \mid y, X)$, the approximate posterior distributions of the regression coefficients, in the form of updated prior distributions on $(\beta_j, \beta_k)$. More precisely, in the second stage, the updated prior distribution $\Pi(\beta_j, \beta_k)$ is set as $\Pi_j^*(\beta_j\mid y, X) \Pi_k^*(\beta_k\mid y, X)$.
Let $\Pi_{jk}^*({\sigma_{jk}\mid y, X})$ be the approximate marginal posterior distribution of $\sigma_{jk}$ and $\Pi_{jk}(\cdot)$ be the marginal prior on $\sigma_{jk}$ derived from the joint prior $\Pi(\Sigma)$ or the product marginal prior $\prod_{j<k} \Pi_{jk}(\sigma_{jk})$. We have,
\begin{align*}
\Pi_{jk}^* (\sigma_{jk} \mid y, X)& = \Pi_{jk}^*(\sigma_{jk} \mid y^{(j)}, y^{(k)}, X) = \dfrac{\prod_{i=1}^n \mbox{pr}(z_{ij} \in E_{ij}, z_{ik} \in E_{ik} \mid \sigma_{jk}) \Pi_{jk}(\sigma_{jk})}{p(y^{(j)}, y^{(k)})} \\
 = \dfrac{\Pi_{jk}(\sigma_{jk})}{p(y^{(j)}, y^{(k)})}  &\prod_{i=1}^n \int \mbox{pr}(z_{ij} \in E_{ij}, z_{ik} \in E_{ik} \mid \sigma_{jk}, \beta_j, \beta_k) \Pi(\beta_j, \beta_k) d\beta_j d\beta_k \\
 \approx \dfrac{\Pi_{jk}(\sigma_{jk})}{p(y^{(j)}, y^{(k)})}  &\prod_{i=1}^n \int \mbox{pr}(z_{ij} \in E_{ij}, z_{ik} \in E_{ik} \mid \sigma_{jk}, \beta_j, \beta_k) \Pi_j^*(\beta_j \mid y, X) \Pi_k^*(\beta_k \mid y, X) d\beta_j d\beta_k\\
 = \dfrac{\Pi_{jk}(\sigma_{jk})}{C_{jk}}   \prod_{i=1}^n & \int \left\lbrace\int_{(E_{ij}, E_{ik})} p(z_{ij}, z_{ik})dz_i\right\rbrace \Pi_j^*(\beta_j \mid y, X) \Pi_k^*(\beta_k \mid y, X) d\beta_j d\beta_k,
\end{align*}
where $C_{jk} = p(y^{(j)}, y^{(k)})$ and $p(z_{ij}, z_{ik})$ is the pdf of a bivariate Gaussian distribution with mean vector $(x_i^\T \beta_j, x_i^\T \beta_k)$, variance $1$ and correlation coefficient $\sigma_{jk}$.
Using Fubini's theorem, we  interchange the order of integration in the above display to marginalize out $\beta_j$ and $\beta_k$ to obtain an updated distribution of the latent variables $(z_{ij}, z_{ik})$. 
Hence, we have
\begin{equation}\label{eq:sigma_approx}
\Pi_{jk}^*(\sigma_{jk} \mid y, X) =  \dfrac{1}{C_{jk}} \prod_{i=1}^n \mbox{pr}(\tilde{z}_{ij} \in E_{ij}, \tilde{z}_{ik} \in E_{ik}) \Pi_{jk}(\sigma_{jk}),
\end{equation}
where $(\tilde{z}_{ij}, \tilde{z}_{ik}) \sim \Gauss(\tilde{\mu}^i_{jk}, \tilde{\Sigma}^i_{jk})$, $\tilde{\mu}^i_{jk} = (x_i^\T \hat{\beta_j}, x_i^\T \hat{\beta_k})$, and $\tilde{\Sigma}^i_{jk} = \{(1 + x_i^\T H_j x_i, \sigma_{jk})^\T; (\sigma_{jk}, 1 + x_i^\T H_k x_i)^\T\}$. For convenience, define the transformed data $r_{ij}=2y_{ij}-1$, $r_{ik}=2y_{ik}-1$ and $r_i = (r_{ij}, r_{ik})^\T$.  Also, define the sign-transformed correlation matrix $\overline{\Sigma}^i_{jk}  = \{(1 + x_i^\T H_j x_i,\,  r_{ij}r_{ik}\sigma_{jk})^\T;  (r_{ij}r_{ik}\sigma_{jk},\, 1 + x_i^\T H_k x_i)^\T\}$. Then the log-likelihood $\ell_{i}^{jk} (\sigma_{jk})$ of the $i$th data point $(y_{ij}, y_{ik})^\T$ under the updated bivariate probit model can be written as $\ell_{i}^{jk} (\sigma_{jk}) = \log \{ \Phi_{\overline{\Sigma}^i_{jk}}(r_i \odot \tilde{\mu}_{jk})\}$ and hence $\Pi_{jk}^*(\sigma_{jk}\mid y, X) \propto \{\prod_{i=1}^n e^{\ell_{i}^{jk}(\sigma_{jk})}\} \Pi_{jk}(\sigma_{jk})$. 
In Section \ref{sec:snd_stage_anal} we show that $\Pi^*_{jk}(\sigma_{jk} \mid y, X)$ can be approximated by a univariate Gaussian distribution with appropriate mean and variance. We obtain the mean $\hat{\sigma}_{jk}$ and variance $s^2_{jk}$ of $\Pi_{jk}^*(\sigma_{jk} \mid y, X)$ using Gauss-Legendre quadrature and set $\Pi^*_{jk}(\sigma_{jk} \mid y, X) \equiv \Gauss(\hat{\sigma}_{jk}, s^2_{jk})$. For each pair of outcomes $(j,k)$, the computational complexity to obtain the mean and variance of $\Pi_{jk}^*(\sigma_{jk} \mid y, X)$ using $m$ quadrature points is $\mathcal{O}(2mn)$ which implies $\mathcal{O}(q^2mn)$ complexity for all pairs. A concise summary of the two-stages of bigMVP is provided in Algorithm \ref{algo:bigMVP_algo}.

The combined complexity of the two stages of the proposed method scale as $\mathcal{O}(nq^2)$ to obtain the approximate marginals of the regression coefficients and the correlation coefficients. Data augmented MCMC samplers need to sample the latent variable $z_i$ for each data point at every iteration, 
having a best case complexity of $\mathcal{O}(nq^2)$ \citep{pakman2014exact}. This is highly expensive when one has to run the sampler for thousands of iterations.  Unlike MCMC, we need to do the calculations only once and computations can be trivially parallelized.
\begin{rem}\label{rm:marginal_prior}
The marginal prior $\Pi_{jk}(\sigma_{jk})$ on $\sigma_{jk}$ under $\Pi(\Sigma)$ is available in closed form for several popular classes of priors for covariance/correlation matrices. For example, when $\Sigma \sim \text{inverse-Wishart}(\nu, \Lambda)$ and $\Lambda$ is diagonal, then $\Pi_{jk}(\sigma_{jk})\propto (1-\sigma_{jk}^2)^{(\nu-2q)/2 - 1}$; when $\Pi(\Sigma) \propto |\Sigma|^{\nu - 1}$ \citep{lewandowski2009generating}, then $\Pi_{jk}\{(\sigma_{jk} + 1)/2\} \sim \text{Beta}(\nu + q/2 - 1, \nu + q/2 - 1)$; for the prior in \cite{huang2013simple}, one obtains $\Pi_{jk}(\sigma_{jk})\propto (1-\sigma_{jk})^{\nu/2 - 1}$ which induces Uniform[-1,1] marginals over the correlations for $\nu = 2$. Assuming a $\Sigma \sim \text{LKJ}(1)$ prior, in Figure \ref{fig:prior_performance} we compare MCMC-based and bigMVP approximations to the marginal posterior of $\sigma_{12}$.
Alternatively, product marginal priors $\prod_{j<k}\Pi_{jk}(\sigma_{jk})$ simplify prior elicitation and provide accurate results as illustrated by our theoretical investigations in Section \ref{sec:theory} and our numerical results in Section \ref{sec:simulations}.
We include additional experimental results in Section \ref{sec:simulations} of the supplementary materials to study the effect of replacing a joint prior by marginal priors in more detail. In addition to the priors mentioned above, we consider a low-rank favouring prior in Section \ref{sec:hierarchical} tailored to handle situations when many of the binary outcomes are rare.
\end{rem}
\begin{figure}
\centering
\begin{subfigure}{.45\textwidth}
  \centering
  \includegraphics[height=4.5cm, width = 6cm]{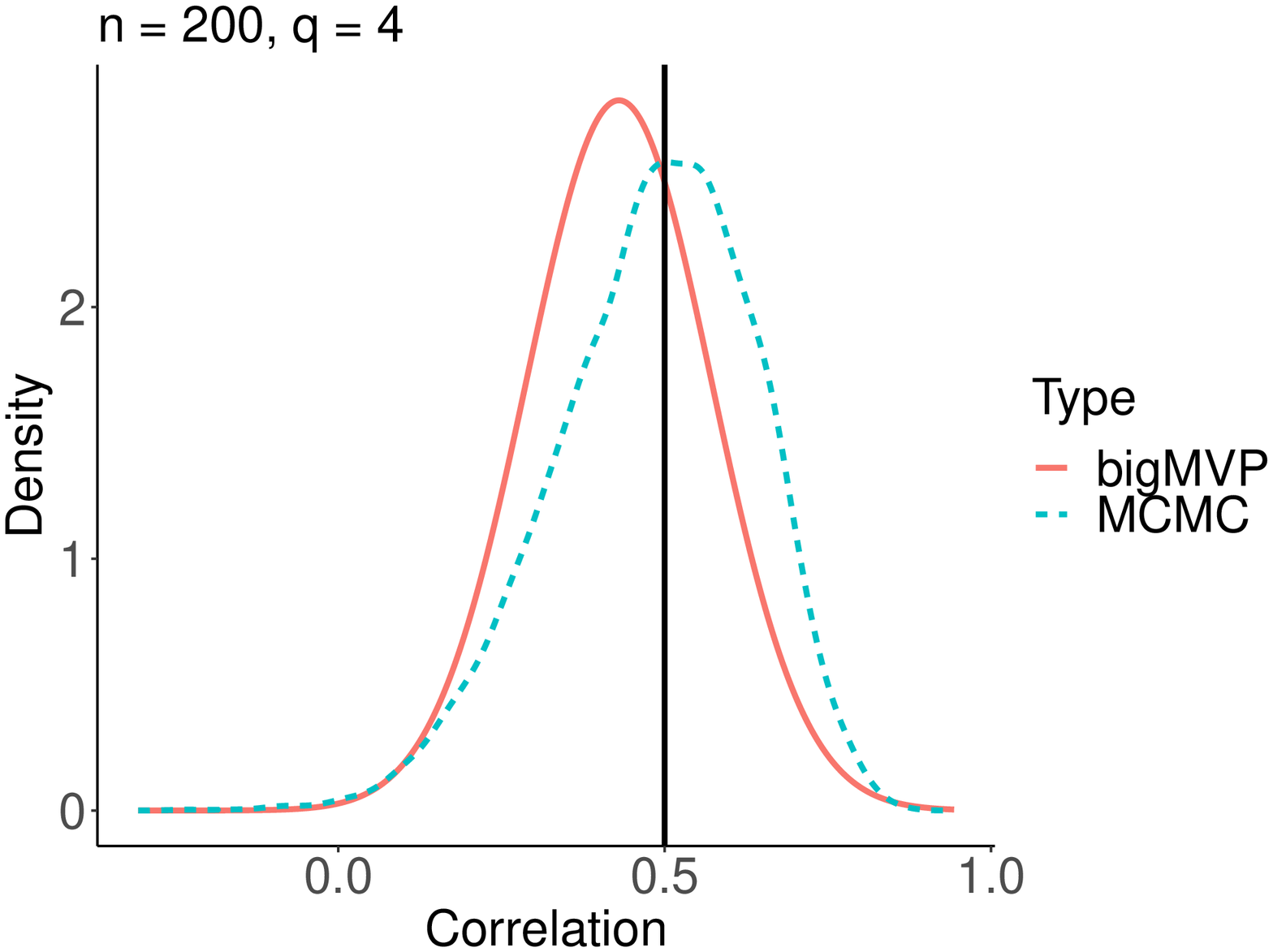}
\end{subfigure}%
\begin{subfigure}{.45\textwidth}
  \centering
  \includegraphics[height=4.5cm, width = 6cm]{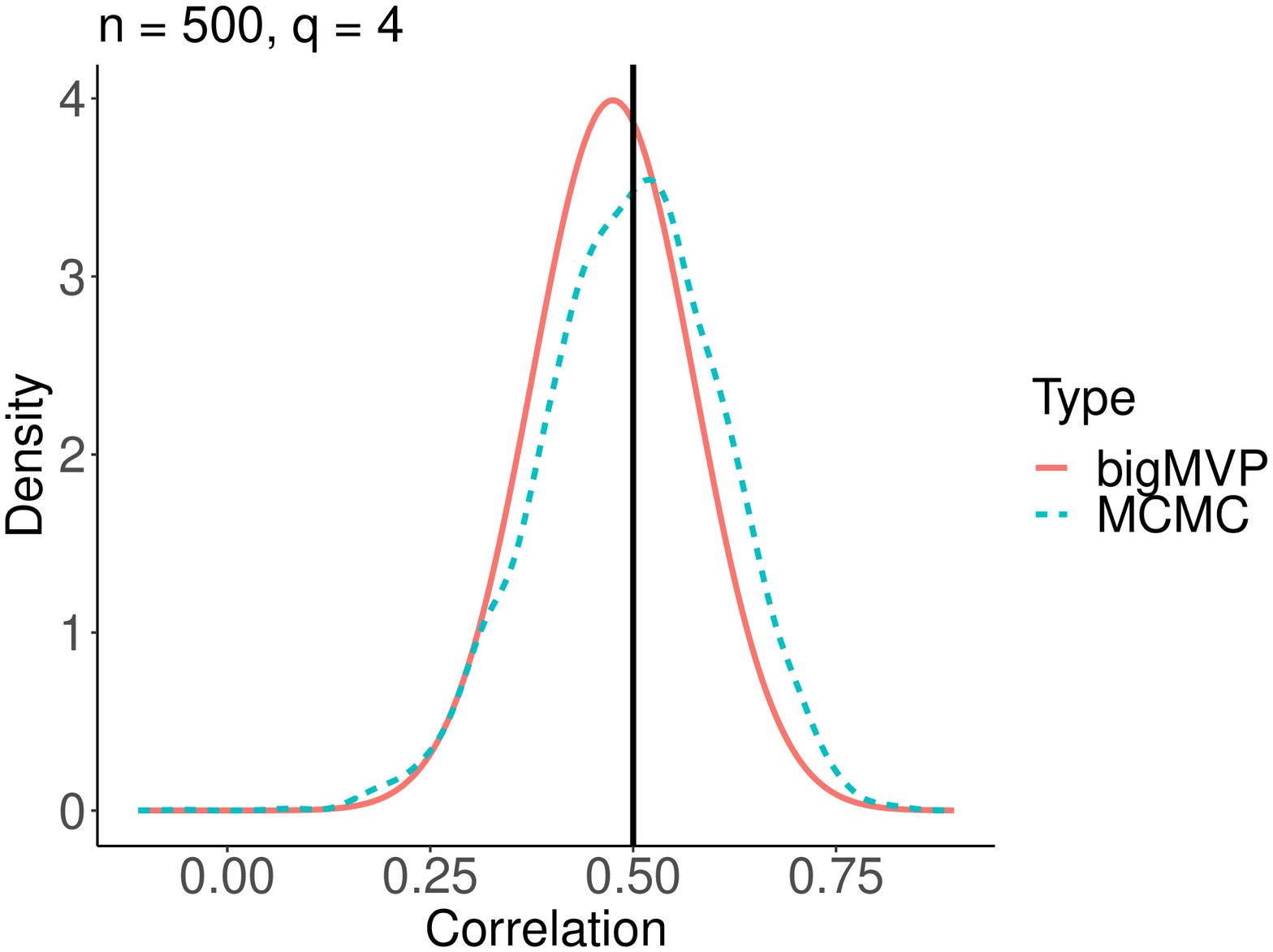}
\end{subfigure}
\caption{Marginal posterior distributions from full posterior analysis of \eqref{eq:full_posterior} when $\Pi(\Sigma)$ is LKJ(1) implemented with MCMC (blue dotted line) versus approximate marginal posteriors obtained by bigMVP (black solid line) with the corresponding marginal prior $\Pi(\sigma_{jk}) \propto \text{Beta}(q/2, q/2)$. The true correlation structure here is $\Sigma^* = (1-\rho^*)\mathrm{I}_q + \rho^* \mathbf{11^\T}$ with $\rho^* = 0.5$. A vertical black line is added at $\rho^* = 0.5$.}
\label{fig:prior_performance}
\end{figure}

\begin{algorithm}
\caption{bigMVP for approximating marginal posteriors in MVP models.}
\textbf{Data:} $y^{n \times q}$, $X^{n \times p}$.

\textbf{Priors:} $\Pi_j(\beta_j)$ for $j =1, \ldots, q$,  $\Pi_{jk}(\sigma_{jk}),$ for $j<k = 1, \ldots, q$.

\underline{\textbf{Stage 1:}} \For{$j = 1:q$}{Set $\Pi_j^*(\beta_j\mid y, X) \propto \exp\{-n\bar{\ell}_{nj}(\beta_j)\} \Pi_j(\beta_j)$ where $-n\bar{\ell}_{nj}(\beta_j) = \sum_{i=1}^n y_{ij} \log \Phi (x_i\beta_j) + (1 - y_{ij})\log \{1 - \Phi(x_i^\T \beta_j)\}$
and approximate this posterior by $\Gauss(\hat{\beta}_j, H_j)$, where $\hat{\beta_j} = \argmax_{\beta_j}\left\lbrace -n \bar{\ell}_{nj}(\beta_j) + \log \Pi_j(\beta_j)\right\rbrace$ and $H_j$ is the corresponding inverse Hessian. 
}

\underline{\textbf{Stage 2:}} \For{$j<k = 1, \ldots, q$}{ Set $\Pi_{jk}^*(\sigma_{jk}\mid y, X) \propto \{\prod_{i=1}^n e^{\ell_i^{jk}(\sigma_{jk})}\} \Pi_{jk}(\sigma_{jk})$, where $\ell_{i}^{jk} (\sigma_{jk}) = \log \{ \Phi_{\overline{\Sigma}_{jk}}(r_i \odot \tilde{\mu}_{jk})\}$. This is approximated by $\Gauss(\hat{\sigma}_{jk}^2, s_{jk}^2)$ where $\hat{\sigma}_{jk}^2, s_{jk}^2$ are obtained by numerical integration.}
\textbf{Output:} Marginal posteriors $\Pi_j^*(\beta_j \mid y, X)$ for regression coefficients where $j = 1, \ldots, q$ and $\Pi_{jk}^*(\sigma_{jk}\mid y, X)$ for correlation coefficients where $j<k = 1, \ldots, q$.


\label{algo:bigMVP_algo}
\end{algorithm}

\subsection{Extension to hierarchical setting}\label{sec:hierarchical}
In this section we extend the MVP model to a hierarchical setting, with the goal being to borrow information across the different outcomes to obtain more accurate estimates of the  outcome-specific regression coefficients and cross outcome correlations.  In our motivating application, this is particularly important to enable accurate inferences on the coefficients for rare species that are only observed a small number of times in the entire dataset.
 For example, in the breeding bird survey data in \cite{lindstrom2015large}, out of the 141 bird species observed at 599 locations, 19 bird species were observed in less than 12 locations. Estimating probit regression coefficients for these species without borrowing of information will invariably lead to very high standard errors. To accommodate these situations, we assume a hierarchical structure for both the regression coefficients $\beta_j$ and correlation coefficients $\sigma_{jk}$. Specifically, letting $\gamma_{jk} = 0.5 \log\{(1+\sigma_{jk})/(1-\sigma_{jk})\}$ correspond to the Fisher transformation of $\sigma_{jk}$, we consider the hierarchy 
\begin{align}\label{eq:jsdm_hierarchy}
y_i \mid B, \Sigma, X, \eta, \Omega & \sim \text{MVP}(B^\T x_i, \Sigma), \,\, i = 1, \ldots, n, \nonumber\\
\beta_j\mid \eta,\Omega \, & \overset{iid}{\sim} \,\Gauss(\eta, \Omega),  \nonumber \\
 \gamma_{jk} \mid \omega & \overset{iid}{\sim} \Gauss (0, \omega^2), \,\, j<k = 1,\ldots, q, \nonumber\\
(\eta, \Omega)\mid \eta_0, \nu_0,  \delta_0, \Lambda_0 & \sim \text{NIW}(\eta_0,\nu_0, \delta_0, \Lambda_0), \, \omega^2\mid a_\omega, b_\omega \sim \text{IW}(a_\omega, b_\omega)
\end{align}
where $\text{NIW}(\eta_0,\nu_0, \delta_0, \Lambda_0)$ represents a Normal-Inverse Wishart distribution. Shrinking $\gamma_{jk}$'s towards 0 equivalently shrinks the correlations $\sigma_{jk}$ towards 0. 

Although this hierarchical model has the advantage of reducing mean square errors in estimation through borrowing of information, efficient computation is more challenging due to the dependence between the $\beta_j$s and $\sigma_{jk}$s for different outcomes, which is induced through shared dependence on $(\eta, \Omega)$ and $\omega$. A natural way to maintain computational scalability is to consider empirical Bayes \citep{morris1983parametric} estimates $(\hat{\eta}, \hat{\Omega})$, $\hat{\omega}$ of $(\eta, \Omega)$ and $\omega$, respectively,  wherein one marginalizes over $\beta_j$s and $\gamma_{jk}$s under the hierarchy \eqref{eq:jsdm_hierarchy}. For a fixed $\Sigma$, after marginalization, the distribution of the latent variables $z_i$ is $\Gauss(\Gamma_1^\T x_i, \Gamma_2^i)$, where $\Gamma_1^{p \times q} = (\eta, \eta, \ldots, \eta)$ and $\Gamma_2^i = \Sigma + \mbox{diag}(x_i^\T \Omega x_i)$. Unfortunately, as $\Omega$ is now involved in the dependence structure of the latent $z_i$s, it becomes necessary in estimating $\Omega$ to evaluate multivariate Gaussian orthant probabilities or simulate from truncated multivariate Gaussian random variables in conducting data augmentation.  Hence, in estimating $\Omega$, we encounter the same computational bottlenecks as discussed previously.

We address this issue by designing a fast approximate sampler leveraging on the fact that conditional on $(\eta, \Omega)$ and $\omega$, Algorithm \ref{algo:bigMVP_algo} can be employed with minor modifications adjusting for the priors in \eqref{eq:jsdm_hierarchy}. The sampler updates the hyperparameters $(\eta, \Omega)$ and $\omega$ in two conditional moves akin to standard Gibbs samplers. However, when updating the $(\eta, \Omega)$, we consider the approximate likelihood considered in Section \ref{sec:first_stage}. Under this approximate likelihood, the joint posterior distribution of $\beta_j$'s and $(\eta, \Omega)$ is
\begin{equation}\label{eq:conditional_posteriors}
    \Pi(\beta_1, \ldots, \beta_q, \eta, \Omega\mid y, X) \propto \prod_{i=1}^n \prod_{j=1}^q \mbox{pr}(z_{ij} \in E_{ij}\mid \beta_j) \prod_{j=1}^q \Pi(\beta_j \mid \eta, \Omega) \Pi(\eta, \Omega).
\end{equation}
Sampling from the joint posterior can be easily implemented alternating between the full conditionals $\Pi(\beta_j \mid \eta, \Omega, y, X)$, $\Pi(\eta \mid \beta_1, \ldots, \beta_q, \Omega, y, X)$ and $\Pi(\Omega \mid \beta_1, \ldots, \beta_q, \eta, y, X)$ for $j = 1, \ldots, q$. Here, we approximate $\Pi(\beta_j \mid \eta, \Omega, y, X)$ by their corresponding Laplace approximations. Next, conditional on the regression coefficients, we update the correlations using the second stage approximate likelihood as in Section \ref{sec:second_stage}.
The joint posterior of the $\sigma_{jk}$s and $\omega$ is then 
\begin{equation}\label{eq:conditional_posteriors_correlations}
    \Pi(\sigma_{11}, \ldots, \sigma_{q-1,q}, \omega\mid y, X) \propto  \prod_{j<k=1}^q \prod_{i=1}^n  \mbox{pr}(\tilde{z}_{ij} \in E_{ij}, \tilde{z}_{ik} \in E_{ik}\mid \sigma_{jk}) \prod_{j<k=1}^q \Pi(\gamma_{jk}\mid \omega^2) \Pi(\omega^2),
\end{equation}
where the regression coefficients $(\beta_j, \beta_k)$ have been marginalized out following \eqref{eq:sigma_approx}. Conditional on $\omega$, an approximation to $\Pi(\sigma_{jk} \mid y, X, \omega)$ is obtained as $\Gauss(\hat{\sigma}_{jk}, s_{jk}^2)$. We then draw $\sigma_{jk} \sim \Gauss(\hat{\sigma}_{jk}, s_{jk}^2)$ and set $\gamma_{jk} = 0.5 \log\{(1+\sigma_{jk})/(1-\sigma_{jk})\}$. These samples are then used to update $\omega^2$.

The details are given in Algorithm \ref{algo:empirical_sampler} of the supplementary materials, which we call the two-stage conditional sampler. Conditional on $(\eta, \Omega)$, sampling the $\beta_j$s requires the same complexity as mentioned in Section \ref{sec:first_stage} whereas sampling $(\eta, \Omega)$ can be done in $\mathcal{O}(qp^3)$ complexity. Sampling the other hyperparameter $\omega$ and the correlations has $\mathcal{O}(q^2)$ complexity.
In our experience, the sampler mixes really fast with approximately 10 effective samples per second for $(n, p, q) = (200, 5, 100)$ when run on a 64 bit Intel i7-8700K CPU @3.7 GHz processor. Having obtained samples from the posterior distributions of $(\eta, \Omega)$, we simply plug-in the average $(\hat{\eta}, \hat{\Omega})$ and $\hat{\omega}$ of these quantities so that, conditional on the plug-in estimates, bigMVP can be implemented in a straightforward manner. 

This conditional sampler is different from implementing MCMC for the entire model. We use the special dependence structure of hierarchy \eqref{eq:jsdm_hierarchy}.  As the prior on $\beta_j$ is unrelated to 
$\Sigma$, we base inference on  $(\eta, \Omega)$ on the likelihood contribution relevant to the $\beta_j$s using a product of 
 independent univariate probit likelihoods as in our previous first stage inferences.  Implementing an ``exact'' Gibbs sampler is massively more computationally expensive in alternating from simulating the latent $z_i$s from high-dimensional truncated Gaussians and drawing from the full conditional distributions of $\beta_j$, $\Sigma$ and $(\eta, \Omega)$.  In addition, while MCMC for the full model requires a joint prior specification, our bigMVP method is more general as mentioned in Remark \ref{rm:generalized_bayes}.



\section{Theory}\label{sec:theory}
Suppose $\theta = (B, \Sigma) \in \Theta$. In this section, we assume data are generated from the MVP model with true parameters $\theta^* = (B^*, \Sigma^*)$ and provide asymptotic results for $\Pi_j^*(\beta_j \mid y, X)$ and $\Pi_{jk}^*(\sigma_{jk} \mid y, X)$ assuming fixed number of covariates $p$. Our results hold irrespective of whether one allows the number of outcomes $q$ to grow with the sample size or not. In particular, we are interested in two key aspects of these approximations: (1) concentration - whether the posteriors converge to a point mass at the true parameter value and (2) shape - whether the posteriors are asymptotically normal. We recognize the likelihoods in both stages of our proposed method as versions of composite likelihoods \citep{lindsay1988composite} and leverage results from \cite{miller2021asymptotic} to establish these properties of the approximate marginal likelihoods. Asymptotic validity of the Laplace approximations of $\Pi_j^*(\beta_j \mid y, X)$ is also proved paralleling the classical results of \cite{geisser1990validity} for posteriors obtained without likelihood misspecification.
We assume the data $Y=(y_1,\cdots,y_n)^T$ are generated by an MVP model under true parameters $\theta^*=(\Sigma^*,B^*)$, where $B^*\in \RR^{p\times q}$ and $\Sigma^* \in \mathcal{S}^q$, the cone of positive definite matrices.  Given two densities $p$ and $q$ with respect to the Lebesgue measure, the total variation distance is defined as $\norm{p - q}_{1} = \int |p(u)-q(u)|du$. For two positive sequences $a_n$ and $b_n$, we write $a_n \sim b_n$ to denote that $a_n/b_n \to 1$ as $n \to \infty$. We use $\norm{v}$ for the Euclidean norm of a real valued vector $v$.
\subsection{Assumptions}
The following assumptions are made on the parameter space, design matrix and prior distributions.
\begin{assumption}[Regularity] \label{as:boundedparam}
	Let $B \in\Xi \subset \RR^{p\times q}$ where  $\Xi$ is an open bounded subset of $\RR^{p\times q}$. 
\end{assumption}
\begin{assumption}[Regularity]\label{as:corr}
	There exists a real interval $(M_1,M_2)\subset (-1,1)$ such that for every $(j,k)$, $\sigma^*_{jk}\in (M_1,M_2)$ for $1\leq j,k\leq q$.
\end{assumption}

\begin{assumption}[Prior support] \label{as:prior}
	Let $\Pi_j$ be the prior probability density of $\beta_j$ with respect to the Lebesgue measure for $1\leq j \leq q$. Then for each $j$, $\Pi_j(\cdot)$ is continuous at $\beta^*_j$ and there exists an $\epsilon>0$ such that $\Pi_j(\beta^*_j)>\epsilon>0$ uniformly in $j$. Let $\Pi_{jk}(\cdot): (-1,1) \rightarrow \mathbb{R}$ be the prior probability density of $\sigma_{jk}$ with respect to Lebesgue measure. Then $\Pi_{jk}(\cdot)$ is continuous at $\sigma^*_{jk}$ and there exists an $\epsilon>0$ such that $\Pi_{jk}(\sigma^*_{jk})>\epsilon>0$ for every $1\leq j,k\leq q$.
\end{assumption}

\begin{assumption}[Design matrix] \label{as:designmat}
    The Euclidean norms of the rows $x_{i}$ of $X$, i.e. $\|x_i\|$, are uniformly bounded in $i$ and $\lim _{n \rightarrow \infty} n^{-1} \Sigma_{i=1}^{n}$ $x_{i} x_{i}^\T$ is a finite nonsingular matrix. Furthermore the empirical distribution  of $\left\{x_{i}\right\}$ converges to a distribution function. 
\end{assumption}

Assumption \ref{as:boundedparam} together with Assumption \ref{as:designmat} imply there exists a real number $M\in\RR$ such that $| x_i^T\beta_j | \leq M$ almost surely for every $1\leq i\leq n$ and $1\leq j\leq q$. We note here that no assumption is made on the specific relation between $n$ and $q$ which is natural given the focus is on marginal posterior distributions.
Assumptions \ref{as:boundedparam} and \ref{as:corr} are standard in classical asymptotic theory of maximum likelihood estimation for parametric models \citep{van2000asymptotic}. The prior support assumption, i.e. Assumption \ref{as:prior}, ensures positive prior probability around true parameter values. In the special case where the priors $\Pi_{jk}(\sigma_{jk})$ are induced from a joint prior $\Pi(\Sigma)$ on $\Sigma$, these marginals need to satisfy the prior support condition.
Assumption \ref{as:designmat} on the design matrix is also used in \citet[Theorem 9.2.2]{amemiya1985advanced} in establishing asymptotic normality of maximum likelihood estimators for univariate probit regression. 


\subsection{First-stage Analysis}
\label{sec:first_stage_anal}

We analyze the posterior distribution of $\beta_j$ asymptotically, for any $1\leq j\leq q$. The misspecified likelihood used in Section \ref{sec:first_stage}, where we replace $\Sigma$ by the identity matrix, can be viewed as a product of marginal likelihoods and thus falls under the umbrella of composite likelihoods \citep{lindsay1988composite}. Properties of estimators derived by maximizing composite likelihoods, such as consistency and asymptotic normality, are well established; see \cite{varin2011overview} for a survey. \cite{miller2021asymptotic} provide sufficient conditions under which posterior distributions obtained by combining a composite likelihood derived from a correct model combined with a suitable prior concentrate at the true parameter value and exhibit asymptotic normality.

In Theorem \ref{thm.fiststagebvm} we show that $\Pi_j^*(\beta_j\mid y, X)$ is asymptotically normal centered at the maximum composite likelihood estimator and the Laplace approximation we employ is valid. Our proof relies on the observation that the marginal distribution of $y^{(j)}$ is the same under the joint model \eqref{eq:full_lik} and the independence model obtained by plugging in $\Sigma = \mathrm{I}$ in \eqref{eq:full_lik}; under both models $y^{(j)}$ follows a univariate probit model conditional on $\beta_j$. The proof guarantees that the maximum composite likelihood estimator converges to $\beta_j^*$.
This also implies a parametric contraction rate $O_{P_{\theta^*}}(n^{-1/2})$ for $\Pi^*_j(\beta_j \mid y, X)$, i.e.,
$
\Pi_{j}^{*}\left( \|\beta_{j}-\beta^*_{j}\|>\frac{M_n}{\sqrt{n}} \mid y, X\right) \rightarrow 0
$
for every $M_n\rightarrow +\infty$. 

Recall the definition of $\bar{\ell}_{nj}(\beta_j)$ from Section \ref{sec:first_stage} . We write  $\Phi_{ij}, \Phi^*_{ij},\phi_{ij}, \phi^*_{ij}$ as shorthand for $\Phi(x_i^T\beta_{j}), \Phi(x_i^T\beta^*_{j}), \phi(x_i^T\beta_j), \phi(x_i^T\beta^*_j)$, respectively. Following the analysis in Theorem 9.2.2 of \cite{amemiya1985advanced}, $\lnjbar{\cdot} \rightarrow \ljbar{\cdot}$ pointwise in $P_{\theta^*}$-probability, where
\begin{equation}\label{eq:p1def}
\ljbar{\beta_{j}} = -\lim_{n \rightarrow \infty} n^{-1} \sum_{i=1}^{n} \Phi^*_{ij} \log \Phi_{ij}-\lim_{n \rightarrow \infty} n^{-1} \sum_{i=1}^{n}\left(1-\Phi^*_{ij}\right) \log \left(1-\Phi_{ij}\right).
\end{equation}
Also, define the first two derivatives of $\ljbar{\beta_{j}}$ with respect to $\beta_j$ as
\begin{align}
	\bar{\ell}'_{j} \left( \beta_j \right) & = -\lim_{n \rightarrow \infty}  n^{-1} \sum_{i=1}^{n} \frac{\Phi^*_{ij}}{\Phi_{ij}} \phi_{ij} x_{i}+\lim_{n \rightarrow \infty} n^{-1} \sum_{i=1}^{n} \frac{1-\Phi^*_{ij}}{1-\Phi_{ij}} \phi_{ij} x_{i}, \nonumber \\
	\bar{\ell}''_{j} \left( \beta_j \right) & = \lim_{n \rightarrow \infty} n^{-1} \sum_{i=1}^{n} \frac{(\phi^*_{ij})^2}{\Phi^*_{ij}\left(1-\Phi^*_{ij}\right)} x_{i} x_{i}^{\prime}. \label{eq:p1hessian}
\end{align}

\begin{theorem}\label{thm.fiststagebvm} If Assumptions \ref{as:boundedparam}, \ref{as:prior} and \ref{as:designmat} hold,  then there exists a sequence $\tilde{\beta_{j}}\rightarrow \beta_{j}^*$ such that $\bar{\ell}'_{j}(\tilde{\beta}_j) = 0$ for all $n$ sufficiently large, and $\lnjbar{\tilde{\beta_j}} \rightarrow \ljbar{\beta_{j}^*}$ in probability. Let $R_j = \bar{\ell}''_{j}(\beta_j^*)$, then 
$$d_n \sim \dfrac{\exp\{-n\lnjbar{\tilde{\beta_j}}\}\pi(\beta_j^*)}{|R_j|^{1/2}}\left(\dfrac{2\pi}{n}\right)^{p/2}.$$
Furthermore, if we let $g_{nj}$ be the density of $\sqrt{n}(\beta_j-\tilde{\beta_j})$ with $\beta_j \sim \Pi_j^*(\beta_j \mid y, X)$, then 
\begin{equation}\label{eq:first_stage_tv}
     \norm{g_{nj} - \phi_{}}_1 \to 0, \quad \text{in } P_{\theta^*}-\text{probability}, 
\end{equation}
where $\phi_{R_j^{-1}}$ is the density of a $p$-dimensional Gaussian distribution centered at $0$ with covariance matrix $R_j^{-1}$.
\end{theorem}
Theorem \ref{thm.fiststagebvm} guarantees the asymptotic normality of the approximator $\Pi^*_j(\beta_j \mid y, X)$, and the validity of the Laplace approximation. Roughly speaking, \eqref{eq:first_stage_tv} shows the approximator $\Pi^*_j(\beta_j \mid y, X) \approx \Gauss(\tilde{\beta_{j}}, R_j^{-1} / n)$ (Laplace approximation) when $n$ is large. Furthermore, the consistency of the mean of such Laplace approximation, $\tilde{\beta_{j}}$, is also shown in Theorem \ref{thm.fiststagebvm}.

The proof is provided in the supplementary materials. While this result establishes that the approximate marginal posteriors concentrate around $\beta_j^*$ and are asymptotically normal, in general this does not guarantee nominal coverage of credible intervals obtained from $\Pi_j^*(\beta_j \mid y, X)$. This is because the maximum composite likelihood estimator $\sqrt{n}(\tilde{\beta}_j - \beta_j^*)$ asymptotically follows $\Gauss(0, R_j^{-1}V_j R_j^{-1})$ where $V_j = E_{\theta^*}(u_{\beta_j}(y^{(j)}, \beta_j^*)u_{\beta_j}(y^{(j)}, \beta_j^*)^\T)$ and $u_{\beta_j}(y^{(j)}, \beta_j^*) = \nabla_{\beta_j} \ell_j(\beta_j)\big|_{\beta_j = \beta_j^*}$. Hence, $\Pi_j^*(\beta_j \mid y, X)$ has correct asymptotic frequentist coverage iff $R_j = V_j$ which happens when each of the marginal likelihoods are correctly specified \citep{ko2019model}. This is true if the joint multivariate probit model is the true data-generating model. Indeed, we have,
\begin{equation} 
    \sqrt{n}(\tilde{\beta}_j - \beta_j^*) \xrightarrow{d} \Gauss(0, R^{-1}_j). \label{eq:firststageasy}
\end{equation}
\begin{theorem}
\label{thm.firststagecov}
	Under Assumptions \ref{as:boundedparam} - \ref{as:designmat} define a Borel-measurable sequence of sets $W_{n,l} =\left\lbrace \beta_{j,l} :  \hat{\beta}_{j,l}-\left(R_{j}^{-1}/n\right)_{ll}^{1/2} v_{1-\frac{\alpha}{2}}\leq \beta_{j,l} \leq \hat{\beta}_{j,l}+\left(R_{j}^{-1}/n\right)_{ll}^{1/2}v_{1-\frac{\alpha}{2}}\right\rbrace$, where $\Phi(v_{1-\frac{\alpha}{2}}) = 1- \alpha/2$ for $0\leq \alpha \leq 1$ and $1\leq j\leq q$ and $\beta_{j,l}$ is the $l$-th element of $\beta_j$. Then
	$$
	    P_{\theta^*}\left(\beta^*_{j,l}\in W_{n,l}\right) \rightarrow \Phi\left( v_{1-\frac{\alpha}{2}}\right)-\Phi\left(- v_{1-\frac{\alpha}{2}}\right) = 1-\alpha.
	$$
	
\end{theorem}
\begin{proof}
We have
$P_{\theta^*}\left(\beta^*_{j,l}\in W_{n,l}\right) = 
P_{\theta^*}\left\lbrace\sqrt{n}\left(R_{j}^{-1/2}\right)_{ll}(\hat{\beta}_{j,l}-\beta^*_{j,l}) \in [-v_{1-\frac{\alpha}{2}}, v_{1-\frac{\alpha}{2}}]\right\rbrace$. This probability tends to $\Phi\left( v_{1-\frac{\alpha}{2}}\right)-\Phi\left(- v_{1-\frac{\alpha}{2}}\right)$ which follows from (\ref{eq:firststageasy}) and the fact that $\sqrt{n}\left(R_{j}^{-1/2}\right)_{ll}(\hat{\beta}_{j,l}-\beta^*_{j,l})$ and $\sqrt{n}\left(R_{j}^{-1/2}\right)_{ll}(\tilde{\beta}_{j,l}-\beta^*_{j,l})$ have the same asymptotic distribution.
\end{proof}
Theorem \ref{thm.firststagecov} shows that equi-tailed or highest posterior density credible intervals credible intervals constructed for the $l$-th component of $\beta_j$ using $\Pi^*_j(\beta_j \mid y, X)$ have the correct frequentist coverage. The equi-tailed $1-\alpha$ credible interval for $\beta_{j,l}$, which is $\left[\hat{\beta}_{j,l}-\left(R_{j}^{-1}/n\right)_{ll}^{1/2} v_{1-\frac{\alpha}{2}}, \hat{\beta}_{j,l}+\left(R_{j}^{-1}/n\right)_{ll}^{1/2}v_{1-\frac{\alpha}{2}}\right]$, covers the truth $\beta^*_{j,l}$ with probability close to $1-\alpha$ when $n$ is large. 

\subsection{Second-Stage Analysis}
\label{sec:snd_stage_anal}
In our second stage analysis,  we use the likelihood $\prod_{j=1}^{q-1} \prod_{k=(j+1)}^q \mbox{pr}(z_{ij} \in E_{ij}, z_{ik} \in E_{ik})$ for the $i$th data point which can be seen as a pairwise composite likelihood. As a result, the results of \cite{miller2021asymptotic} can be used to study concentration and asymptotic normality of $\Pi_{jk}^*(\sigma_{jk} \mid y, X)\propto \{\prod_{i=1}^n e^{\ell_i^{jk}(\sigma_{jk})}\} \Pi_{jk}(\sigma_{jk})$. Intuitively, if $\tilde{\sigma}_{jk}$ is the maximum composite likelihood estimator from the bivariate margins, i.e. $\tilde{\sigma}_{jk} = \argmax \sum_{i=1}^n \ell_i^{jk}(\sigma_{jk})$, then $\Pi_{jk}^*(\sigma_{jk} \mid y, X)$, when suitably scaled, is close to a Gaussian distribution centered at $\tilde{\sigma}_{jk}$.
Treating $\sigma_{jk}$ as the parameter of interest, these bivariate margins are correctly specified when the latent $(z_{ij}, z_{ik}) \sim \Gauss(\mu_{jk}^*, \Sigma_{jk})$ where $\mu_{jk}^* = (x_i^\T \beta_j^* , x_i^\T \beta_k^*)$ and $\Sigma_{jk} = \{(1, \sigma_{jk})^\T; (\sigma_{jk}, 1)^\T\}$. 
However, in incorporating the uncertainty associated with estimating the regression coefficients $(\beta_j,\beta_k)$, we fit the likelihood $\Gauss(\tilde{\mu}_{jk}, \tilde{\Sigma}_{jk})$, where $\tilde{\mu}_{jk}$ and $\tilde{\Sigma}_{jk}$ are defined in Section \ref{sec:second_stage}. Let $\tilde{\sigma}_{jk}$ to be the two-stage M-estimator obtained as the solution of
\begin{equation}\label{eq:sndestequ}
    \sum_{i=1}^n \frac{\partial \ell_i^{jk}(\sigma_{jk})}{\partial \sigma_{jk}} = 0.
\end{equation}
We then have the following asymptotic result on  $\Pi^*_{jk}(\sigma_{jk}\mid y, X)$.

\begin{theorem}\label{prop:sndstagebvm}
	If Assumptions \ref{as:boundedparam} to \ref{as:designmat} hold, then there exists a solution $\tilde{\sigma}_{jk}$ of \eqref{eq:sndestequ} for all sufficiently large $n$ with $\tilde{\sigma}_{jk} \rightarrow \sigma^*_{jk}$ in $P_{\theta^*}$-probability. Let $R_{jk} = \lim_{n \rightarrow \infty} -\dfrac{1}{n}\sum_{i=1}^n E_{\theta^*}\left( \frac{\partial^2{\ell_i^{jk}}}{\partial{\sigma^2_{jk}}}\right) = 
	\lim_{n \rightarrow \infty} \frac{1}{n}\sum_{i=1}^n \Var_{\theta^*}\left(\frac{\partial{\ell_{i}^{jk}}}{\partial{\sigma_{jk}}}\right)$. 
	If $g_{njk}$ is the density of $\sqrt{n}(\sigma_{jk} - \tilde{\sigma}_{jk})$, where $\sigma_{jk} \sim \Pi_{jk}^*(\sigma_{jk} \mid y, X)$, then
	$$ \norm{g_{njk} - \phi_{R^{-1}_{jk}}}_1 \to 0 \quad \text{in } P_{\theta^*}-\text{probability},
	$$
	 where $\phi_{R^{-1}_{jk}}$ is the density of a univariate Gaussian centered at 0 with variance $R_{jk}^{-1}$.
	
\end{theorem}
Theorem \ref{prop:sndstagebvm} establishes the asymptotic normality of $\Pi_{jk}^*(\sigma_{jk} \mid y, X)$ and thus validates our method of approximating it by a Gaussian distribution as done in Section \ref{sec:second_stage}.
The proof is provided in the supplementary materials. In Lemma \ref{prop:twostageestimator} we show that asymptotically $\tilde{\sigma}_{jk}$ has a Gaussian distribution centered at the true value but with a larger variance. This inflation of the variance results from the extra uncertainty induced by using $(\hat{\beta_{j}}, \hat{\beta_{k}})$ in \eqref{eq:sndestequ} which ideally should be evaluated at $(\beta_j^*, \beta_k^*)$. Such inflation in two stage estimators has been observed previously \citep{murphy2002estimation}. To derive the correct variance, one thus needs to a) characterize the behaviour of $\ell_i^{jk}(\sigma_{jk})$ as a function of $\beta_j$ and $\beta_k$, and for that we now make the dependence of $\ell_i^{jk}(\sigma_{jk})$ on $\beta_j$ and $\beta_k$ explicit by writing $\ell_i^{jk}(\sigma_{jk}; \mu_{jk}, \Sigma_{jk})$ where $\mu_{jk} = (x_i^\T \beta_j, x_i^\T \beta_k)$; b) quantify the effect of using $\hat{\beta}_j, \hat{\beta}_k$ instead of $\beta_j^*$, $\beta_k^*$ in the second stage of our inference method. In particular, we rely on the score functions with respect to $\beta_j$, $\beta_k$, and $\sigma_{jk}$, i.e., $\nabla_{\beta_j} \ell_{i}^j(\beta_j)$, $\nabla_{\beta_k} \ell_{i}^k(\beta_k)$ and  $\nabla_{\sigma_{jk}} \ell_{i}^{jk}(\sigma_{jk}; \mu_{jk}; \Sigma_{jk})$ respectively.
The following quantities will be helpful in defining the correct variance:
\begin{align} \label{eq.asymptotic.quantities}
	R_j^{jk}  & =  \lim_{n \rightarrow \infty} \frac{1}{n} \sum_{i=1}^{n} \Cov_{\theta^*}\left\lbrace \nabla_{\beta_j} \ell^j_{i}(\beta_j)\bigg|_{\beta_j = \beta_j^*},  \nabla_{\sigma_{jk}}\ell_{i}^{jk}(\sigma_{jk}; \mu_{jk}^*, \Sigma_{jk})\bigg|_{\sigma_{jk} = \sigma_{jk}^*}\right\rbrace , \nonumber\\
	R_k^{jk}  & =  \lim_{n \rightarrow \infty} \frac{1}{n} \sum_{i=1}^{n} \Cov_{\theta^*}\left\lbrace \nabla_{\beta_k} \ell^k_{i}(\beta_k)\bigg|_{\beta_k = \beta_k^*} ,  \nabla_{\sigma_{jk}}\ell_{i}^{jk}(\sigma_{jk}; \mu_{jk}^*, \Sigma_{jk})\bigg|_{\sigma_{jk} = \sigma_{jk}^*}\right\rbrace  \nonumber, \\
	V^{jk} & = \lim_{n \rightarrow \infty} \frac{1}{n} \sum_{i=1}^{n} \Cov_{\theta^*}\left\lbrace \nabla_{\beta_j} \ell^j_{i}(\beta_j)\bigg|_{\beta_j = \beta_j^*}, \nabla_{\beta_k} \ell^k_{i}(\beta_k)\bigg|_{\beta_k = \beta_k^*}\right\rbrace, \nonumber \\
	Q_{j}^{jk} &= - \lim_{n \rightarrow \infty}\frac{1}{n}\sum_{i=1}^{n}E_{\theta^*} \left\lbrace\sndpartial{\ell_{i}^{jk}(\sigma_{jk}; \mu_{jk}, \Sigma_{jk})}{\sigma_{jk}}{\beta^T_j}\bigg|_{\sigma_{jk} = \sigma_{jk}^*, \beta_j = \beta_j^*}\right\rbrace, \nonumber \\
	Q_{k}^{jk} &=  -\lim_{n \rightarrow \infty} \frac{1}{n}\sum_{i=1}^{n}E_{\theta^*} \left\lbrace\sndpartial{\ell_{i}^{jk}(\sigma_{jk}; \mu_{jk}, \Sigma_{jk})}{\sigma_{jk}}{\beta^T_k}\bigg|_{\sigma_{jk} = \sigma_{jk}^*, \beta_k = \beta_k^*}\right\rbrace. 
\end{align}

In the above display, $R_j^{jk}$ is the average covariance between the score function of the first stage with respect to $\beta_j$ and the score function of the second stage with respect to $\sigma_{jk}$ when evaluated at true parameter values. Similarly, 
$V_{jk}$ is the average covariance between the score function of the first stage of $\beta_j$ and $\beta_k$. Finally, $Q_{j}^{jk}$ quantifies the change in $\ell_i(\sigma_{jk};\mu_{jk}, \Sigma_{jk})$ with respect to both $\beta_j$ and $\sigma_{jk}$, on average.  Detailed expressions of these quantities are provided in Section \ref{sec.quantities} of the supplementary materials where we show that $R_j^{jk}$ and $R_k^{jk}$ are zero for all $1\leq j \neq k \leq q$. The existence of all limits in the above display is guaranteed by Assumption \ref{as:designmat}. Also, recall the definition of $R_j$ from Theorem \ref{thm.fiststagebvm} for any arbitrary $j$. Set
\begin{equation}\label{eq:second_stage_correct_variance} 
\tau_{jk} = R_{jk}^{-1}  + R_{jk}^{-1}\left(Q_j^{jk}R_j^{-1}Q_j^{{jk}^\T}\right) R^{-1}_{jk} +  R_{jk}^{-1}\left(Q_k^{jk}R_k^{-1}Q_k^{{jk}^\T}\right) R^{-1}_{jk} +  2R_{jk}^{-1}Q_j^{jk}R_j^{-1}V^{jk}R_k^{-1}Q^{{jk}^\T}_k R_{jk}^{-1},
\end{equation}
where the second and third terms in the preceding display account for the extra uncertainty induced in estimating $\sigma_{jk}$ by using estimates $(\hat{\beta}_j, \hat{\beta}_k)$ of $(\beta_j, \beta_k)$, respectively. The final term roughly quantifies cross-covariance between the first stage score function and the second stage score function.
\begin{lemma}\label{prop:twostageestimator}
	Under Assumptions \ref{as:boundedparam}, \ref{as:corr} and \ref{as:designmat} the following asymptotic normality of two-stage M-estimator $\tilde{\sigma}_{jk}$ holds:
	$$
	\sqrt{n}(\tilde{\sigma}_{jk}-\sigma^*_{jk}) \xrightarrow{d} \Gauss(0, \tau_{jk})
	$$
	where $\tau_{jk}$ is defined in \eqref{eq:second_stage_correct_variance}.
\end{lemma}

The detailed proof of Lemma \ref{prop:twostageestimator} is presented in Section \ref{sec:normality_of_two_stage} of the online supplement.
Due to this miscalibration in the variance, equi-tailed credible intervals computed from the Gaussian approximation $\Gauss(\hat{\sigma}_{jk}^2, s_{jk}^2)$ of $\Pi^*_{jk}(\sigma_{jk}\mid y, X)$ will typically have under coverage since $\tau_{jk} > R_{jk}^{-1}$ \citep[Section 4.1]{miller2021asymptotic}. 

\begin{theorem}\label{thm.coverage}

	Under Assumptions \ref{as:boundedparam} - \ref{as:designmat} define a Borel-measurable sequence of sets $S_n =\left\lbrace \sigma_{jk} :  \hat{\sigma}_{jk}-\left(R_{jk}^{-1}/n\right)^{1/2} v_{1-\frac{\alpha}{2}}\leq \sigma_{jk} \leq \hat{\sigma}_{jk}+\left(R_{jk}^{-1}/n\right)^{1/2}v_{1-\frac{\alpha}{2}}\right\rbrace$, where $\Phi(v_{1-\frac{\alpha}{2}}) = 1- \alpha/2$ for $0\leq \alpha \leq 1$ and $1\leq j\neq k\leq q$. Suppose $\tau_{jk}$ is as defined in \eqref{eq:second_stage_correct_variance}. Then,
	$$
	    P_{\theta^*}\left(\sigma^*_{jk}\in S_n\right) \rightarrow \Phi\left(\sqrt{R^{-1}_{jk}/\tau_{jk}}\,z_{1-\frac{\alpha}{2}}\right)-\Phi\left(-\sqrt{R^{-1}_{jk}/\tau_{jk}}\,z_{1-\frac{\alpha}{2}}\right) < 1-\alpha
	$$
	
\end{theorem}

The detailed proof of Lemma \ref{thm.coverage} is presented in Section \ref{sec:snd_stg_coverage} of the online supplement. Theorem \ref{thm.coverage} shows that equi-tailed credible intervals of $\sigma_{jk}$ are expected to have bias in coverage. To remove this bias, a natural remedy is to consider a consistent estimator of $\tau_{jk}$ and construct intervals based on this variance. Unfortunately, this involves an additional 
$\mathcal{O}(p^2q^2)$ complexity.
Fortunately, in practice in all our simulations, we found the bias to be negligible, so that bias removal is not practically worth the additional computational expense. We carry out a detailed simulation study in Section \ref{sec:tauVR} of the supplementary materials where our results show that adjusting for the bias requires 8-20 times the computational time while 
lengths and coverage of intervals obtained from the two variances are almost identical. We conjecture this is due to the special structure of the MVP model for which covariance between score functions for the regression coefficients and the correlations is very small.

\section{Simulation results}\label{sec:simulations}
We evaluate performance of the non-hierarchical (bigMVP) 
and hierarchical (bigMVPh) versions of the proposed method through simulation studies. To benchmark the results, we compare it with a \texttt{STAN} \citep{stan} implementation of dynamic Hamiltonian Monte Carlo (HMC) for the MVP model. We additionally consider a Variational Bayes (VB) approximation as a computationally faster alternative for approximate Bayesian inference implemented using \texttt{STAN}; our repeated attempts of implementing the Automatic Differentiation Variational Inference method in \texttt{pymc3} did not work. For the dynamic HMC and VB implementation we consider a product prior $\prod_{j=1}^q \Pi(\beta_j)$ on the regression coefficients and an $\text{LKJ}(\nu)$ prior on the correlation matrix; specifically we set $\Pi(\beta_j)$ as $\Gauss(0,5^2)$ and $\nu = 1$. With this choice of priors, bigMVP is implemented with the same prior on the regression coefficients and the corresponding marginal prior on the correlations which are proportional to a Beta density with parameters $(q/2, q/2)$. The results of bigMVPh are not directly comparable in involving a different set of prior distributions. However, we include these results to illustrate efficacy of the hierarchical extension when many of the binary outcomes are rarely observed. In addition to the Bayesian methods listed above, we also consider the two-stage frequentist method (TSF) of \cite{ting2022fast} and the method of \cite{pichler2020new} which uses the algorithm from \cite{chen2018end} to compute Gaussian orthant probabilities in parallel and imposes an elastic net penalty on the correlation matrix for handling large numbers of outcomes. We abbreviate the method due to \cite{pichler2020new} as fMVP.

The dynamic HMC sampler is the most computationally intensive method for which we fix a computational budget of $t_{\text{wall}}= \min\{t_N, 1.5 \text{ hrs}\}$ where we set $t_N$ to be the time needed to obtain $N$ number of samples from the posterior. We set $N = 6000$, of which we discard the first 1000 samples to compute posterior summaries. For a comparison of runtime of the other methods under consideration, we report the ratio $t_{wall}/t_{\text{method}}$ where $t_{\text{method}}$ is the runtime of a particular method. All experiments were carried out on a computer with the following specifications - 64 bit Intel i7-8700K CPU @3.7 GHz processor. 

We consider sample sizes $n = 200, 500$, fix the number of covariates to $p = 5$ and vary the number of outcomes as $q = 10, 15, 20, 100, 200$. The HMC sampler and the TSF method were too computationally intensive for $q = 100, 200$. Hence, for these choices of $q$, we only report results for bigMVP, bigMVP$_\text{h}$, VB and  fMVP. The metrics on which the methods are evaluated are estimation error and uncertainty quantification. Given data $(y, X)$ simulated assuming parameters $\theta^* = (B^*, \Sigma^*)$, we compute the estimation error as $\norm{\hat{B} - B^*}_F/pq$ (E1) and $\norm{\hat{\Sigma} - \Sigma^*}_F/q^2$ (E2) for any estimator $\hat{B}$ and $\hat{\Sigma}$. For HMC and VB , $\hat{B}$ and $\hat{\Sigma}$ are the estimated posterior means, for TSF we take the maximum (composite) likelihood estimator as $\hat{B}$ and $\hat{\Sigma}$, for $\text{fMVP}$ we take the maximum penalized likelihood estimator and for bigMVP and $\mbox{bigMVP}_h$
 we take $\hat{B} = (\hat{\beta}_1, \ldots, \hat{\beta}_q)$ and $\hat{\Sigma} = (\hat{\sigma}_{jk})$. 

In simulating the data we consider two different settings for the regression coefficients and the correlation matrix. For the coefficient matrix we consider the two cases 1) (Dense) $\beta_{lj}^* \sim \Gauss(0,1)$  and the intercept term $\beta_{0j}^* \sim \Gauss(0,1)$ for $l=1, \ldots p$, $j = 1, \ldots q$ and 2) (Rare) $\beta_{lj}^* \sim \Gauss(0,1)$ and the intercept term is set to $\beta_{0j}^* = -3$ so that many outcomes are rarely observed. 
For the correlation matrix we first generate a covariance matrix $\Gamma$ and then set the corresponding correlation matrix as $\Sigma = D^{-1} \Gamma D^{-1}$ where $D = \mathrm{diag}(\Gamma_{11}^{1/2}, \ldots, \Gamma_{qq}^{1/2})$. The settings we consider are 1) (Factor) generate $\Gamma^* = \Lambda \Lambda^\T + \mathrm{I}_q$, where $\Lambda$ is a $q\times k$ matrix with $k = 3$ and  $\lambda_{jl} \sim N(0,1)$ and 2) (Block) $\Gamma^* = \mathrm{diag}(\Gamma_1^*, \Gamma_2^*, \ldots)$, where each $\Gamma_l = LL^\T$ with the elements of $L$ generated from a $\Gauss(0,1)$ distribution; the size of each $\Gamma_l$ is taken to be $5 \times 5$. For each combination of the regression coefficients and correlations we consider two data generating scenarios - 1) well-specified: where the latent variable $z \sim \Gauss({B^*}^T x, \Sigma^*)$ and 2) misspecified: where  $z \sim t_{10}({B^*}^T, \Sigma^*)$, i.e. a multivariate $t$ distribution with 10 degrees of freedom with mean $B^\T x$ and scale matrix $\Sigma$.
Elements of the design matrix are simulated from $\Gauss(0,1)$. To implement the hierarchical extension of the proposed method, the initial conditional sampler is run 200 times of which the first 50 samples are discarded. 
For each setting in the well-specified case, the logarithm of the errors averaged over 30 independent replications are displayed in Figures \ref{fig:accuracy_E1_well} and \ref{fig:accuracy_E2_well} for $n=200$; corresponding numerical values are given in Table \ref{tab:error_table} of the supplementary materials along with similar plots displaying results obtained when the model is misspecified.

In terms of estimating the regression coefficients, the proposed method performs better across all settings compared to the other methods. The improvement over other methods is especially stark when many binary outcomes are rare; for methods based on sampling this could potentially be due to very poor mixing in the rare outcome case. For instance, with 5000 MCMC iterations, the average effective sample size is roughly 800 when all the outcomes are balanced, whereas when some of the outcomes are rare this number goes down to 200. In terms of performance in estimating the correlation matrix, bigMVP and bigMVPh are almost equal to HMC, especially when $q$ is high. 
Moreover, the runtime ratio $t_{wall}/t_{method}$, when averaged over all possible data generating cases and $q = 10, 15, 20$, is approximately 20000 for bigMVP and  120 for bigMVPh. This ratio for the other methods fMVP, TSF and VB are 590, 70 and 680, respectively. When the latent variables are sampled from a $t_{10}$ distribution, HMC and VB marginally outperform bigMVP and bigMVPh, albeit using much more computing resource.
Clearly, bigMVP provides huge gains in computational time while maintaining similar, if not better, levels of accuracy. 

\begin{figure}
    \centering
    \includegraphics[width = 0.8\textwidth, height = 9cm]{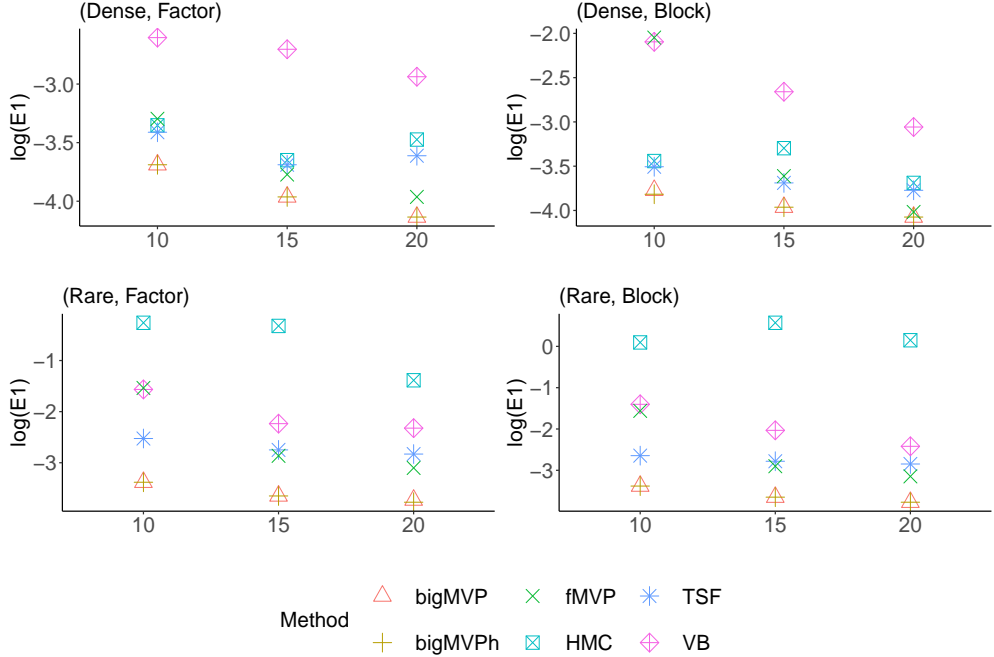}
    \caption{Logarithm of Frobenius errors in estimating the matrix of regression coefficients $B^*$ when the sample size $n = 200$, number of covariates $p = 5$ and the number of binary responses considered are $q = 10, 15, 20$.  }
    \label{fig:accuracy_E1_well}
\end{figure}

\begin{figure}
    \centering
    \includegraphics[width = 0.8\textwidth, height = 9cm]{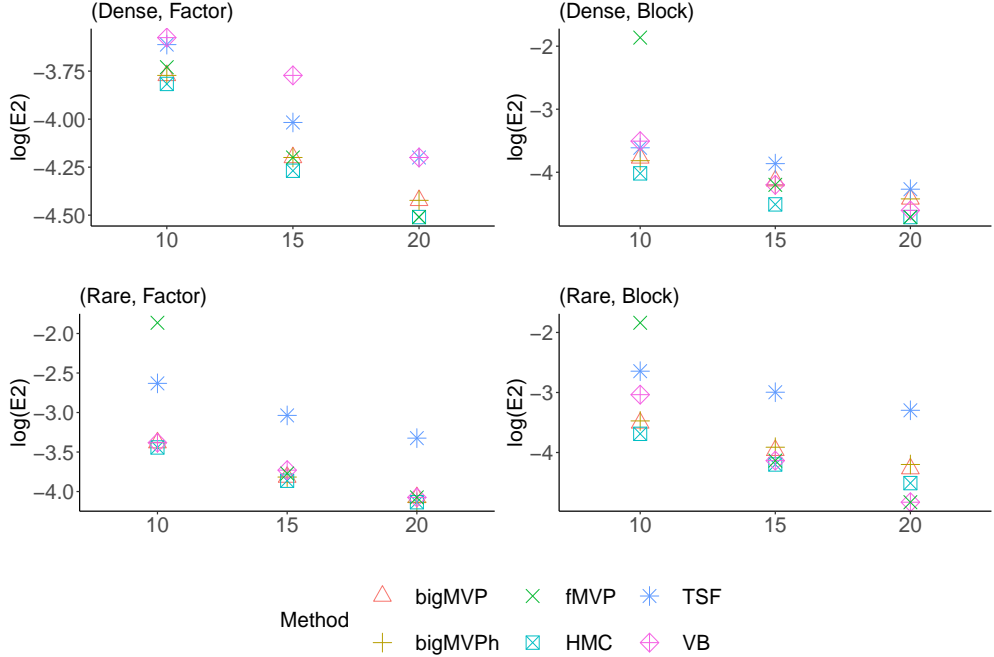}
    \caption{Logarithm of Frobenius errors in estimating the matrix of correlation coefficients $B^*$ when the sample size $n = 200$, number of covariates $p = 5$ and the number of binary responses considered are $q = 10, 15, 20$.  }
    \label{fig:accuracy_E2_well}
\end{figure}
We further investigate the proposed method's ability to accurately quantify uncertainty for the correlation matrix. We compare with the two-stage frequentist method 
 and leave out the data augmented sampler for these experiments due to the very high computing time. 
 Variance of the parameters for the two-stage method in \cite{ting2022fast} is computed following \cite{hardin2002robust} which adjusts for the two-step nature of the method. We exclude results from VB since it had significant undercoverage even for small $q$; for $q = 4$, we obtained a coverage percentage of only 28\% for 95\% credible intervals.
 For individual parameters, $\sigma_{jk}$ where $j, k = 1, \ldots, q$, we compute a marginal confidence interval from estimated variance-covariance matrices for \cite{ting2022fast}. Similarly, we compute credible intervals for these parameters for bigMVP. 

We consider the same combination of data generating parameters and for each of these settings we generate 10 different values of $\theta^*$. Then, for each of these values of $\theta^*$, we generate 100 data sets and calculate how many of these intervals contain the true parameter values. The average coverage for 95\% confidence/credible intervals across all the parameter values and all correlation coefficients is reported in Table \ref{tab:coverage_table} of the supplement. A visual summary of the results is provided in Figure \ref{fig:coverage} and comparison of width of the intervals averaged over all correlation coefficients is provided in Figure \ref{fig:int_width} of the supplementary materials for $n=200$. For all cases considered, the proposed method provides very close to nominal coverage but the corresponding frequentist method has severe under coverage when many of the binary outcomes are rare. The interval widths also reflect that incorporating the entire marginal posterior distributions of the regression coefficients in the second stage improves the coverage substantially compared to the asymptotic adjustment in \cite{ting2022fast, hardin2002robust}.
 When the outcomes are relatively more common, the coverage of the frequentist method does improve, although it is still not close to the nominal level.
\begin{figure}
    \centering
    \includegraphics[width = 0.8\textwidth, height = 9 cm]{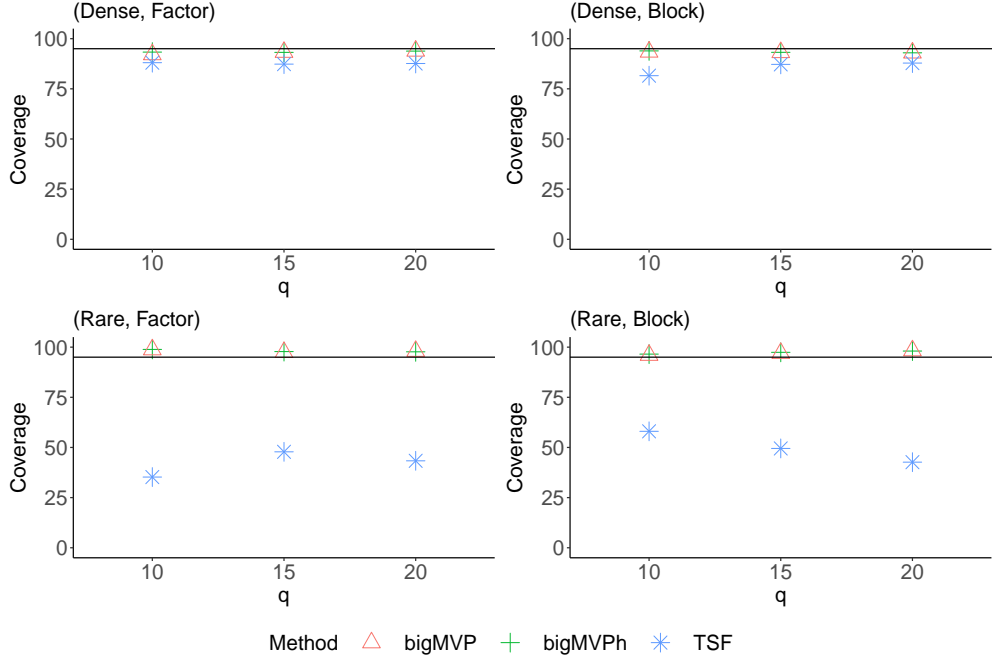}
    \caption{Comparison of coverage of 95\% credible/confidence intervals for the correlation coefficients obtained from bigMVP, bigMVPh versus TSF \citep{ting2022fast}. The black solid line represents the line $y = 95$. }
    \label{fig:coverage}
\end{figure}

\section{Applications}
We apply the proposed methodology to two data sets in this section, 1) \textbf{Bird data:} \cite{lindstrom2015large} compiled these data from the national bird monitoring programs in Finland, Sweden and Norway. In total $q = 141$ bird species were sampled using line transects (Finland and Sweden) and point counts (Norway). Information on 21 covariates related to land cover, climate and other factors were also collected. We follow \cite{norberg2019comprehensive} in including $p = 5$ 
    covariates obtained from an initial principal components analysis on the 21 covariates, and 2) \textbf{Vegetation data:} These data on arctic vegetation come from a community ecology survey conducted in northern Norway \citep{niittynen2018importance}. The data consist of $q = 242$ different species of plants, bryophytes and lichens and 6 environmental covariates related to soil, topography and climate were also recorded. We follow \cite{norberg2019comprehensive} to include $p = 4$ covariates after an initial principal components analysis. For both datasets the number of principal components is chosen so that about 93\% of the variation in the original covariates is explained.

We compare the prediction performance of bigMVP and bigMVPh with TSF. Predictions for bigMVP and bigMVPh are obtained using pairwise approximations to the posterior predictive detailed in Section \ref{sec:pairwise_prediction} of the supplementary materials. We held out 50 test points for each dataset. The number of unique pairs of species for the bird data is 9870 and for the vegetation data is 29161. For each of these pairs of species we sampled 100 predictive samples. Suppose $x_t$ is the test point and we are considering the $(j,k)$-th pair of species. We computed the predictive mean for this test point and pair of species averaging over the 100 predictive samples. The predictive accuracy is then computed as the difference between the predictive mean and the observed species indicators at this test point.  We additionally computed the class of these outcomes as 1 or 0 according to whether the predictive mean is above or below 0.5. We follow a similar pairwise prediction strategy with maximum composite likelihood estimates plugged in for TSF. The misclassification rate is then computed as the difference in observed values and predicted classes. For the pair $(j,k)$ and $n_t$ test points we then obtained the average prediction error and misclassification error. We summarise the results in Table \ref{tab:pred_error}. On the misclassification metric, bigMVP and bigMVPh have lower error rate compared to TSF almost always. Interestingly, the hierarchical extension bigMVPh has a much better misclassification error for the Bird data, where many of the species are rare. 
\begin{table}
    \centering
    \scalebox{0.7}{\begin{tabular}{c|cccc|cccc}
    \hline
         &\multicolumn{3}{c}{Bird data} & \multicolumn{3}{c}{Vegetation data} \\
         \hline
         & Prediction error & \multicolumn{3}{c}{Misclassification error} & Prediction error & \multicolumn{3}{c}{Misclassification error}\\
         \hline
         &bigMVP & bigMVP & bigMVPh& TSF & bigMVP & bigMVP & bigMVPh & TSF \\
         \hline
       Minimum  & 0.0005 & 0  & 0 & 0 & 0.0002 & 0 & 0&  0\\
       1st quartile & 0.048 & 0.049 & 0.046& 0.066 & 0.010 & 0 & 0& 0.028\\
       Median & 0.067 & 0.077 & 0.073& 0.091 &  0.039 & 0.04 & 0.03 & 0.056\\
       3rd quartile & 0.08 & 0.097 & 0.092 & 0.107 & 0.064 & 0.077 & 0.075 & 0.089\\
       Maximum & 0.117 & 0.165  & 0.154 & 0.161 & 0.117 & 0.16 & 0.157 & 0.158\\
       \hline
    \end{tabular}}
    \caption{Summary statistics of the pairwise prediction and misclassification errors for the bird data and vegetation data.}
    \label{tab:pred_error}
\end{table}

\section{Discussion}
Our proposed bigMVP Bayesian method provides huge computational benefits compared to sampling based methods without compromising on statistical accuracy. The hierarchical extension leads to very substantial gains in practical performance over frequentist competitors, enabling borrowing of information across outcomes, leading to a particularly substantial improvement in performance for rare outcomes.
Our proposed approach is supported by theoretical guarantees showing the performance improves with sample size.  The focus of this theory is on marginal posterior distributions for the parameters, which is the emphasis of inference in our motivating application areas. We also develop an approximate method for pairwise prediction. 

There are several interesting future directions from both methodological and application perspectives. In ecology, species occurrence data typically come with important spatial information. Extending the model and corresponding methodology to incorporate spatial dependence is straightforward. Consider $q$ dimensional binary outcomes indexed by spatial locations $s$, that is, we observe $y(s)$ and $X(s)$ for each location. Assuming the underlying correlation structure of the species is the same across locations, one only needs to replace the first and second stage likelihoods by their spatial versions. This approach can be extended to additionally account for spatio-temporal dependence.  A more challenging problem is to
incorporate covariate-dependent 
correlation; for example, in ecology the dependence between species can vary according to the habitat. The proposed method can also be extended to handle species count data under a latent Gaussian assumption. 

\bigskip
\begin{center}
{\large\bf SUPPLEMENTARY MATERIAL}
\end{center}
The supplementary materials contain proofs of all the results in Section \ref{sec:theory}, an approximate sampling algorithm for fitting the hierarchical model in Section \ref{sec:hierarchical}, and numerical results a) comparing intervals of correlations obtained using the miscalibrated variance versus the correct variance b) marginal approximations under a joint prior (Remark \ref{rm:marginal_prior}) c) comparisons with INLA-MCMC and d) results under model misspecification. In addition, we provide details of the approximation to the pairwise predictive distributions. An \texttt{R} package has been developed for implementing the developed methodology which can be found at \href{https://anonymous.4open.science/r/bigMVP-9EAC}{https://anonymous.4open.science/r/bigMVP-9EAC}. 
\bibliographystyle{chicago}
\bibliography{refs}

\newpage
\setcounter{equation}{0}
\setcounter{page}{1}
\setcounter{table}{1}
\setcounter{section}{0}
\numberwithin{table}{section}
\renewcommand{\theequation}{S.\arabic{equation}}
\renewcommand{\thesubsection}{S.\arabic{section}.\arabic{subsection}}
\renewcommand{\thesection}{S.\arabic{section}}
\renewcommand{\thetable}{S.\arabic{table}}
\renewcommand{\thefigure}{S\arabic{figure}}
\renewcommand{\thelemma}{S\arabic{lemma}}
\renewcommand{\bibnumfmt}[1]{[S#1]}
\renewcommand{\citenumfont}[1]{S#1}

\begin{center}
\Large{\bf Supplementary materials for ``Bayesian inference on high-dimensional multivariate binary responses"}
\end{center}

\section{Proofs of results in the main document}
\subsection{Proof of Theorem \ref{thm.fiststagebvm}}\label{sec:firststagebvm_proof}
     In the proof, we write $\beta_1$ instead of $\beta_j$ for convenience as the proof holds for any arbitrary $j=1,\ldots, q $. 
     The theorem is a direct consequence of Theorem 3.2 of \cite{miller2021asymptotic}; we will need to verify four sufficient conditions: {\bf A)} $\bar{\ell}'''_{n1}(\beta_{1})$ is uniformly bounded in $E$ where $E\subset \RR^p$ is some open, convex and bounded set and $\beta_1^* \in E$, {\bf B)} $\bar{\ell}_1''(\beta_1^*)$ is positive definite, {\bf C)} Each $\bar{\ell}_{n1}$ is convex in $E$ and  {\bf D)} $\bar{\ell}_1'(\beta_1^*) = 0$. Recall  $r_{i1} = 2y_{i1} - 1$ from Section \ref{sec:snd_stage_anal}.
     
{\bf Condition A)}  Letting the sign-transformed mean $\psi_{i1} = r_{i1} x_i^\T \beta_1$, we have
     	$$
	 \bar{\ell}_n'''(\beta_{1})_{l_{1}l_{2}l_{3}} = -\frac{1}{n} \sum_{i=1}^n \frac{\partial^3 \ell_i}{\partial (\psi_{i1})^3 } r_{i1}^3 x_{il_1}x_{il_2}x_{il_3},
     	$$
     	where
     	$$
     	\frac{\partial^3 \ell_i}{\partial (\psi_1)^3}(\psi_{i1}) = \frac{T(\Phi, \phi, \phi', \phi'')}{\Phi^4}(\psi_{i1}),
     	$$
and $l_1, l_2, l_3 = 1, \ldots, p$.     	
Here $T(\Phi, \phi, \phi', \phi'')$ is a polynomial function of $\Phi(\psi_{i1})$, $\phi(\psi_{i1})$, $\phi'(\psi_{i1})$ and $\phi''(\psi_{i1})$. From Assumption \ref{as:boundedparam} and \ref{as:designmat}, $\{\psi_{i1}\}$ is bounded for every $\beta_1\in E$ so there exists an $\epsilon>0$ such that $\Phi(\psi_{i1})>\epsilon>0$ for every $i$. Moreover, since $\{\psi_{i1}\}$ is bounded, $\Phi(\psi_{i1})$, $\phi(\psi_{i1})$, $\phi'(\psi_{i1})$ and $\phi''(\psi_{i1})$ are also bounded and so is $T(\Phi, \phi, \phi', \phi'')$. Therefore, there exists an $M>0$ such that for every $\beta_1 \in E$,
     	$$
     	\left| \frac{\partial^3 \ell_i}{\partial (\psi_{i1})^3 } r_{i1}^3 \right| \leq M \quad a.s.
     	$$
     	We then have
     	$$
     	|\bar{\ell}'''_n(\beta_{1})_{l_1l_2l_3}| \leq \frac{1}{n} \sum_{i=1}^n \left| \frac{\partial^3 \ell_i}{\partial (\psi_{i1})^3 } r_{i1}^3 \right| | x_{il_1}x_{il_2}x_{il_3}| \leq \frac{M}{n} \sum_{i=1}^n |x_{il_1}x_{il_2}x_{il_3}| \quad a.s.
     	$$
     	In addition, from Assumption \ref{as:designmat} we have $\frac{1}{n}\sum_{i=1}^n |x_{il_1}x_{il_2}x_{il_3}| $ converges almost surely for every $j,k,l$. Thus, $|\bar{\ell}_n'''(\beta_{1})_{l_1l_2l_3}|$ is uniformly bounded for every $\beta_1 \in E$ and $l_1,l_2,l_3 = 1, \ldots, p$.
     	
{\bf Condition B)}  $\bar{\ell}''_1(\beta_{1}^*)$ as defined in \eqref{eq:p1hessian} and its positive definiteness is implied by Assumption \ref{as:designmat}.
         
{\bf Condition C)}  The convexity of every $\bar{\ell}_{n1}$ is proved in Theorem 9.2.3 of \cite{amemiya1985advanced}. 

{\bf Condition D)} From \eqref{eq:p1hessian}, we also have that  $$\bar{\ell}'_1(\beta^*_{1})= \lim_{n \rightarrow \infty} n^{-1} \left( \sum_{i=1}^{n} \frac{\Phi^*_{1i}}{\Phi^*_{1i}} \phi^*_{1i} x_i-\sum_{i=1}^{n} \frac{1-\Phi^*_{1i}}{1-\Phi^*_{1i}} \phi^*_{1i} x_i \right) = 0.$$

\subsection{Proof of Theorem \ref{prop:sndstagebvm}}
We will rely on Theorem 3.2 of \cite{miller2021asymptotic} to prove this result. We set $(j, k) = 1,2$ without loss of generality. Before proceeding to the main body of the proof, we introduce the definition of \textit{equi-Lipschitz} as in \cite{miller2021asymptotic}. A family of functions $h_n : E \rightarrow F$, where $E$ and $F$ are subsets of a normed space, is \textit{L-equi-Lipschitz} if there exists an $L>0$ such that such that for all $n \in \mathbb{N}, x, y \in E$, we have $\left\|h_{n}(x)-h_{n}(y)\right\| \leq L\|x-y\|$. We define the negative average log likelihood of $\sigma_{12}$ in the second-stage to be $$
\bar{\ell}_{n12}(\sigma_{12}) = -\frac{1}{n}\sum_{i=1}^n \ell_i^{12}(\sigma_{12}; \tilde{\mu}_{12}, \tilde{\Sigma}_{12}).
$$
For every $\sigma_{12} \in (-1, 1)$, let $\theta^{**}_{12}(\sigma_{12}) =(\beta_1^{*T}, \beta_2^{*T}, vec^T(\Sigma_{12}))$ where only $\beta_1$ and $\beta_2$ are fixed at their respective true values. We also introduce consistent estimator $\tilde{\theta}_{12}(\sigma_{12}) = (\tilde{\beta}_1, \tilde{\beta}_2, vec^T(\tilde{\Sigma}_{12}))$ where $\tilde{\theta}_{12}(\sigma_{12}) \rightarrow \theta^{**}_{12}(\sigma_{12})$ in $P_{\theta^{**}_{12}(\sigma_{12})}$ for every $\sigma_{12} \in (-1, 1)$. Recall the definitions of $\tilde{\Sigma}_{jk}$ from Section \ref{sec:theory}. Hereafter, we write $\theta^{**}_{12}$ and $\tilde{\theta}_{12}$ as shorthand for $\theta^{**}_{12}(\sigma_{12})$ and $\tilde{\theta}_{12}(\sigma_{12})$ respectively.
In addition, we define $\theta^{*}_{12} =(\beta_1^{*T}, \beta_2^{*T}, vec^T(\Sigma^*_{12}))$ where $\beta_1$ and $\beta_2$, and $\Sigma_{12}$ are all fixed on their respective true values.

We will first prove the pointwise convergence of $\bar{\ell}_{n12}(\sigma_{12})$ in $P_{\theta^*}$. We write
\begin{equation}\label{eq:ln12bar}
\bar{\ell}_{n12}(\sigma_{12}) = \underbrace{-\left[\frac{1}{n}\sum_{i=1}^n \ell_i^{12}(\sigma_{12}; \tilde{\mu}_{12}, \tilde{\Sigma}_{12})-
\frac{1}{n}\sum_{i=1}^n \ell_i^{12}(\sigma_{12};\mu^*_{12}, \Sigma_{12})
\right]}_\textrm{part (a)}
-
\underbrace{\frac{1}{n}\sum_{i=1}^n \ell_i^{12}(\sigma_{12};\mu^*_{12}, \Sigma_{12})}_\textrm{part (b)}.
\end{equation}
We will show that in the preceding display part (a) converges to $0$ and part (b) converges to some limit in $P_{\theta^*}$. To prove part (a) converges to $0$, we fix a small enough convex and open neighborhood $U$ of $\theta^{**}_{12}$ and define $h_n(\theta_{12})\coloneqq\frac{1}{n}\sum_{i=1}^n\ell_i^{12}(\sigma_{12};\mu_{12}, \Sigma_{12})$. Under Assumption \ref{as:boundedparam}, \ref{as:corr} and \ref{as:designmat}, we have
$
\sup_n \sup_{\theta_{12}\in U} \norm{ \frac{\partial h_n(\theta_{12})} {\partial \theta_{12}}} < \infty
$.
From Lemma \ref{lem.equilip}, for any $\sigma_{12}\in (-1,1)$, $h_n(\theta_{12})$ is $L$-equi-Lipschitz in $\theta_{12}$ for any $\theta_{12}\in U$. Thus, for any $\xi>0$, we have
\begin{align*}
      & \limsup_{n\rightarrow \infty} P_{\theta^*}\left\{
      \left|\frac{1}{n}\sum_{i=1}^n \ell_i^{12}(\sigma_{12}; \tilde{\mu}_{12}, \tilde{\Sigma}_{12})-
        \frac{1}{n}\sum_{i=1}^n \ell_i^{12}(\sigma_{12};\mu^*_{12}, \Sigma_{12})
        \right| > \xi\right\}  \\ 
        & =
    \limsup_{n\rightarrow \infty} P_{\theta^*}\left\{\left|h_n(\tilde{\theta}_{12}) - h_n(\theta_{12}^{**})\right| > \xi\right\} \\ 
    & \leq  \limsup_{n\rightarrow \infty}  P_{\theta^*}\left\{L\left\|\tilde{\theta}_{12} - \theta_{12}^{**}\right\| > \xi\right\} + 
    \limsup_{n\rightarrow \infty}P_{\theta^*}\left\{\tilde{\theta}_{12}\not\in U\right\} = 0,
\end{align*}
where the last equality holds since $\tilde{\theta}_{12}$ is consistent for $\theta^{**}_{12}$. The convergence of part (b) is implied by Kolmogorov's strong law for independent but not identically distributed random variable series (Theorem 7.3.3 of \cite{resnick2019probability}), which requires that 
$$\frac{1}{n^2}\sum_{i=1}^n \Var( \ell_i^{12}(\sigma_{12};\mu^*_{12}, \Sigma_{12})) <\infty.$$
Such a requirement can be guaranteed by the bounded parameter assumption (Assumption \ref{as:boundedparam}), bounded design matrix assumption (Assumption \ref{as:designmat}) and correlation  assumption (Assumption \ref{as:corr}), which guarantees that $\ell_i^{12}(\sigma_{12};\mu^*_{12}, \Sigma_{12})$ is uniformly bounded for every $i$, and hence $\Var( \ell_i^{12}(\sigma_{12};\mu^*_{12}, \Sigma_{12}))$ is also uniformly bounded for every $i$. Since in \eqref{eq:ln12bar} part (a) converges to zero and part (b) is convergent, putting together we see that $\bar{\ell}_{n12}(\sigma_{12})$ converges to some limit in $P_{\theta^*}$ for every $\sigma_{12}\in(-1, 1)$.
We denote such pointwise limit by $\bar{\ell}_{12}(\sigma_{12})$. From equation (\ref{eq:ln12bar}) and the reasoning above, we have
\begin{align*}
\bar{\ell}_{12}(\sigma_{12}) & =
\lim_{n \rightarrow \infty}\bar{\ell}_{n12}(\sigma_{12}) \\
& = \underbrace{-\lim_{n \rightarrow \infty}\left[\frac{1}{n}\sum_{i=1}^n \ell_i^{12}(\sigma_{12}; \tilde{\mu}_{12}, \tilde{\Sigma}_{12})-
\frac{1}{n}\sum_{i=1}^n \ell_i^{12}(\sigma_{12};\mu^*_{12}, \Sigma_{12})
\right]}_\textrm{part (a)}
-
\underbrace{\lim_{n \rightarrow \infty}\frac{1}{n}\sum_{i=1}^n \ell_i^{12}(\sigma_{12};\mu^*_{12}, \Sigma_{12})}_\textrm{part (b)}. \\
& = 0 \ - \ \lim_{n \rightarrow \infty}\frac{1}{n}\sum_{i=1}^n \ell_i^{12}(\sigma_{12};\mu^*_{12}, \Sigma_{12}) \\
& = - \ \lim_{n \rightarrow \infty}\frac{1}{n}\sum_{i=1}^n E_{\theta^*} \left\{\ell_i^{12}(\sigma_{12};\mu^*_{12}, \Sigma_{12}) \right\} = -\lim_{n \rightarrow \infty} \frac{1}{n} \sum_{i=1}^n E_{\theta^*}\left( \ell_i^{12} \right)
\end{align*}
where we write $\ell_i^{12}$ for $\ell_i^{12}(\sigma_{12};\mu^*_{12}, \Sigma_{12})$. Let $\overline{\Sigma}^*_{jk}  = \{(1 + x_i^\T H_j x_i,  r_{ij}r_{ik}\sigma^*_{jk})^\T;  (r_{ij}r_{ik}\sigma^*_{jk}, 1 + x_i^\T H_k x_i)^\T\}$. We write $\Phi_{\overline{\Sigma}_{jk}}$ and $\phi_{\overline{\Sigma}_{jk}}$ for $\Phi_{\overline{\Sigma}_{jk}}(r_i \odot \tilde{\mu}_{jk})$ and  $\phi_{\overline{\Sigma}_{jk}}(r_i \odot \tilde{\mu}_{jk})$, respectively.


Similar to the proof of Theorem \ref{thm.fiststagebvm}, we need to verify certain conditions so that the conclusion of the theorem holds. In the current context this involves verifying the following conditions - {\bf A)}  $\bar{\ell}'''_{n12}(\sigma_{12})$ is uniformly bounded for $\sigma_{12} \in E \subset (-1,1)$ where $E\subset \RR^p$ is some open, convex and bounded set and $\sigma_{12}^* \in E$, {\bf B)} $\bar{\ell}''_{12}(\sigma_{12}^*) = R_{12} > 0$, {\bf C)} Each $\bar{\ell}_{n12}(\sigma_{12})$ is convex in $E$ and {\bf D)} $\bar{\ell}'_{12}(\sigma_{12}^*) = 0$.

{\bf Condition A}: 
The analysis in verifying this condition is similar to its counterpart in the proof of Theorem \ref{thm.fiststagebvm}. We define $\kappa_{i12} = r_{i1}r_{i2}\sigma_{12}$ and write $\bar{\ell}'''_{n12}(\sigma_{12})$ as
$$
\bar{\ell}'''_{n12}(\sigma_{12}) = -\frac{1}{n}\sum_{i=1}^n \frac{\partial^3\ell_i^{12}(\sigma_{12}; \tilde{\mu}_{12}, \tilde{\Sigma}_{12})}
{\partial \sigma_{12}^3} = -\frac{1}{n}\sum_{i=1}^n \frac{\partial^3\ell_i^{12}(\sigma_{12}; \tilde{\mu}_{12}, \tilde{\Sigma}_{12})}
{\partial \kappa_{i12}^3}r_{i1}^3 r_{i2}^3
$$
where
$$
\frac{\partial^3\ell_i^{12}(\sigma_{12}; \tilde{\mu}_{12}, \tilde{\Sigma}_{12})}
{\partial \kappa_{i12}^3} 
= \frac{U\left\lbrace\Phi_{\bar{\Sigma}_{12}}(\kappa_{i12}), \phi_{\bar{\Sigma}_{12}}(\kappa_{i12}),
\phi'_{\bar{\Sigma}_{12}}(\kappa_{i12}),
\phi''_{\bar{\Sigma}_{12}}(\kappa_{i12})
\right\rbrace}{\Phi^4_{\bar{\Sigma}_{12}}(\kappa_{i12})}
.$$
In the preceding display $U\left\lbrace\Phi_{\bar{\Sigma}_{12}}(\kappa_{i12}), \phi_{\bar{\Sigma}_{12}}(\kappa_{i12}),
\phi'_{\bar{\Sigma}_{12}}(\kappa_{i12}),
\phi''_{\bar{\Sigma}_{12}}(\kappa_{i12})
\right\rbrace$ involves polynomial functions of $\Phi_{\bar{\Sigma}_{12}}(\kappa_{i12}), \phi_{\bar{\Sigma}_{12}}(\kappa_{i12}),
\phi'_{\bar{\Sigma}_{12}}(\kappa_{i12}),
\phi''_{\bar{\Sigma}_{12}}(\kappa_{i12})$. Because of Assumption \ref{as:corr}, $\sigma_{12}$ is bounded away from both $-1$ and $1$, hence, $\bar{\Sigma}_{12}$ is strictly positive definite with probability $1$. Therefore, $\Phi_{\bar{\Sigma}_{12}}(\kappa_{i12})$ is bounded away from $0$ for every $\sigma_{12}\in E$ and $1\leq i\leq n$. Assumption \ref{as:boundedparam} and \ref{as:designmat} imply that $\tilde{\mu}_{12}$ is bounded, so that $\Phi_{\bar{\Sigma}_{12}}(\kappa_{i12}), \phi_{\bar{\Sigma}_{12}}(\kappa_{i12}),
\phi'_{\bar{\Sigma}_{12}}(\kappa_{i12}),
\phi''_{\bar{\Sigma}_{12}}(\kappa_{i12})$ are  bounded  for every $\sigma_{12}\in E$ and $1\leq i\leq n$ since $\bar{\Sigma}_{12}$ is strictly positive definite with probability $1$. Therefore, $U\left\lbrace\Phi_{\bar{\Sigma}_{12}}(\kappa_{i12}), \phi_{\bar{\Sigma}_{12}}(\kappa_{i12}),
\phi'_{\bar{\Sigma}_{12}}(\kappa_{i12}),
\phi''_{\bar{\Sigma}_{12}}(\kappa_{i12})
\right\rbrace$ is also uniformly bounded  in $i$. Since $\Phi^4_{\bar{\Sigma}_{12}}$ is bounded away from $0$ uniformly and $U(\cdot)$ is uniformly bounded in $i$, $\frac{\partial^3\ell_i^{12}(\sigma_{12}; \tilde{\mu}_{12}, \tilde{\Sigma}_{12})}
{\partial \kappa_{i12}^3}$ is uniformly bounded and thus $\bar{\ell}'''_n(\sigma_{12})$ is uniformly bounded in $n$.

{\bf Condition B}:  Define $\widehat{\Sigma}_{12} = \{(1, r_{i1}r_{i2}\sigmatru)^\T; (r_{i1}r_{i2}\sigmatru, 1)^\T\}$.  We have
\begin{align*}
    \bar{\ell}_{12}''(\sigma_{12}^*) & = R_{12} = -\lim_{n \rightarrow \infty} \frac{1}{n}\sum_{i=1}^n \frac{\partial^2{\ell^{12}_{i}}}{\partial{\sigma^2_{12}}} \bigg|_{\sigma_{12}=\sigma^*_{12}}
    = -\lim_{n \rightarrow \infty} \frac{1}{n}\sum_{i=1}^n \frac{\partial^2\ell_i^{12}(\sigma_{12}; \tilde{\mu}_{12}, \tilde{\Sigma}_{12})}
{\partial \sigma_{12}^2}\bigg|_{\sigma_{12}=\sigma^*_{12}} 
\\
   &  = -\lim_{n \rightarrow \infty} \frac{1}{n}\sum_{i=1}^n \frac{\partial^2\ell_i^{12}(\sigma_{12}; \mu^*_{12}, \Sigma^*_{12})}
    {\partial \sigma_{12}^2} \stackrel{(a)}{=} 
    -\lim_{n \rightarrow \infty} \frac{1}{n} \sum_{i=1}^n  E_{\theta^*}\bigg\{ \frac{\partial^2\ell_i^{12}(\sigma_{12}; \mu^*_{12}, \Sigma^*_{12})}
    {\partial \sigma_{12}^2}\bigg\}
\\
   & = \lim_{n \rightarrow \infty}\frac{1}{n}  \sum_{i=1}^n \Var_{\theta^*}\bigg\{ \frac{\partial\ell_i^{12}(\sigma_{12}; \mu^*_{12}, \Sigma^*_{12})}
    {\partial \sigma_{12}}\bigg\}
   =  \lim_{n \rightarrow \infty} \frac{1}{n} \sum_{i=1}^n \Var_{\theta^*}\bigg(\frac{\phi_{\widehat{\Sigma}_{12}}}{\Phi_{\widehat{\Sigma}_{12}}}r_{i1}r_{i2}\bigg),
\end{align*}
where 
$\Var_{\theta^*}\bigg(\frac{\phi_{\widehat{\Sigma}_{12}}}{\Phi_{\widehat{\Sigma}_{12}}}r_{i1}r_{i2}\bigg) = \sum_{r_{i1}=-1,1}\sum_{r_{i2}=-1,1}\frac{\phi^2_{\widehat{\Sigma}_{12}}}{\Phi_{\widehat{\Sigma}_{12}}}$ and (a) is because of Kolmogorov's strong law. Because $\sigma_{12}\in(-1,1)$ and $\mu^*_{12}$ is uniformly bounded for every $i$, $\frac{\phi^2_{\widehat{\Sigma}_{12}}}{\Phi_{\widehat{\Sigma}_{12}}}$ is uniformly bounded away from 0 for every $i$. Therefore, $R_{12} = \lim_{n \rightarrow \infty}\frac{1}{n} \sum_{i=1}^n \Var_{\theta^*}\left(\frac{\phi_{\widehat{\Sigma}_{12}}}{\Phi_{\widehat{\Sigma}_{12}}} r_{i1}r_{i2}\right)>0.$

{\bf Condition C}:
It suffices to prove $\bar{\ell}_{n12}(\sigma_{12})$ is convex in $E$ with probability arbitrarily close to 1 for all sufficiently large $n$, since we are proving convergence in $P_{\theta^*}$. First, we have
$$
\bar{\ell}''_{n12}(\sigma_{12}) = -\frac{1}{n}\sum_{i=1}^n \frac{\partial^2\ell^{12}_i}{\partial \sigma^2_{12}} (\sigma_{12})
= -\frac{1}{n}\sum_{i=1}^n \left[\frac{\frac{\partial\phi_{\bar{\Sigma}_{12}} }{\partial \sigma_{12}}}{\Phi_{\bar{\Sigma}_{12}}} - \left(\frac{\phi_{\bar{\Sigma}_{12}}}{\Phi_{\bar{\Sigma}_{12}}}\right)^2\right].
$$
We will prove the first term in the preceding display converges to $0$ when $\theta^{**}_{12} = \theta^*_{12}$ in $P_{\theta^*}$, i.e., 
$-\frac{1}{n}\sum_{i=1}^n \frac{\left(\partial \phi_{\bar{\Sigma}_{12}}/\partial \sigma_{12}\right)}{\Phi_{\bar{\Sigma}_{12}}} \stackrel{P_{\theta^*}}{\rightarrow} 0.$ Recall that $\theta^{**}_{12} =(\beta_1^{*T}, \beta_2^{*T}, vec^T(\Sigma_{12}))$, and $\theta^{*}_{12} =(\beta_1^{*T}, \beta_2^{*T}, vec^T(\Sigma^*_{12})).$
We write $T_i(\tilde{\theta}_{12}) = \frac{\left(\partial \phi_{\bar{\Sigma}_{12}}/\partial \sigma_{12}\right)}{\Phi_{\bar{\Sigma}_{12}}}$ for every $i$ so that we will need to prove that when $\theta^{**}_{12} = \theta^*_{12}$,
$
\frac{1}{n}\sum_{i=1}^n T_i(\tilde{\theta}_{12}) \stackrel{P_{\theta^*}}{\rightarrow} 0
$
for every $\sigma_{12}\in (-1,1)$. To that end, we write 
\begin{align*}
    \frac{1}{n}\sum_{i=1}^n T_i(\tilde{\theta}_{12}) & =
    \underbrace{\left\{\frac{1}{n}\sum_{i=1}^n T_i(\tilde{\theta}_{12}) -
    \frac{1}{n}\sum_{i=1}^n T_i(\theta^{*}_{12})\right\}}_\textrm{part (a)} 
     + \underbrace{\frac{1}{n}\sum_{i=1}^n T_i(\theta^{*}_{12})}_\textrm{part (b)} 
\end{align*}
and prove part (a) and (b) in the preceding display all converge to $0$ in $P_{\theta^*}$. Fix a neighborhood $U$ of $\theta^{*}_{12}$. Then for every $\xi>0$ we have
\begin{align}
&
    \limsup_{n\rightarrow \infty} P_{\theta^*}\left\{\left|\frac{1}{n}\sum_{i=1}^n T_i(\tilde{\theta}_{12}) -
    \frac{1}{n}\sum_{i=1}^n T_i(\theta^{*}_{12})\right| > \xi\right\} \nonumber \\
& \leq 
\limsup_{n\rightarrow \infty} P_{\theta^{*}}\left\{ L \left\|\tilde{\theta}_{12} - \theta^{*}_{12}\right\| > \xi\right\} + \limsup_{n\rightarrow \infty} P_{\theta^*}\left\{ \tilde{\theta}_{12} \not\in U \right\} \label{eq:ulln} = 0
\end{align}
The preceding display is implied by the fact that $\tilde{\theta}_{12}$ is a consistent estimator of $\theta^*_{12}$. Thus, part (a) converges to $0$ in $P_{\theta^*}$.  To prove part (b) converges to $0$ in $P_{\theta^*}$, it can be easily verified that $E_{\theta^*}\left[T_i(\theta^{**}_{12})\right] = 0$. Then, by the Kolmogorov's strong law, we have
$
\frac{1}{n}\sum_{i=1}^n T_i(\theta^{**}_{12}) \stackrel{P_{\theta^*}}{\rightarrow} 0.
$
Combining part (a) and part (b), when $\theta^{**}_{12} = \theta^*_{12}$,
$-\frac{1}{n}\sum_{i=1}^n \frac{\frac{\partial \phi_{\bar{\Sigma}_{12}}}{\partial \sigma_{12}}}{\Phi_{\bar{\Sigma}_{12}}} \stackrel{P_{\theta^*}}{\rightarrow} 0.$
The $L$-equi-Lipschitz property also guarantees that for any $\epsilon>0$, if we choose $E$ to be small enough, for any $\sigma_{12}\in E$,
\begin{equation}
    \lim_{n\rightarrow\infty}\frac{1}{n}\sum_{i=1}^n -\frac{\left(\partial \phi_{\bar{\Sigma}_{12}}/\partial \sigma_{12}\right)}{\Phi_{\bar{\Sigma}_{12}}} \ \text{ exists and }
\lim_{n\rightarrow\infty}\frac{1}{n}\sum_{i=1}^n -\frac{\left(\partial \phi_{\bar{\Sigma}_{12}}/\partial \sigma_{12}\right)}{\Phi_{\bar{\Sigma}_{12}}} > -\epsilon \label{eq:pn12first}
\end{equation}
in $P_{\theta^*}$. From Assumption \ref{as:boundedparam}, \ref{as:corr} and \ref{as:designmat}, we also have that $\frac{\phi_{\bar{\Sigma}_{12}}}{\Phi_{\bar{\Sigma}_{12}}}$ is bounded away from $0$ uniformly for every $i$ and so is $
\frac{1}{n}\sum_{i=1}^n \left(\frac{\phi_{\bar{\Sigma}_{12}}}{\Phi_{\bar{\Sigma}_{12}}}\right)^2
$. We can choose a small enough $\epsilon$ such that for all sufficiently large $n$,
\begin{equation}
    \frac{1}{n}\sum_{i=1}^n \left(\frac{\phi_{\bar{\Sigma}_{12}}}{\Phi_{\bar{\Sigma}_{12}}}\right)^2 > 2\epsilon \label{eq:pn12second}
\end{equation}
Combining (\ref{eq:pn12first}) and (\ref{eq:pn12second}), with probability arbitrary closed to 1 and all $n$ sufficiently large, we have
$$
\bar{\ell}_{n12}''(\sigma_{12}) = -\frac{1}{n}\sum_{i=1}^n \left[\frac{\frac{\partial\phi_{\bar{\Sigma}_{12}} }{\partial \sigma_{12}}}{\Phi_{\bar{\Sigma}_{12}}} - \left(\frac{\phi_{\bar{\Sigma}_{12}}}{\Phi_{\bar{\Sigma}_{12}}}\right)^2\right] \geq 2\epsilon-\epsilon = \epsilon >0,
$$
for every $\sigma_{12}\in E$ which implies that with probability arbitrary closed to $1$, $\bar{\ell}_{n12}(\sigma_{12})$ is convex for all sufficiently large $n$. 

{\bf Condition D}: Write $\widehat{\Sigma}_{12} = \{(1, r_{i1}r_{i2}\sigmatru)^\T; (r_{i1}r_{i2}\sigmatru, 1)^\T\}$ and recall the definition of $\bar{\Sigma}_{12}$ from Section \ref{sec:theory}. Then we have
\begin{align}
 E_{\theta^*}\bigg(\frac{\phi_{\widehat{\Sigma}_{12}}}{\Phi_{\widehat{\Sigma}_{12}}}r_{i1}r_{i2}\bigg) 
    = \sum_{r_{i1}=-1,1}\sum_{r_{i2}=-1,1} r_{i1}r_{i2}\,\phi_{\widehat{\Sigma}_{12}}(r_{i1}x_i^T\beta^*_1, r_{i2}x_i^T\beta^*_2) = 0. \nonumber
\end{align}

Therefore, 
\begin{align*}
\bar{\ell}_{12}'(\sigma_{12}^*) & = \lim_{n \rightarrow \infty} -\frac{1}{n}\sum_{i=1}^n \bigg(\frac{\phi_{\bar{\Sigma}_{12}}}{\Phi_{\bar{\Sigma}_{12}}} r_{i1}r_{i2}\bigg)
= \lim_{n \rightarrow \infty} -\frac{1}{n}\sum_{i=1}^n \bigg(\frac{\phi_{\widehat{\Sigma}_{12}}}{\Phi_{\widehat{\Sigma}_{12}}}r_{i1}r_{i2}\bigg) \\
 & = \lim_{n \rightarrow \infty} -\frac{1}{n}\sum_{i=1}^n E_{\theta^*}\bigg(\frac{\phi_{\widehat{\Sigma}_{12}}}{\Phi_{\widehat{\Sigma}_{12}}}r_{i1}r_{i2}\bigg) = 0. \\
\end{align*}

\subsection{Asymptotic normality of $\tilde{\sigma}_{jk}$}\label{sec:normality_of_two_stage}
 In Lemma \ref{prop:twostageestimator} we prove that the two-step M-estimator asymptotically has a Gaussian distribution. One critical step of Lemma \ref{prop:twostageestimator} is to use the asymptotic normality of the score functions of both stages, i.e. 
\begin{equation}\label{eq:multclt}
	\left[\begin{array}{c}
	\frac{1}{\sqrt{n}} \sum_{i=1}^{n} \nabla_{\beta_1} \ell_{i}^1( x_i^T\beta_1)\big|_{\beta_1 = \beta_1^*} \\
	\frac{1}{\sqrt{n}} \sum_{i=1}^{n} \nabla_{\beta_2} \ell_{i}^2( x_i^T\beta_2)\big|_{\beta_2 = \beta_2^*} \\
	\frac{1}{\sqrt{n}} \sum_{i=1}^{n} \nabla_{\sigma_{12}} \ell_{i}^{12}(\sigma_{12}; \mu_{12}; \Sigma_{12})\big|_{\sigma_{12}=\sigma_{12}^*, \beta_1 = \beta_1^*, \beta_2 = \beta_2^*}
	\end{array}\right] \stackrel{D}{\rightarrow} \Gauss\left( 0, 
	\left[\begin{array}{lll}
	R_{1} & V^{12} & R^{12}_1 \\
	V^{12^T} & R_{2} & R^{12}_2 \\
	R^{12^T}_1 & R^{12^T}_2 & R_{12} \\
	\end{array}\right]\right).
\end{equation}
Such asymptotic normality is implied by the multivariate (Lyapunov) Central Limit Theorem (CLT) since each term in the summation is independent and has a bounded third order moment implied by Assumption \ref{as:designmat}.

\begin{proof}
	We prove the result for $(j, k) = (1,2)$. Recall that $\tilde{\beta}_j$ and $\tilde{\sigma}_{12}$ are the MLE estimators of $\beta_j$ and $\sigma_{12}$ respectively. We first write down the estimating equations for both stages: 
	\begin{equation}\label{eq:estequsys}
		 \begin{cases}
		\sum_{i=1}^n \nabla_{\beta_1} \ell_{i}^1( x_i^T\beta_1)\big|_{\beta_1 = \tilde{\beta}_1} =0 &\text{first stage estimation of $\beta_1$,}\\
		\sum_{i=1}^n \nabla_{\beta_2} \ell_{i}^2( x_i^T \beta_2)\big|_{\beta_2 = \tilde{\beta}_2} = 0 &\text{first stage estimation of $\beta_2$,} \\
		\sum_{i=1}^n \nabla_{\sigma_{12}} \ell^{12}_{i}(\sigma_{12}; \tilde{\mu}_{12} ;\tilde{\Sigma}_{12})\big|_{\sigma_{12} = \tilde{\sigma}_{12}} =0 & \text{second stage.}
		\end{cases}
	\end{equation}
	
	Let $\mu_{12}^* = (x_i^\T\beta_1^*, x_i^\T \beta_2^*)$ and $\Sigma_{12}^* = \{(1, \sigmatru)^\T; (\sigmatru, 1)^\T\}$.
	Then expand the third equation of  \eqref{eq:estequsys} at $(\mu_{12}^*,\Sigma_{12}^*)$.
	\begin{align}\label{eq:sndstagetaylorex}
	    -\frac{1}{\sqrt{n}} \sum_{i=1}^{n} \nabla_{\sigma_{12}} \ell^{12}_{i}(\sigma_{12}^*; \mu_{12}^*, \Sigma_{12}^*) & = \frac{1}{n} \sum_{i=1}^n \sndpartial{\ell^{12}_{i}}{\sigma_{12}}{\beta_{1}^T}(\bar{\theta}_{12}) \sqrt{n}(\tilde{\beta}_1-\beta_{1}^*) \nonumber \\
	    & + \frac{1}{n} \sum_{i=1}^n \sndpartial{\ell^{12}_{i}}{\sigma_{12}}{\beta_{2}^T}(\bar{\theta}_{12}) \sqrt{n}(\tilde{\beta}_2-\beta_{2}^*) \nonumber \\
	    & + \frac{1}{n} \sum_{i=1}^n \sndpartial{\ell^{12}_{i}}{\sigma_{12}}{\sigma_1^2}(\bar{\theta}_{12}) \sqrt{n}(1+\hat{\delta}_{1i}-1) \nonumber\\
	    & + \frac{1}{n} \sum_{i=1}^n \sndpartial{\ell^{12}_{i}}{\sigma_{12}}{\sigma_1^2}(\bar{\theta}_{12}) \sqrt{n}(1+\hat{\delta}_{2i}-1) \nonumber \\
	    & + \frac{1}{n} \sum_{i=1}^n \frac{\partial^2{\ell^{12}_{i}}}{\partial{\sigma_{12}}^2}(\bar{\theta}_{12}) \sqrt{n}(\tilde{\sigma}_{12}-\sigma_{12}^*),
	\end{align}
	where $\hat{\delta}_{1i} = x_i^\T H_1 x_i$ , $\hat{\delta}_{2i} = x_i^\T H_2 x_i$ and $\bar{\theta}_{12}$ lies between $\tilde{\theta}_{12}$ and $\theta_{12}^*$. We note here that since both $\hat{\delta}_{1i}$ and $\hat{\delta}_{1i}$ are  of order $O_{P_{\theta^*}}(n^{-1})$, the third and fourth term of the right hand side of \eqref{eq:sndstagetaylorex} are $O_{P_{\theta^*}}(n^{-1/2})$. Hence, rearranging the terms in \eqref{eq:sndstagetaylorex} and recalling the definition in \eqref{eq.asymptotic.quantities}, asymptotically we have, 
	\begin{align}\label{eq:asymequi}
		\sqrt{n}(\tilde{\sigma}_{12} - \sigma_{12}^*) & = -R_{12}^{-1}  \frac{1}{\sqrt{n}} \sum_{i=1}^{n} \nabla_{\sigma_{12}} \ell^{12}_{i}(\sigma_{12}; \mu_{12}, \Sigma_{12})\big|_{\sigma_{12}=\sigma_{12}^*, \beta_1 = \beta_1^*, \beta_2 = \beta_2^*} \nonumber\\
		& + R_{12}^{-1}Q_1^{12}R_1^{-1}  \frac{1}{\sqrt{n}} \sum_{i=1}^{n} \nabla_{\beta_1} \ell_{i}^1( x_i^T\beta_1)\big|_{\beta_1 = \beta_1^*} \nonumber\\
		& + R_{12}^{-1}Q_2^{12}R_2^{-1}  \frac{1}{\sqrt{n}} \sum_{i=1}^{n} \nabla_{\beta_2} \ell_{i}^2( x_i^T\beta_2)\big|_{\beta_2 = \beta_2^*}.
	\end{align}
We also have from a simple application of the multivariate CLT that 
	\begin{equation}
			\left[\begin{array}{c}
	\frac{1}{\sqrt{n}} \sum_{i=1}^{n} \nabla_{\beta_1} \ell_{i}^1( x_i^T\beta_1)\big|_{\beta_1 = \beta_1^*} \\
	\frac{1}{\sqrt{n}} \sum_{i=1}^{n} \nabla_{\beta_2} \ell_{i}^2( x_i^T\beta_2)\big|_{\beta_2 = \beta_2^*} \\
	\frac{1}{\sqrt{n}} \sum_{i=1}^{n} \nabla_{\sigma_{12}} \ell_{i}^{12}(\sigma_{12}; \mu_{12}; \Sigma_{12})\big|_{\sigma_{12}=\sigma_{12}^*, \beta_1 = \beta_1^*, \beta_2 = \beta_2^*}
	\end{array}\right] \stackrel{D}{\rightarrow} \Gauss \left( 0,
	\left[\begin{array}{lll}
	R_{1} & V^{12} & R^{12}_1 \\
	V^{12^T} & R_{2} & R^{12}_2 \\
	R^{12^T}_1 & R^{12^T}_2 & R_{12} \\
	\end{array}\right] \right) \nonumber
	\end{equation}
As a result, the right hand side of \eqref{eq:asymequi} converges to $\Gauss(0, \tau_{12})$ in distribution, where
\begin{align*}
 \tau_{12} = R_{12}^{-1} & + \underbrace{R_{12}^{-1}\left(- 2R^{12}_1 R^{-1}_1 Q^{12}_1 + Q^{12}_1R_1^{-1}Q^{12}_1\right) R^{-1}_{12}}_\textrm{extra variance from estimating $\beta_{1}$} \\
 & +  \underbrace{R_{12}^{-1}\left(-2R^{12}_{2} R^{-1}_2 Q^{12}_2  + Q^{12}_2R_2^{-1}Q^{12}_2\right) R^{-1}_{12}}_\textrm{extra variance from estimating $\beta_{2}$} \\
 & +  \underbrace{2R_{12}^{-1}Q^{12}_1R_1^{-1}V^{12}R_2^{-1}Q^{'12}_2R_{12}^{-1}}_\textrm{cross covariance between estimating $\beta_1$ and $\beta_2$}
\end{align*}
From Appendix \ref{sec.quantities}, we have $R^{12}_1=R^{12}_2=0$. Therefore,
\begin{align*}
 \tau_{12} = R_{12}^{-1} & + \underbrace{R_{12}^{-1}\left( Q^{12}_1R_1^{-1}Q^{'12}_1\right) R^{-1}_{12}}_\textrm{extra variance from estimating $\beta_{1}$} 
 +  \underbrace{R_{12}^{-1}\left(  Q^{12}_2R_2^{-1}Q^{'12}_2\right) R^{-1}_{12}}_\textrm{extra variance from estimating $\beta_{2}$} \\
 & +  \underbrace{2R_{12}^{-1}Q^{12}_1R_1^{-1}V^{12}R_2^{-1}Q^{'12}_2R_{12}^{-1}}_\textrm{cross covariance between estimating $\beta_1$ and $\beta_2$}
\end{align*}
\end{proof}

\subsection{Proof of Theorem \ref{thm.coverage}}\label{sec:snd_stg_coverage}
We will first show that $\sqrt{n}(\hat{\sigma}_{jk} - \sigma_{jk}^*)$ and $\sqrt{n}(\tilde{\sigma}_{jk} - \sigma_{jk}^*)$ have the same asymptotic distribution. To show this, it suffices to show $\sqrt{n}(\tilde{\sigma}_{jk} - \hat{\sigma}_{jk})\stackrel{P_{\theta^*}}{\rightarrow} 0$. To this end, let $t = \sqrt{n}(\sigma_{jk}-\tilde{\sigma}_{jk})$, then
\begin{align*}
    \hat{\sigma}_{jk} & = \int \sigma_{jk} \Pi_{jk}^*(\sigma_{jk} \mid y, X)d\sigma_{jk} 
    =\int \frac{1}{\sqrt{n}}\left(\tilde{\sigma}_{jk} + \frac{t}{\sqrt{n}}\right) \Pi_{jk}^*(\tilde{\sigma}_{jk} + \frac{t}{\sqrt{n}} \mid y, X)dt \\
    & =\int \left(\tilde{\sigma}_{jk} + \frac{t}{\sqrt{n}}\right) g_{njk}^*(t)dt = \tilde{\sigma}_{jk} + \int \frac{t}{\sqrt{n}} g_{njk}^*(t)dt
\end{align*}
where the last equality holds since  $g^*_{njk}$ is the posterior density of $t$. Since $\int t \phi_{R^{-1}_{jk}}(t)dt = 0$, we have
\begin{align}\label{eq:1tv_convergence}
    \sqrt{n}|\tilde{\sigma}_{jk} - \hat{\sigma}_{jk}| & = |\int t g_{njk}^*(t)dt| = |\int t g_{njk}^*(t)dt - \int t \phi_{R^{-1}_{jk}}(t)dt| \nonumber \\
    & \leq \int |t| |g_{njk}^*(t) - \phi_{R^{-1}_{jk}}(t)|dt 
\end{align}
We have shown in \Cref{prop:sndstagebvm} that $\int |g_{njk}^*(t) - \phi_{R^{-1}_{jk}}(t)|dt \stackrel{P_{\theta^*}}{\rightarrow} 0$. Following \cite[Theorem 8.2, page 489]{lehmann2006theory}, an extra condition required for $\int |t| |g_{njk}^*(t) - \phi_{R^{-1}_{jk}}(t)|dt \overset{P_{\theta_0}}{\rightarrow} 0$ is the boundedness of the prior mean which in our case is trivially satisfied due to the fact that $\sigma_{jk} \in [-1,1]$. 

Noting that $\sqrt{n}(\hat{\sigma}_{jk} - \sigma_{jk}^*)$ and $\sqrt{n}(\tilde{\sigma}_{jk} - \sigma_{jk}^*)$ have the same asymptotic distribution, we will use $\sqrt{n}(\tilde{\sigma}_{jk} - \sigma_{jk}^*)$ in place of $\sqrt{n}(\hat{\sigma}_{jk} - \sigma_{jk}^*)$ hereafter in this proof. Let $U_n = \sqrt{n}(\tilde{\sigma}_{jk} - \sigma_{jk}^*)$. Let the cdf of $U_n$ be $F_{U_n}(\cdot)$ and $F_{U_n}(\cdot)$ be a continuous function on the real line. From Lemma \ref{prop:twostageestimator}, we have $U_n\stackrel{d}{\rightarrow} \Gauss(0, \tau_{jk})$. Hence, $F_{U_n}(u)\rightarrow \Phi(u/\sqrt{\tau_{jk}})$ as $n\rightarrow\infty$ for every $u\in\RR$. Thus,
	\begin{align*}
	  &  P_{\theta^*}\left(U_n \in \left[-\sqrt{R_{jk}^{-1}}\,v_{1-\frac{\alpha}{2}}, \sqrt{R_{jk}^{-1}}\,v_{1-\frac{\alpha}{2}}\right]\right)  = F_{U_n}\left(\sqrt{R_{jk}^{-1}}\,v_{1-\frac{\alpha}{2}}\right) - F_{U_n}\left(-\sqrt{R_{jk}^{-1}}\,v_{1-\frac{\alpha}{2}}\right) \\
	& \rightarrow
	\Phi\left(\sqrt{R^{-1}_{jk}/\tau_{jk}}\, v_{1-\frac{\alpha}{2}}\right)-\Phi\left(-\sqrt{R^{-1}_{jk}/\tau_{jk}}\, v_{1-\frac{\alpha}{2}}\right) < 1- \alpha,
	\end{align*}
since $\tau_{jk} > R_{jk}^{-1}$.

\section{Two-stage approximate conditional sampler}
We present in Algorithm \ref{algo:empirical_sampler} a fast way to sample from the hierarchical extension developed in Section \ref{sec:hierarchical} of the main document. 
\begin{algorithm}[ht]
\caption{Two-stage approximate conditional sampler to sample from \eqref{eq:jsdm_hierarchy} in main document}

\begin{enumerate}
\item Initialize $(\eta, \Omega)$ and $\omega$.

\item  Given $(\eta, \Omega)$ obtain approximations to $\Pi(\beta_j \mid y, X, \eta, \Omega)$ as $\Gauss(\hat{\beta_j}, H_j)$ for $j = 1, \ldots, q$ replacing the prior in \eqref{eq:first_stage} in the main document by $\Gauss(\eta, \Omega)$.

\item Draw $\beta_j \sim \Gauss(\hat{\beta_j}, H_j)$ independently for $j = 1, \ldots, q$.

\item Update $(\eta, \Omega) \sim \mathrm{NIW}(\eta_q, \nu_q, \delta_q, \Lambda_q)$ where $\nu_q = \nu_0 + q$, $\delta_q = \gamma_0 + q$, $\bar{\beta} = (\sum_{j=1}^q \beta_j)/q$, $\eta_q = (\nu_0 \eta_0 + q \bar{\beta})/\nu_q$, $S = \sum_{j=1}^q (\beta_j - \bar{\beta})(\beta_j - \bar{\beta})^\T$ and $\Lambda_q = \Lambda_0 + S + (\nu_0q/\nu_q)\sum_{j=1}^q (\beta_j - \eta_0)(\beta_j - \eta_0)^\T$.
\item Given $\omega$, obtain approximations to $\Pi(\sigma_{jk} \mid y, X, \omega)$ as  $\Gauss(\hat{\sigma}_{jk}, s_{jk}^2)$ for $j<k = 1, \ldots, q$ with pseudo-priors $\Pi_j^*(\beta_j\mid y, X) \Pi_k^*(\beta_k\mid y, X)$.

\item Draw $\sigma_{jk} \sim \Gauss(\hat{\sigma}_{jk}, s_{jk}^2)$ independently  and set $\gamma_{jk} = 0.5 \log\{(1+\sigma_{jk})/(1-\sigma_{jk})\}$ for $j<k = 1, \ldots, q$.

\item Update $\omega^2 \sim \text{inverse-Gamma}\left( \dfrac{q(q-1)}{2} + a_\omega,\,\, \dfrac{1}{2}\sum_{j<k} \gamma_{jk}^2 + b_\omega\right)$.

\item Repeat Steps 2-7 $T$ times to obtain $T$ samples of $(\eta, \Omega)$ and $\omega$.


\end{enumerate}
\label{algo:empirical_sampler}
\end{algorithm}

\section{Discrepancy between $\tau_{jk}$ and $R_{jk}^{-1}$}\label{sec:tauVR}
Theorem \ref{thm.coverage} shows the equi-tailed credible intervals obtained from the second stage posterior distribution can cause under coverage. The extent of under coverage clearly depends on the ratio $R_{jk}^{-1}/\tau_{jk}$. To that end, consider
\begin{equation}\label{eq:var_ratio}
   	\frac{\tau_{jk}}{R_{jk}^{-1}}  = 1 + \frac{Q_j^{jk}R_j^{-1}Q_j^{{jk}^\T}}{R_{jk}} + \frac{Q_k^{jk}R_{k}^{-1}Q_k^{{jk}^\T}}{R_{jk}} + 2\,\frac{Q_j^{jk}R_j^{-1}V^{jk}R_k^{-1}Q_k^{{jk}^\T}}{R_{jk}}.
\end{equation}
An upper bound on the ratio $\tau_{jk}/R_{jk}^{-1}$ can be obtained by bounding each term in \eqref{eq:var_ratio} separately. Rewrite $Q_j^{jk} = \sum_{i=1}^n ({a_i x_i}/{n})^T$ and $R_j = \sum_{i=1}^n (b_i x_i x_i^T/{n})$ where $a_i = -E_{\theta^*} \left( \dfrac{\partial^2{\ell_{i}^{jk}}}{\partial{\sigma_{jk}}\partial{\beta_j}}\right)$ and $b_i = - E_{\theta^*} \left( \dfrac{\partial^2{\ell_{i}}}{\partial{\beta_j}^2} \right)$.  Then by Lemma \ref{lem.quadbound}, we have 
$
	Q_j^{jk}R_j^{-1}Q_j^{{jk}^\T }\leq p \max_i \frac{a_i^2}{b_i}.
$
By symmetry, $Q_k^{jk}R_j^{-1}Q_k^{{jk}^\T} \leq p \max_i \frac{a_i^2}{b_i}$. Finally, let $$Q_j^{jk}R_j^{-1}  \frac{1}{\sqrt{n}} \sum_{i=1}^{n} \nabla_{\beta_j} \ell_{i}^j( x_i^T\beta^*_j) = W^n_j.$$ Then by the Cauchy-Schwartz inequality we get 
\begin{align*}
2\lim_{n\rightarrow\infty}\Cov(W^n_j, W^n_k) &= 2Q_j^{jk}R_j^{-1}V^{jk}R_k^{-1}Q_k^{{jk}^\T} \leq  2 \lim_{n\rightarrow\infty}
\sqrt{\Var\left(W^n_k\right)\Var\left(W^n_j\right)}\\ &=2\sqrt{Q_j^{jk}R_j^{-1}Q_j^{{jk}^\T}}\sqrt{Q_k^{jk}R_k^{-1}Q_k^{{jk}^\T}} \\
&\leq 2p \max_i \frac{a_i^2}{b_i}.
\end{align*}
Combining the results we get 
\begin{equation}\label{eq:varbound}
      \frac{\tau_{jk}}{R_{jk}^{-1}} = 1 + \frac{Q_j^{jk}R_j^{-1}Q_j^{{jk}^\T}}{R_{jk}} + \frac{Q_k^{jk}R_{k}^{-1}Q_k^{{jk}^\T}}{R_{jk}} + 2\,\frac{Q_j^{jk}R_j^{-1}V^{jk}R_k^{-1}Q_k^{{jk}^T}}{R_{jk}}
      \leq 1 + 4p\frac{\max_i  \frac{a_i^2}{b_i}}{R_{jk}},
\end{equation}
where $\max_i \frac{a_i^2}{b_i}$ is finite due to Assumptions \ref{as:boundedparam} - \ref{as:designmat}. The bound is  proportional to the number of covariates and to a transformation of the true correlation $\sigma_{jk}^*$. 

We end this section with a numerical study investigating the size of the bias in finite samples. In this experiment, we consider sample sizes $n = 200, 500$, fix the number of covariates to $p = 5$ and vary the number of outcomes as $q = 10, 15, 20, 100, 200$. We generate the regression coefficients following the setting of the dense case in Section \ref{sec:simulations}, i.e., $\beta_{lj}^* \sim \Gauss(0,0.5^2)$  and the intercept term $\beta_{0j}^* \sim \Gauss(0,0.5^2)$ for $l=1, \ldots p$, $j = 1, \ldots q$. We adopt the dense correlation matrix design where we set the correlation matrix to be $(1-\rho^*)\mathrm{I}_q + \rho^* \mathbf{1}_q\mathbf{1}_q^\T$ and we vary $\rho^* = 0.3, 0.5, 0.7, 0.9, 0.99, 0.999$. We obtain the equi-tailed credible intervals based on estimates of $R_{jk}^{-1}$ and $\tau_{jk}$, which renders an asymptotically correct coverage.

To compare the discrepancy between these two intervals, we display the maximum ratio of their length, and the average ratio of their length across all pairs of $(j,k)$ in Table \ref{tab:undercoverage}. This discrepancy is most severe in the extremely correlated cases when $\rho^* = 0.999$ or $\rho^* = 0.99$. However, even in these cases, the maximum discrepancy over ${200 \choose 2} \approx 20000$ parameters is roughly $6\%$, and the average discrepancy is less than $1\%$. Furthermore, the two intervals provide almost identical coverage. 
In Table \ref{tab:runtime_table}, we report the runtime for intervals obtained from estimates of $R_{jk}^{-1}$ and $\tau_{jk}$. For a fair comparison, all experiments were run on a 64 bit Intel i7-8700K CPU @3.7 GHz processor. The results indicate that to correct for the variance, $~8-20$ times the computational resources need to be allocated. In view of this, we implement bigMVP without adjusting the second stage variance.
\begin{table}[ht]
    \centering
    \scalebox{0.73}{
    \begin{tabular}{c|ccccccccccccc}
    \hline
    & &  \multicolumn{2}{c}{$\rho^* = 0.3$} & \multicolumn{2}{c}{$\rho^* = 0.5$} & \multicolumn{2}{c}{$\rho^* = 0.7$}  & \multicolumn{2}{c}{$\rho^* = 0.9$}  & \multicolumn{2}{c}{$\rho^* = 0.99$}  & \multicolumn{2}{c}{$\rho^* = 0.999$}\\
    \hline
    & &  MAX & AVG & MAX & AVG & MAX & AVG & MAX & AVG  & MAX & AVG & MAX & AVG  \\
    \hline
    \multirow{5}{*}{$n=200$} & $q=10$ & 1.0028 & 1.0010 & 1.0065 & 1.0030 & 1.0063 & 1.0042 & 1.0116 & 1.0060 & 1.0264 & 1.0086 & 1.0346 & 1.0097  \\
    & $q=20$ & 1.0031 & 1.0008 & 1.0050 & 1.0021 & 1.0112 & 1.0042 & 1.0235 & 1.0089 & 1.0279 & 1.0077 & 1.0330 & 1.0086   \\
    & $q = 50$ & 1.0057 & 1.0013 & 1.0069 & 1.0023 & 1.0195 & 1.0043 & 1.0353 & 1.0083 & 1.0460 & 1.0102 & 1.0440 & 1.0077 \\
    & $q = 100$ & 1.0081 & 1.0010 & 1.0110 & 1.0028 & 1.0292 & 1.0055 & 1.0518 & 1.0082 & 1.0621 & 1.0099 & 1.0652 & 1.0084 \\ 
    & $q = 200$ & 1.0100 & 1.0013 & 1.0115 & 1.0022 & 1.0253 & 1.0043 & 1.0388 & 1.0074 & 1.0688 & 1.0084 & 1.0655 & 1.0089 \\ 
    \hline
    \multirow{5}{*}{$n=500$} & $q=10$ & 1.0024 & 1.0010 & 1.0043 & 1.0020 & 1.0076 & 1.0042 & 1.0096 & 1.0043 & 1.0482 & 1.0126 & 1.0315 & 1.0113  \\
    & $q=20$  & 1.0031 & 1.0009 & 1.0044 & 1.0020 & 1.0086 & 1.0036 & 1.0193 & 1.0046 & 1.0813 & 1.0104 & 1.0801 & 1.0087   \\
    & $q = 50$ & 1.0034 & 1.0010 & 1.0066 & 1.0022 & 1.0128 & 1.0038 & 1.0349 & 1.0053 & 1.0617 & 1.0095 & 1.1116 & 1.0101 \\
    & $q = 100$ & 1.0056 & 1.0008 & 1.0062 & 1.0023 & 1.0213 & 1.0043 & 1.0383 & 1.0054 & 1.1009 & 1.0094 & 1.0941 & 1.0103 \\ 
    & $q = 200$ & 1.0042 & 1.0007 & 1.0076 & 1.0022 & 1.0140 & 1.0038 & 1.0434 & 1.005 & 1.0916 & 1.0105 & 1.1002 & 1.0098 \\
    \hline
    \end{tabular}}
    \caption{Summary of the maximum (MAX) and average ratio (AVG) of the length of credible intervals obtained by the inflated variance $\tau_{jk}$ to the length of credible intervals obtained by the non-inflated variance $R^{-1}_{jk}$.}
    \label{tab:undercoverage}
\end{table}

\begin{table}[ht]
    \centering
    
    \scalebox{0.73}{
    \begin{tabular}{c|ccccccccccccc}
    \hline
    & &  \multicolumn{2}{c}{$\rho^* = 0.3$} & \multicolumn{2}{c}{$\rho^* = 0.5$} & \multicolumn{2}{c}{$\rho^* = 0.7$}  & \multicolumn{2}{c}{$\rho^* = 0.9$}  & \multicolumn{2}{c}{$\rho^* = 0.99$}  & \multicolumn{2}{c}{$\rho^* = 0.999$}\\
    \hline
    & & ORI  & INF & ORI & INF & ORI & INF & ORI  & INF & ORI & INF & ORI & INF  \\
    \hline
    \multirow{5}{*}{$n=200$} & $q=10$ & 0.03 & 0.27 & 0.03 & 0.31 & 0.03 & 0.30 & 0.03 & 0.48 & 0.03 & 0.56 & 0.03 & 0.60  \\
    & $q=20$ & 0.12 & 1.12 & 0.13 & 1.32 & 0.14 & 1.36 & 0.12 & 2.25 & 0.13 & 2.31 & 0.12 & 2.31   \\
    & $q = 50$ & 0.78 & 7.10 & 0.77 & 8.16 & 0.87 & 8.78 & 0.79 & 15.83 & 0.79 & 15.66 & 0.78 & 15.45 \\
    & $q = 100$ & 3.07 & 26.57 & 2.94 & 31.54 & 2.93 & 43.08 & 2.78 & 57.78 & 2.79 & 57.24 & 2.98 & 58.02 \\ 
    & $q = 200$ & 11.52 & 111.79 & 11.92 & 116.43 & 12.17 & 133.9 & 11.29 & 231.98 & 12.21 & 232.92 & 11.18 & 232.54 \\ 
    \hline
    \multirow{5}{*}{$n=500$} & $q=10$  & 0.07 & 0.66 & 0.07 & 0.69 & 0.08 & 0.70 & 0.07 & 1.31 & 0.07 & 1.33 & 0.07 & 1.32  \\
    & $q=20$   & 0.29 & 2.59 & 0.29 & 2.88 & 0.28 & 2.91 & 0.30 & 5.67 & 0.29 & 5.47 & 0.29 & 5.40  \\
    & $q = 50$ & 1.76 & 16.37 & 1.78 & 18.09 & 1.76 & 18.35 & 1.69 & 34.62 & 1.72 & 34.74 & 1.73 & 34.88 \\
    & $q = 100$ & 6.99 & 61.08 & 6.94 & 72.27 & 6.89 & 75.48 & 6.97 & 143.34 & 7.02 & 143.47 & 6.99 & 142.48 \\ 
    & $q = 200$ & 29.03 & 249.49 & 28.14 & 296.77 & 27.96 & 302.74 & 28.11 & 585.61 & 28.18 & 1215.91 & 27.82 & 580.49 \\
    \hline
    \end{tabular}}
    \caption{Comparison of runtime of the bigMVP method based on the original variance (ORI) $R^{-1}_{jk}$ and the inflated variance (INF) $\tau_{jk}$.}
    \label{tab:runtime_table}
\end{table}

\section{Accuracy of posterior approximation}\label{sec:marginal_distributions}
Here, we carry out an extensive simulation study to gauge the effect of replacing a joint prior on $\Sigma$ by the corresponding marginal prior. Recall, the joint posterior distribution under a MVP model is $\Pi(B, \Sigma\mid y, X) \propto \prod_{i=1}^n \text{pr}(z_i \in E_i) \Pi(\Sigma)\Pi(B)$. We assume a product prior for the regression coefficients, that is, $\Pi(B) = \prod_{j=1}^q \Pi(\beta_j)$ where each component is $\Gauss(0, a^2)$. Our main objective of study is the marginal distribution of $\Pi(\sigma_{jk}\mid y, X)$ obtained from a full posterior analysis and the computationally scalable two-stage alternative $\Pi^*(\sigma_{jk}\mid y, X)$ proposed here. Both of these quantities are studied under two different sets of priors for $\Sigma$ - the LKJ($\nu_1$) prior \citep{lewandowski2009generating} and the marginally non-informative prior MNI($\nu_2$) prior \citep{huang2013simple}. As mentioned earlier, when $\nu_1 =1$, the LKJ prior is uniform over the space of $q$-dimensional correlation matrices whereas for $\nu_2 = 2$, the MNI prior provides uniform marginal distributions over individual correlation coefficients. The proposed approximate marginal distribution is obtained by following Algorithm \ref{algo:bigMVP_algo} where in the second stage we set $\Pi_{jk}(\sigma_{jk})$ as either proportional to $\text{Beta}(\nu_1 +q/2 - 1, \nu_1 + q/2 -1)$ or $(1-\sigma_{jk}^2)^{\nu_2/2 - 1}$ according to the joint prior.

We fix the sample size $n = 200, 500$, the number of covariates $p = 5$ and vary the number of binary outcomes $q = 4, 5, 6, 7$. We consider a relatively low number of outcomes as a full MCMC analysis of the joint posterior becomes computationally prohibitive for higher values of $q$. To obtain samples from the full marginal $\Pi(\sigma_{jk}\mid y, X)$, we use a data augmented (DA) MCMC analysis \citep{chib1998analysis} where the latent Gaussian variables are simulated from their full conditional distribution - $z \sim \Gauss(B^\T x, \Sigma)\mathbb{I}_{z\in E}$. We use the sampler from \cite{pakman2014exact} to sample from the truncated Gaussian distribution. Conditional on the latent variables, the regression coefficients can be updated in a block efficiently using vectorization \citep{chakraborty2020bayesian}. For both LKJ($\nu_1$) and MNI($\nu_2$), the matrix $\Sigma$ can be updated from an inverse-Wishart distribution with appropriate parameters.  We simulate the regression coefficients $\beta_{lj} \sim \Gauss(0,1)$ for $l = 1, \ldots, p$ and $j = 1,\ldots, q$. The true correlation structure is set to $(1-\rho^*)\mathrm{I}_q + \rho^* \mathbf{1}_q\mathbf{1}_q^\T$ and we vary $\rho^* = 0, 0.1, 0.3, 0.5, 0.7$. We fix $\nu_1 = 1$ and $\nu_2 = 2$ as these are representative cases of joint non-informativeness and marginal non-informativeness.

Suppose $\hat{\Sigma}_{\text{DA}}$, $\text{var}(\Sigma)_{\text{DA}}$ are the Rao-Blackwellized mean and variance obtained from the samples drawn from a DA sampler and $\hat{\Sigma}$, $\text{var}(\Sigma)$ are the mean and variance obtained by the approximations $\Pi^*(\sigma_{jk}\mid y, X)$ stacked into a $q \times q$ matrix. Then we look at $\norm{\hat{\Sigma}_{\text{DA}} - \hat{\Sigma}}_F/q^2 $ and $\norm{\text{var}(\Sigma)_{\text{DA}} - \text{var}\hat{(\Sigma)}}_F/q^2$ averaged over 30 independent replications. Here, each run of the DA sampler is carried out to ensure an effective sample size \citep{geyer1992practical} of 1000 averaged over the parameters. As can be seen from Table \ref{tab:full_vs_marginal}, the first two moments of marginal posterior distributions estimated via MCMC
with a joint prior 
differ only slightly from the corresponding two-stage approximation obtained by replacing the joint prior by its marginal version. This is true for moderate $q$ and a range of correlation settings although when $q$ is very high, joint non-informativeness of the LKJ(1) prior may incur high penalties on marginal correlations resulting in different conclusions from the two approaches.

\begin{table}
    \centering
    \scalebox{0.6}{
    \begin{tabular}{c|c|cccc|cccc}
    \toprule
       & &\multicolumn{4}{c}{$n = 200$} & \multicolumn{4}{c}{$n = 500$}\\
       \midrule
         & & $q = 4$ & $q = 5$ & $q = 6$ & $q = 7$ &  $q = 4$ & $q = 5$ & $q = 6$ & $q = 7$  \\
         \cmidrule{1-6}  \cmidrule{7-10}
        \multirow{2}{*}{$\rho^* = 0$} & LKJ & (0.02, 0.007) & (0.02, 0.007) & (0.02, 0.005) & (0.02, 0.006) & (0.01, 0.002) & (0.01, 0.002) &(0.01, 0.001) & (0.01, 0.002)\\
        & MNI & (0.03, 0.008)& (0.03, 0.009) & (0.02, 0.005) & (0.03, 0.005) & (0.01, 0.003) & (0.01, 0.003) & (0.01, 0.002) & (0.02, 0.002)\\
        \midrule
        \multirow{2}{*}{$\rho^* = 0.1$} & LKJ & (0.03, 0.009)& (0.02, 0.005) & (0.02, 0.005) & (0.02, 0.004) & (0.01, 0.004) & (0.01, 0.004) &(0.01, 0.004) & (0.01, 0.002)\\
        & MNI  & (0.04, 0.01) & (0.03, 0.007) & (0.03, 0.008) & (0.03, 0.007) & (0.01, 0.004) & (0.01, 0.004) &(0.02, 0.005) & (0.02, 0.003)\\
        \midrule
        \multirow{2}{*}{$\rho^* = 0.3$} & LKJ & (0.04, 0.009) & (0.03, 0.006) & (0.02, 0.005) & (0.02, 0.005) & (0.02, 0.003) & (0.02, 0.001) & (0.01, 0.001) & (0.02, 0.002) \\
        & MNI & (0.05, 0.01) & (0.04, 0.008) & (0.03, 0.007) & (0.03, 0.008) & (0.03, 0.003) & (0.02, 0.001) & (0.01, 0.001) & (0.03, 0.002)\\
        \midrule
        \multirow{2}{*}{$\rho^* = 0.5$} & LKJ & (0.06, 0.006) & (0.05, 0.005) & (0.03, 0.007)& (0.03, 0.004) & (0.03, 0.002) & (0.03, 0.004) & (0.02, 0.003) &(0.01, 0.001)\\
        & MNI & (0.06, 0.007) & (0.06, 0.006) & (0.04, 0.006) & (0.05, 0.006) & (0.03, 0.002) & (0.03, 0.004) & (0.03, 0.003) & (0.03, 0.002)\\
        \midrule
        \multirow{2}{*}{$\rho^* = 0.7$} & LKJ & (0.07, 0.009) & (0.05, 0.005) & (0.04, 0.005) & (0.03, 0.003) & (0.03, 0.003)  &(0.04, 0.002) & (0.03, 0.002) & (0.02, 0.004)\\
        & MNI & (0.08, 0.01) & (0.07, 0.007) & (0.06, 0.007) & (0.05, 0.005) & (0.04, 0.004) & (0.04, 0.001) & (0.03, 0.002) & (0.02, 0.003)\\
        \hline
    \end{tabular}
    }
    \caption{Comparison of the first two moments of $\Pi(\sigma_{jk}\mid y, X)$ and $\Pi^*_{jk}(\sigma_{jk}\mid y, X)$. Here $\Pi(\sigma_{jk}\mid y, X)$ corresponds to the marginal posterior distribution of $\sigma_{jk}$ with $\Pi(\Sigma)$ as the prior on the correlation matrix, and  $\Pi^*_{jk}(\sigma_{jk}\mid y, X)$ is the approximate two-stage marginal posterior obtained by considering the marginal prior $\Pi(\sigma_{jk})$ of $\Pi(\Sigma)$. Two choices of the joint prior are considered - the LKJ(1) prior and the MNI(2) prior.}
    \label{tab:full_vs_marginal}
\end{table}

\section{Comparison with INLA-MCMC}\label{sec:inla-mcmc}
In addition to the competitors above, we consider the INLA-MCMC method proposed in \cite{gomez2017spatial, gomez2018markov}. We treat this separately as the method is not directly applicable to the most general MVP models considered here but can be applied to specific cases. A generic overview of the method is as follows. Suppose we want to sample from the posterior distribution $\Pi(\theta \mid y)$ where $y$ is the observed data and the parameter $\theta$ is a vector of all parameters in the model. INLA provides approximations to marginals of this target posterior, namely $\Pi(\theta_j\mid y)$ using a sequence of Laplace type approximations when the prior distribution of $\theta$ is Gaussian. Here, the parameter $\theta$ may also contain latent Gaussian variables $z$ like we have in the probit model. However, \cite{gomez2017spatial, gomez2018markov} argued that in many practically useful models, INLA is not directly applicable unless some parameters within $\theta$ are fixed. Let $\theta = (\theta_c, \theta_{-c})$ denote the decomposition of $\theta$ such that conditional on $\theta_c$, INLA can be applied to the model. The authors then embed the INLA approximation within a Metropolis-Hastings (MH) algorithm where the parameters $\theta_c$ are updated using an MH step and conditional on $\theta_c$, other relevant posterior quantities are approximated by INLA. The key assumption here is that $\Pi(\theta\mid y) \propto \Pi(y \mid \theta_{c}) \Pi(\theta_{-c}\mid \theta_c)\Pi(\theta_c)$.

We now consider a specific version of the MVP model where we restrict the correlation matrix $\Sigma = (1-\rho)\mathrm{I}_q + \rho \mathbf{1}_q\mathbf{1}_q^\T$ where $\mathbf{1}_q$ is $q$ dimensional vector of 1's. In this particular case, the latent variable $z$ can be alternatively represented as $z = B^\T x + \rho \mathbf{1}_q u + e$ where $u \sim \Gauss(0,1)$ and $e \sim \Gauss(0, 1-\rho)$ independent of $u$. Hence conditional on $u$, univariate probit models suited to INLA can be fitted with the augmented covariates $\tilde{x} = [x, u]$ to each column of $y$. Furthermore, if $\tilde{\beta}_j$ represents the $j$th column of the regression coefficient matrix $\tilde{B}$ for this conditional probit model, then an estimate of $\beta_j$ can be formed by setting $\hat{\beta}_j = (\sqrt{1-\rho^*})\tilde{\beta}_j$. We implement this along with a random walk MH algorithm for updating $u$ where the standard deviation of the proposal distribution is updated following \cite{haario2001adaptive} so that the acceptance rate is optimized to $\sim 0.23$.

We conducted two separate experiments with the true correlation $\rho^* = 0.3, 0.5$ and with each $B^*$ generated according to the ``Dense" case scenario described in Section \ref{sec:simulations} of the main document. For each case we fix the sample size $n = 200$, $p = 3$ and consider $q = 3,5,7,9$; we did not consider higher dimensional responses as the INLA-MCMC method did not scale well computationally with growing $q$. We report the errors in estimating the regression coefficients, namely $\norm{B^* - \hat{B}}_F^2/q^2$, averaged over 30 independent replications where the prior distribution of outcome-specific regression coefficients are $\Gauss(0, 5^2)$ and an LKJ(1) prior on the correlation matrix so that the marginal priors on the correlations are $\text{Beta}(q/2, q/2)$. When executing the rescaling of the coefficients we set $\rho = \rho^*$. We summarize the results in Table \ref{tab:inla_mcmc}. The results show that bigMVP performs better than INLA-MCMC in all the cases considered here although it should be kept in mind that the rescaling of the estimates from INLA-MCMC were computed assuming the truth is known.
\begin{table}
    \centering
    \scalebox{0.7}{
    \begin{tabular}{c|cccccccc}
    \toprule
    & \multicolumn{2}{c}{$q=3$} & \multicolumn{2}{c}{$q=5$} & \multicolumn{2}{c}{$q=7$} & \multicolumn{2}{c}{$q=9$}\\
    \midrule
    & bigMVP & INLA-MCMC & bigMVP & INLA-MCMC & bigMVP & INLA-MCMC & bigMVP & INLA-MCMC \\
    \cmidrule{2-3} \cmidrule{4-5} \cmidrule{6-7} \cmidrule{8-9}
    $\rho^* = 0.3$ & 0.04 & 0.07 & 0.06 & 0.09 & 0.07 & 0.11 & 0.04 & 0.08\\
    $\rho^* = 0.5$ &0.05 & 0.08 & 0.05 & 0.09 & 0.06 & 0.10 & 0.04 & 0.09\\
    \bottomrule
    \end{tabular}
    }
    \caption{Comparison between bigMVP and INLA-MCMC in estimating the regression coefficients when the correlation structure between the binary outcomes are of the form $\Sigma = (1-\rho^*)\mathrm{I}_q + \rho^* \mathbf{1}_q\mathbf{1}_q^\T$. }
    \label{tab:inla_mcmc}
\end{table}

\section{Marginal pairwise prediction}\label{sec:pairwise_prediction}
Although our focus is on inference on the model parameters, it can nonetheless be useful to obtain predictive distributions as key components of model assessments and comparisons.  As obtaining a joint predictive rule for all $q$ outcomes is necessarily computationally demanding when $q$ is large, we focus on prediction of an arbitrary pair of outcomes in a new sample $t$ having covariate value $x_t$.

Fix a pair $(j,k)$. Then from our proposed two-stage approach, we have access to the approximate marginal posterior distributions $\Pi_j^*(\beta_j \mid y, X), \Pi_k^*(\beta_k \mid y, X)$ and $\Pi_{jk}^*(\sigma_{jk} \mid y, X)$. Let $\theta_{jk} = \{\beta_j^\T, \beta_k^\T, vec(\Sigma_{jk})^\T\}$. Given a new test point $x_t$, the marginal predictive distribution for this pair is
\begin{equation}\label{eq:marginal_pairwise}
    p_{jk}(y_t^{(j)}, y_t^{(k)} \mid y, X) = \int p_{jk}(y_t^{(j)}, y_t^{(k)} \mid \theta_{jk}, y, X) \Pi(\theta_{jk} \mid y, X) d \theta_{jk},
\end{equation}
where we approximate the joint posterior $\Pi(\theta_{jk} \mid y, X)$ by a product of the marginals, $\Pi(\theta_{jk} \mid y, X) \approx \Pi_j^*(\beta_j\mid y, X)\Pi_k^*(\beta_k \mid y, X) \Pi_{jk}^*(\sigma_{jk} \mid y, X)$. The first component inside the integral in \eqref{eq:marginal_pairwise} can be written as $p_{jk}(y_t^{(j)}, y_t^{(k)} \mid \theta_{jk}, y, X) = \int p_{jk}(y_t^{(j)}, y_t^{(k)}\mid z_t^{(j)}, z_t^{(k)},\theta_{jk}, y, X) dz_t^{(j)}d z_t^{(k)} $, where $(z_t^{(j)}, z_t^{(k)})$ are the corresponding latent variables for the test point $x_t$. Hence, we focus on the predictive distribution of the latent variables $(z_t^{(j)}, z_t^{(k)})$. Under the MVP model, $(z_t^{(j)}, z_t^{(k)})\mid \theta_{jk}, y, X$ has a bivariate Gaussian distribution with mean $\mu_t = (x_t^\T \beta_j, x_t^\T \beta_k)$, unit variance and correlation $\sigma_{jk}$. Since the posterior distributions of the regression coefficients are Gaussian, we can easily marginalize over them to obtain that $(z_t^{(j)}, z_t^{(k)}) \mid \sigma_{jk} \sim \Gauss(\tilde{\mu}_t, \tilde{\Sigma}_t)$, where $\tilde{\mu}_t = (x_t^\T \hat{\beta}_j, x_t^\T \hat{\beta}_k)$ and the diagonal elements of $\tilde{\Sigma}_t$ are $1+ x_t^\T H_j x_t$, $1+ x_t^\T H_k x_t$. To obtain the distribution of $(z_t^{(j)}, z_t^{(k)})\mid y, X$, it remains to marginalize over $\sigma_{jk}$ where $\sigma_{jk} \sim \Gauss(\hat{\sigma}_{jk}, s_{jk}^2)$. In Section \ref{sec:snd_stage_anal}, we show that $\Pi(\sigma_{jk} \mid y, X) \overset{d}{\to} \delta_{\sigma_{jk}^*}$; a consequence of the total variation convergence of Theorem \ref{prop:sndstagebvm} where $\sigma_{jk}^*$ is the true value of the correlation between outcomes $j$ and $k$. Hence, $\hat{\sigma}_{jk} \overset{P}{\to} \sigma_{jk}^*$. Write the marginal density $(z_t^{(j)}, z_t^{(k)})\mid \sigma_{jk}, y, X$ as the expectation over $\Pi(\sigma_{jk}\mid y, X)$, that is $p(z_t^{(j)}, z_t^{(k)})\mid y, X) = E_{\sigma_{jk}}\,p(z_t^{(j)}, z_t^{(k)})\mid \sigma_{jk}, y, X)$. Then by Assumption \ref{as:corr} the conditional density $p(z_t^{(j)}, z_t^{(k)})\mid \sigma_{jk}, y, X)$ is bounded and integrable. Thus, by the dominated convergence theorem, $p(z_t^{(j)}, z_t^{(k)})\mid \sigma_{jk}, y, X) \to p(z_t^{(j)}, z_t^{(k)})\mid \sigma_{jk}^*, y, X)$ of which a reasonable approximation is provided by $p(z_t^{(j)}, z_t^{(k)})\mid \hat{\sigma}_{jk}, y, X)$. With all the above ingredients, we approximate $p(z_t^{(j)}, z_t^{(k)} \mid y, X)$ by $\Gauss(\tilde{\mu}_t, \hat{\Sigma}_t)$, where $\tilde{\mu}_t$ is as defined before and $\hat{\Sigma}_t$ has the same diagonal elements as $\tilde{\Sigma}_t$ and the off-diagonal elements are $\hat{\sigma}_{jk}$. Finally, a predictive sample for $(y_t^{(j)}, y_t^{(k)})$ can be drawn by first drawing the latent variables $(z_t^{(j)}, z_t^{(k)}) \sim \Gauss(\tilde{\mu}_t, \hat{\Sigma}_t)$ and then thresholding these latent variables. In our numerical experiments, this provided reliable approximations to the predictive distribution. For example, when $(n, p, q) = (150, 5, 2)$ and with 50 test points, the approximation maintained an average 1\% error rate in estimating the posterior predictive mean for different values of the underlying correlation coefficient in $[-1,1]$; here we treated the posterior predictive mean computed using the data augmented MCMC sampler as our benchmark.

\section{Runtime}\label{sec:runtime}
We compare the runtime of the proposed method with the recently developed fMVP method \citep{pichler2020new} which uses the algorithm from \cite{chen2018end} to compute Gaussian orthant probabilities in parallel and imposes an elastic net penalty on the correlation matrix for handling large numbers of outcomes. 
We leave out TSF of \cite{ting2022fast} as their code is not publicly available (authors declined our request) and the code used to produce results for their method in the previous subsection is not optimized. In our runtime analysis we vary the sample size $n$ over a grid from 50 to 500 with increments of 50 and for each sample size we consider 3 different choices of the number of outcomes: $q = 0.1n, 0.3n, 0.5n$. The number of covariates $p =5$ and the true regression matrix $B^*$ is simulated according the ``dense" case listed in the previous subsection. The true correlation matrix is set as $\Sigma^* = \rho 11^\T + (1-\rho) \mathrm{I}_q$ with $\rho = 0.5$. For each value of $(n,q)$ we repeated the procedures 20 times and report the average runtime and average error where the error for each run is computed as $\norm{\hat{B} - B^*}_F/pq + \norm{\hat{\Sigma} - \Sigma^*}_F/q^2$. 
The runtime results are reported in Figure \ref{fig:runtime_all1}, based on analyses conducted on an Intel i7-8700K CPU with 3.7 GHz processor; the corresponding plot showing the errors are provided in the supplement. Evidently, bigMVP improves upon the computing time by an order of magnitude while also performing better in estimation accuracy, especially in small sample sizes. In addition, unlike the competing fMVP method, bigMVP provides standard errors and uncertainty intervals (credible intervals in our case).
\begin{figure}
    \centering
    \includegraphics[width = 0.8\textwidth, height = 7cm]{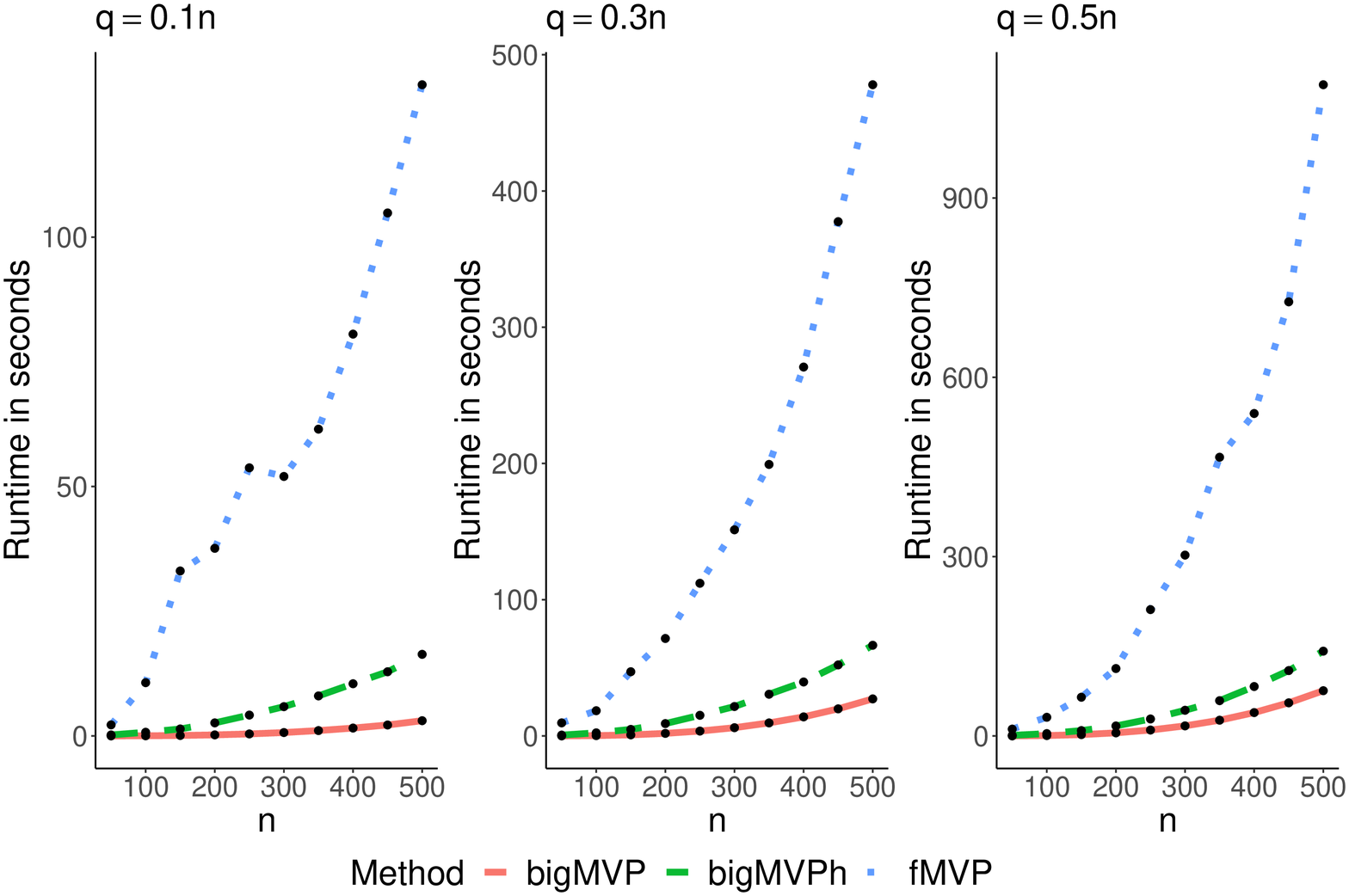}
    \caption{Runtime of the hierarchical (red solid), non-hierarchical (green dashed) versus the fast regularized method (blue dotted) from \cite{pichler2020new} for different numbers of data points $n$ and binary outcomes $q$.}
    \label{fig:runtime_all1}
\end{figure}
\begin{figure}
    \centering
    \includegraphics[width = 0.8\textwidth, height = 7cm]{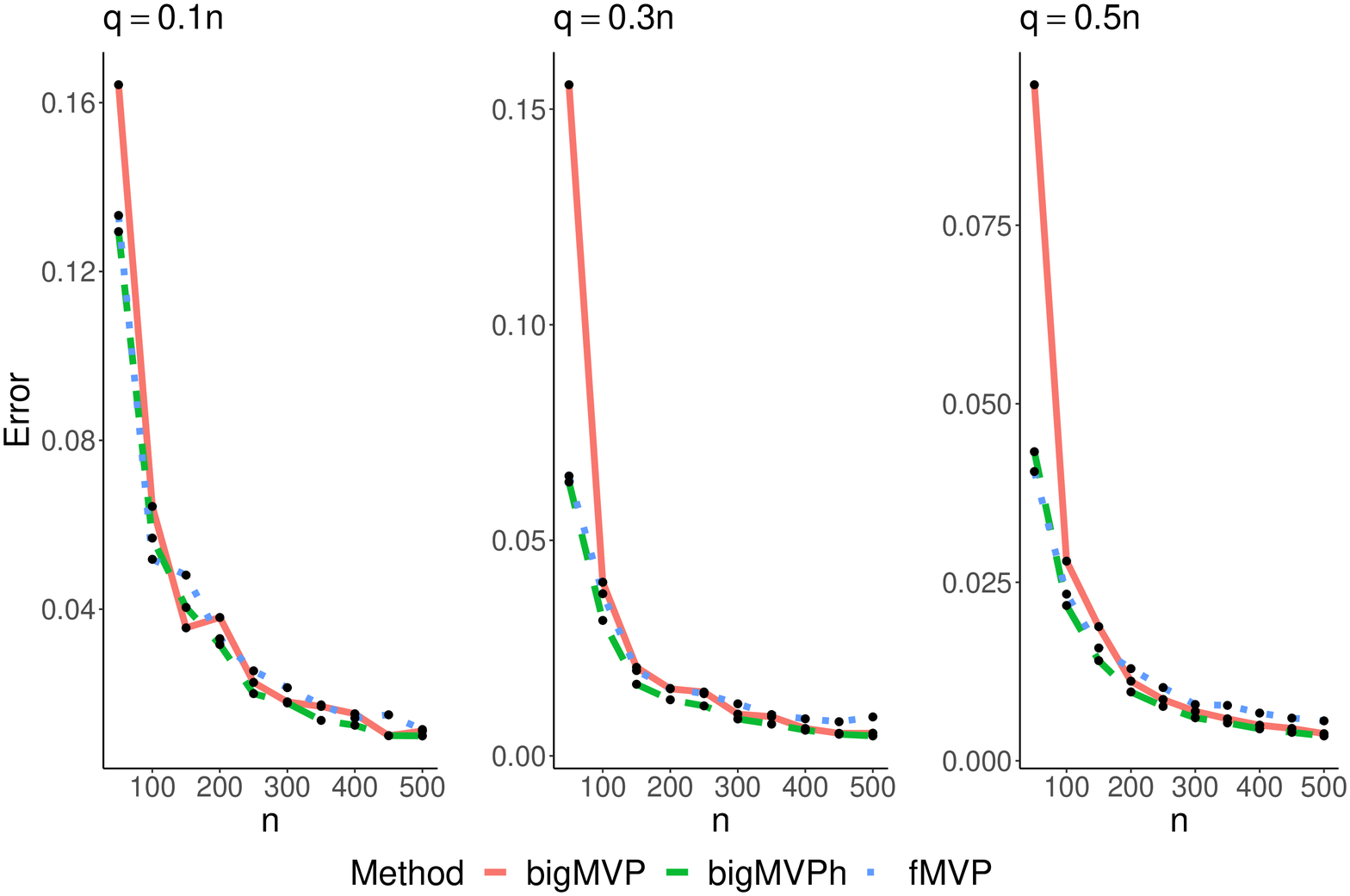}
    \caption{Estimation error of the hierarchical (red solid), non-hierarchical (green dashed) versus the fast regularized method (blue dotted) from \cite{pichler2020new} for different choices of the number of data points $n$ and the number of binary outcomes $q$.}
    \label{fig:runtime_all2}
\end{figure}

\section{Simulation tables and plots} \label{sec:app_tables}
In this section we report results from additional simulation experiments when the model is misspecified. Specifically, in Figures \ref{fig:accuracy_E1_wrong} and \ref{fig:accuracy_E2_wrong}, we plot $\norm{B^* - \hat{B}}_F/pq$ (E1) and $\norm{\Sigma^* - \hat{\Sigma}}_F/q^2$ (E2) when data are generated from  $t_{10}({B^*}^\T x, \Sigma^*)$ and a MVP model is fitted to the data. In Figure \ref{fig:int_width}, the length of 95\% credible intervals are shown when $n = 200$ and $q = 10, 15, 20$, and data are generated under an MVP model. Tables \ref{tab:coverage_table} and \ref{tab:error_table} contains numerical results corresponding to Figures \ref{fig:accuracy_E1_well} and \ref{fig:accuracy_E2_well} in the main document.
\begin{figure}
    \centering
    \includegraphics[width = 0.8\textwidth, height = 9cm]{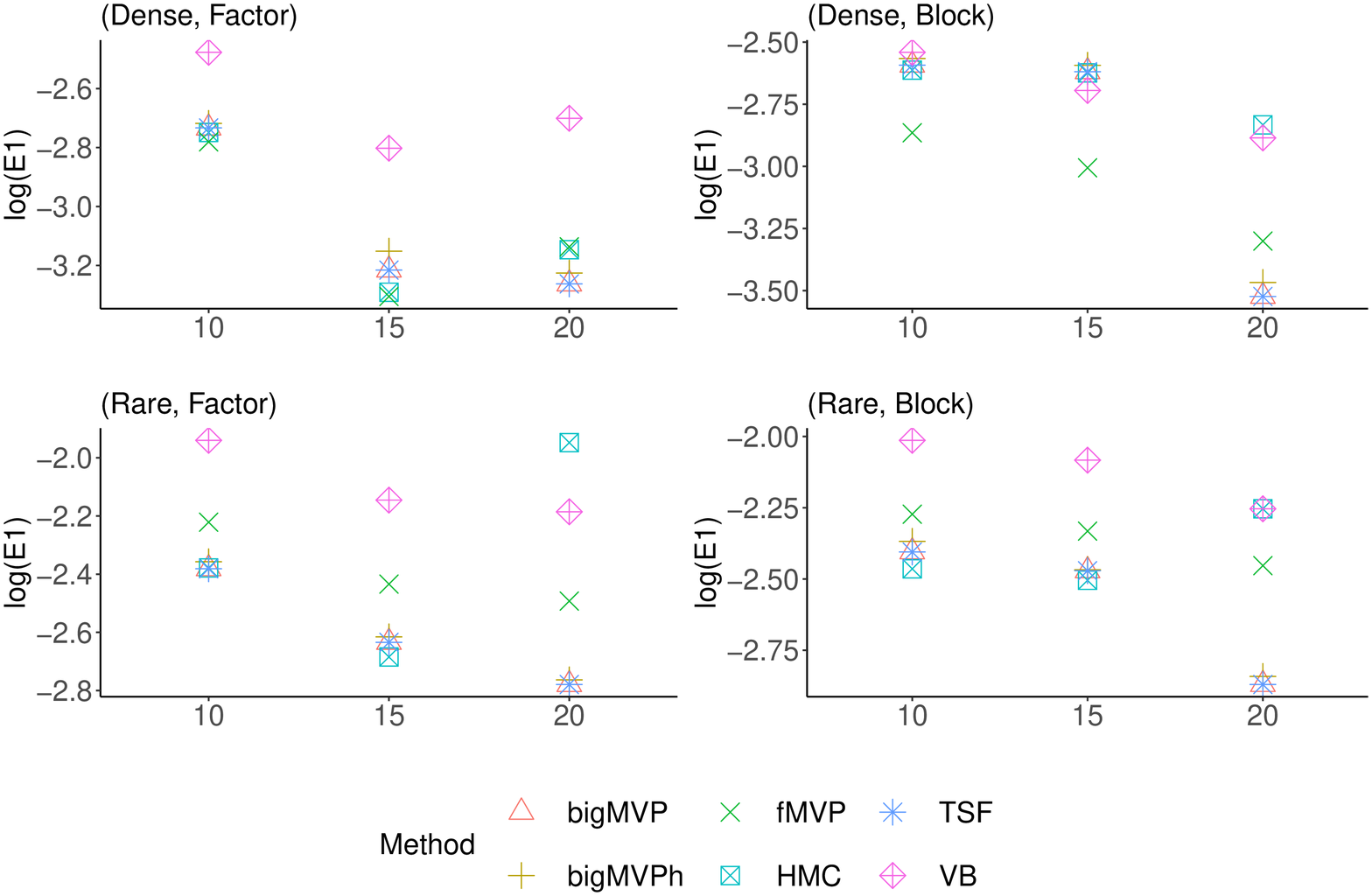}
    \caption{Logarithm of Frobenius errors in estimating the matrix of regression coefficients $B^*$ when the sample size $n = 200$, number of covariates $p = 5$ and the number of binary responses considered are $q = 10, 15, 20$. The latent variables $z \sim t_{10}({B^*}^\T x, \Sigma^*)$.}
    \label{fig:accuracy_E1_wrong}
\end{figure}

\begin{figure}
    \centering
    \includegraphics[width = 0.8\textwidth, height = 9cm]{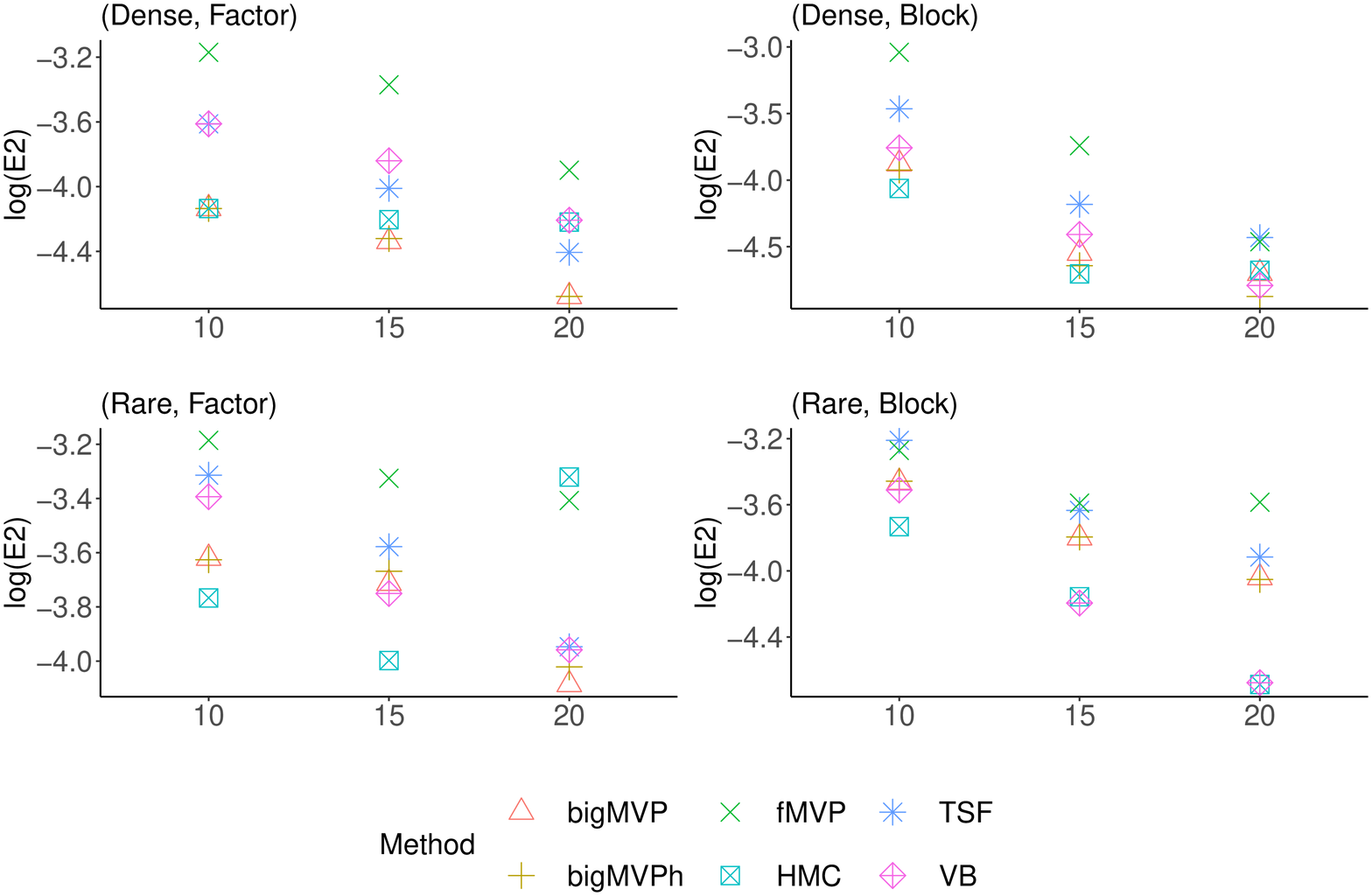}
    \caption{Logarithm of Frobenius errors in estimating the matrix of correlation coefficients $B^*$ when the sample size $n = 200$, number of covariates $p = 5$ and the number of binary responses considered are $q = 10, 15, 20$. The latent variables $z \sim t_{10}({B^*}^\T x, \Sigma^*)$. }
    \label{fig:accuracy_E2_wrong} 
\end{figure}

\begin{figure}
    \centering
    \includegraphics[width = 0.8\textwidth, height = 9 cm]{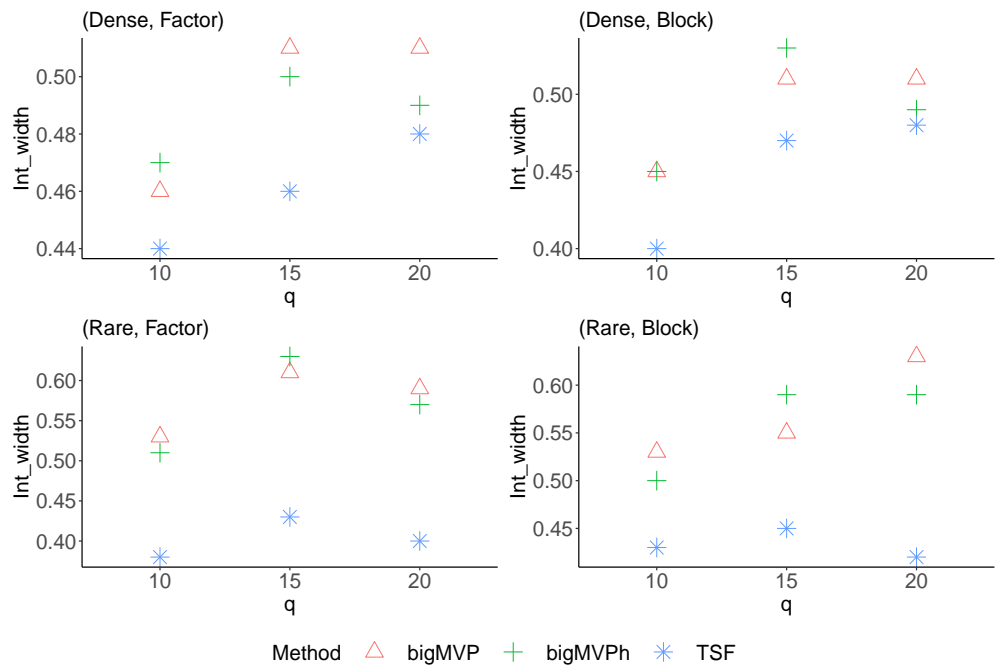}
    \caption{Comparison of widths of 95\% credible/confidence intervals for the correlation coefficients of bigMVP, bigMVPh and TSF for $n=200$ }
    \label{fig:int_width}
\end{figure}
\begin{table}[ht]
    \centering
    \scalebox{0.65}{
    \begin{tabular}{c|ccccccccccccc}
    \hline
    & &  \multicolumn{3}{c}{(Dense, Factor)} & \multicolumn{3}{c}{(Rare, Factor)} & \multicolumn{3}{c}{(Dense, Block)}  & \multicolumn{3}{c}{(Rare, Block)}\\
    \hline
    & &  bigMVP & bigMVPh & TSF & bigMVP & bigMVPh  & TSF & bigMVP & bigMVPh  & TSF & bigMVP & bigMVPh & TSF  \\
    \hline
    \multirow{3}{*}{$n=200$} & $q=10$ &  92.11 & 93.31 & 88.04 & 98.77 &98.87 & 35.24 & 93.41 & 93.91 & 81.56 & 96.53 & 96.02 & 58.04  \\
    & $q=15$ & 93.3 & 93.18 & 87.32 & 97.44 & 97.79 & 47.82 & 93.16 & 93.2 & 87.19 & 97 & 97.42 & 49.51   \\
    & $q = 20$ & 93.89 & 93.81 & 87.63 & 98 & 97.67 & 43.34 & 92.95 & 92.95 & 87.84 & 98.27 & 98.05 & 42.67 \\
    \hline
    \multirow{3}{*}{$n=500$} & $q=10$ & 94.28 & 94.27 & 92.93 & 95.06 & 95.08 & 71.27 & 94.37 & 94.49 & 91.87 & 95.37 & 95.28 & 71.53\\
    & $q=15$  & 94.09 & 94.2 & 92.39 & 95.98 & 95.89 & 65.67 & 95.18 & 95.05 & 92.93  & 95.18 & 95.05 & 70.03\\
    & $q = 20$ & 93.75 & 93.78 & 91.17 & 95.67 & 95.57 & 66.45 & 93.76 & 93.87 & 92.37 & 94.86 & 94.72 & 69.18\\
    \hline
    \end{tabular}}
    \caption{Comparison of bigMVP and bigMVPh versus TSF in terms coverage probabilities. Here for each case, average coverage probabilities are reported across all $q(q-1)/2$ correlation coefficients.}
    \label{tab:coverage_table}
\end{table}
\begin{landscape}
\begin{table}
    \centering
    \scalebox{0.55}{
    \begin{tabular}{c|cccccccccccccccccccccccccc}
    \hline
    & & & \multicolumn{6}{c}{(Dense, Factor)} & \multicolumn{6}{c}{(Rare, Factor)} & \multicolumn{6}{c}{(Dense, Block)} & \multicolumn{6}{c}{(Rare, Block)}\\
    \hline
    & & & bigMVP & bigMVPh &  HMC & TSF & VB & fMVP & bigMVP & bigMVPh &  HMC & TSF & VB & fMVP & bigMVP & bigMVPh &  HMC & TSF & VB & fMVP & bigMVP & bigMVPh &  HMC & TSF & VB & fMVP \\
    \hline
    \multirow{10}{*}{$n=200$} & \multirow{2}{*}{q=10} & E1 & 0.025 & 0.025 & 0.035 & 0.033 & 0.074 & 0.037 &  0.034 & 0.034 & 0.77 & 0.080 & 0.209 & 0.215 &  0.023 & 0.022 & 0.032 & 0.030 & 0.123 & 0.129 & 0.034 & 0.034 & 1.10 & 0.071 & 0.246 & 0.21\\
     & & E2 & 0.023 & 0.023 & 0.022& 0.027 & 0.028 & 0.024 & 0.034 & 0.034 & 0.032 & 0.072 & 0.034 & 0.155 & 0.023 & 0.022 & 0.018 & 0.027 & 0.03 & 0.155 & 0.030 & 0.031 & 0.025 & 0.071 & 0.048 & 0.159\\
     
     \cmidrule{2-27}
      & \multirow{2}{*}{$q=15$} & E1 & 0.019 & 0.019 & 0.026 & 0.025 & 0.067 & 0.023 & 0.026 & 0.026 & 0.724 & 0.064 & 0.107 & 0.057 & 0.019 & 0.019 & 0.037 & 0.025 & 0.07& 0.027 & 0.026 & 0.026 & 1.77 & 0.062 & 0.131 &  0.055\\
     & & E2 & 0.015 & 0.015 & 0.014 & 0.018 & 0.023 & 0.015 & 0.022 & 0.022 & 0.021 & 0.048 & 0.024 & 0.023 & 0.016 & 0.015 & 0.011 & 0.021 & 0.015 & 0.015 & 0.019 & 0.020 & 0.015 & 0.050 & 0.016 & 0.016\\
     
     \cmidrule{2-27}
     & \multirow{2}{*}{$q=20$} & E1 & 0.016 & 0.016 & 0.031 & 0.027 & 0.053 & 0.019 & 0.024 & 0.023 & 0.25 & 0.059 & 0.098 & 0.045 & 0.017 & 0.017 & 0.025 & 0.023 & 0.047 & 0.018 & 0.023 & 0.023 & 1.16 & 0.058 & 0.089 &  0.043\\
     & & E2 & 0.012 & 0.012 & 0.011 & 0.015 & 0.015 & 0.011 & 0.017 & 0.016 & 0.016 & 0.036 & 0.017 & 0.017 & 0.012 & 0.012 & 0.009  & 0.014 & 0.01 & 0.009 & 0.014 & 0.015 & 0.011 & 0.037 & 0.008 &  0.008\\
     
     \cmidrule{2-27} \\
      & \multirow{2}{*}{$q=100$} & E1 & 0.007 & 0.007 & & & 0.016 & 0.015 & 0.009 & 0.008 & & & 0.021 & 0.017 & 0.007 & 0.007  & & & 0.013 & 0.011 & 0.009 & 0.009 & & & 0.026 & 0.024 \\
     & & E2 & 0.002 & 0.002 & & & 0.008 & 0.007 & 0.003 & 0.003 & & & 0.008 & 0.008 & 0.002 & 0.002 & & & 0.009 & 0.008 & 0.002 & 0.002 & & & 0.005& 0.008\\
     \cmidrule{2-27}
      & \multirow{3}{*}{$q=200$} & E1 & 0.005 & 0.005 & & & 0.006 & 0.005 & 0.005 & 0.005 & &  & 0.005 & 0.005 & 0.006 & 0.006 & & & 0.006& 0.005 & 0.005 & 0.005  & &  & 0.007 & 0.006\\
     & & E2 & 0.001 & 0.001  & &  & 0.005 & 0.004 & 0.001 & 0.001 & & & 0.005 & 0.005 &0.001 & 0.001&  &  & 0.004 & 0.004 & 0.001 & 0.001 & &  &0.006 & 0.004\\
     
     \hline
     \multirow{10}{*}{$n=500$} & \multirow{2}{*}{q=10} & E1 & 0.016 & 0.016 & 0.018 & 0.017 & 0.035 & 0.020 &  0.021 & 0.021 & 0.047 & 0.039 & 0.082 & 0.069 & 0.015 & 0.014 & 0.016 & 0.016 & 0.072 & 0.027& 0.021 & 0.021 & 0.033 & 0.030 & 0.137 & 0.071   \\
     & & E2 & 0.014 & 0.013 & 0.013 & 0.016 & 0.032 & 0.025 & 0.028 & 0.028 & 0.026 & 0.050 & 0.050 & 0.050 & 0.015 & 0.014 & 0.013 & 0.016 & 0.028 & 0.024 & 0.028 & 0.029 & 0.021 & 0.051 & 0.029 & 0.029 \\
     
     \cmidrule{2-27}
      & \multirow{2}{*}{$q=15$} & E1 & 0.014 & 0.013 & 0.016 & 0.015 & 0.055 & 0.018 & 0.017 & 0.017 & 0.099 & 0.028 & 0.107 & 0.052 & 0.012 & 0.012 & 0.013 & 0.012 & 0.042 & 0.011 &   0.016 & 0.016 & 0.060 & 0.035 & 0.108 & 0.055\\
     & & E2 & 0.009 & 0.009 & 0.008 & 0.010 & 0.022 & 0.013 & 0.019 & 0.019 & 0.018 & 0.034 & 0.023 & 0.023 & 0.010 & 0.010 & 0.008 & 0.010 & 0.012 & 0.012 & 0.017 & 0.017 & 0.011 & 0.036 & 0.014 & 0.014\\
      
     \cmidrule{2-27}
     & \multirow{2}{*}{$q=20$} & E1 & 0.010 & 0.010 & 0.015 & 0.014 & 0.056 & 0.017 & 0.014 & 0.014 & 0.025 & 0.024 & 0.094 & 0.047& 0.011 & 0.010 & 0.011 & 0.011 & 0.047 & 0.012 & 0.014 & 0.014 & 0.018 & 0.024 & 0.099 & 0.048\\
     & & E2 & 0.007 & 0.007 & 0.008 & 0.009 & 0.018 & 0.012& 0.015 & 0.014 & 0.014 & 0.024 & 0.017 & 0.016 & 0.007 & 0.007 & 0.006 & 0.010 & 0.009 & 0.007& 0.012 & 0.011 & 0.009 & 0.027 & 0.010  & 0.010\\
     
     \cmidrule{2-27} \\
      & \multirow{2}{*}{$q=100$} & E1 & 0.005 & 0.004 & & & 0.012 & 0.009 & 0.005 & 0.005 & & & 0.016 & 0.012 & 0.004 & 0.004  & & &  0.010 & 0.009 & 0.005 & 0.005 & & & 0.019& 0.018 \\
     & & E2 & 0.001 & 0.001 & & & 0.006 & 0.005 & 0.003 & 0.003 & & & 0.007 & 0.007 & 0.002 & 0.001 & & & 0.004 & 0.005 & 0.003 & 0.002 & & & 0.006& 0.006\\
     \cmidrule{2-27}
      & \multirow{2}{*}{$q=200$} & E1 & 0.003 & 0.003 & & & 0.005 & 0.005 & 0.003 & 0.003 & & & 0.011 & 0.007 & 0.004 & 0.003  & & & 0.009 & 0.008& 0.003 & 0.003 & & & 0.010 & 0.008 \\
     & & E2 & 0.0007 & 0.0007 & & & 0.004 & 0.004 & 0.001 & 0.001 & & & 0.005 & 0.003 & 0.0008 & 0.0008 & & & 0.006 & 0.006 & 0.001 & 0.001 & & & 0.005 & 0.005\\
   \hline     
    \end{tabular}
    }
    \caption{Comparison of bigMVP, bigMVPh versus TSF, HMC, VB and fMVP in terms of error in estimating the coefficient matrix $B^*$ (E1) and correlation matrix $\Sigma^*$ (E1). In particular, $\mathrm{E}1 = \norm{B^* - \hat{B}}_F/pq$  and $\mathrm{E}2 = \norm{\Sigma^* - \hat{\Sigma}}_F/q^2$.}
    \label{tab:error_table}
\end{table}

\end{landscape}

\section{Auxiliary results}

\begin{lemma}[Lemma S6.2 of \cite{miller2021asymptotic}] \label{lem.equilip}
   Let $E \subseteq \mathbb{R}^{D}$ be open and convex, and let $f_{n}: E \rightarrow \mathbb{R}$ for $n \in \mathbb{N}$. For any $k \in \mathbb{N}$, if each $f_{n}$ has continuous $k$th-order derivatives and $\left(f_{n}^{(k)}\right)$ is uniformly bounded, then $\left(f_{n}^{(k-1)}\right)$ is equi-Lipschitz.
\end{lemma}

\begin{lemma} \label{lem.quadbound}
	Let $\{x_i\}$ be a sequence of $p$-dimensional vectors where each $x_i\in \RR^p$ and  $\lim _{n \rightarrow \infty} n^{-1} \Sigma_{i=1}^{n}$ $x_i x_i^{\prime}$ is a finite nonsingular matrix. For any $p\times p$ square matrix $A$, we let its trace $\tr(A)$ be the sum of its diagonal elements, i.e., $tr(A) = \sum_{i=1}^p A_{ii}$.  Let $\{a_i\},\{b_i\}$ be real number sequences and $\{b_i\}$ are all positive. Then
$$
\sum_{i=1}^{n} (a_i x_i)^T \left(\sum_{i=1}^{n} b_i x_ix_i^T\right)^{-1} \sum_{i=1}^{n} (a_i x_i) \leq p \max_i \frac{a_i^2}{b_i}
$$
\end{lemma}

\begin{proof}
We have,
\begin{eqnarray}
	\frac{1}{n}\sum_{i=1}^{n} (a_i x_i)^T \left(\sum_{i=1}^{n} b_i x_ix_i^T\right)^{-1} \sum_{i=1}^{n} (a_i x_i) 
	&=& \frac{1}{n}\tr\left[\sum_{i=1}^{n} (a_i x_i)^T \left(\sum_{i=1}^{n} b_i x_ix_i^T\right)^{-1} \sum_{i=1}^{n} (a_i x_i)\right]  \nonumber\\
& =& \frac{1}{n}\tr\left[ \left(\sum_{i=1}^{n} b_i x_ix_i^T\right)^{-1} \sum_{i=1}^{n} (a_i x_i) \sum_{i=1}^{n} (a_i x_i)^T\right]\nonumber \\
& &\hspace{-1.5in}=\frac{1}{n}\tr\left[ \left(\sum_{i=1}^{n} b_i x_ix_i^T\right)^{-1/2} \sum_{i=1}^{n} (a_i x_i) \sum_{i=1}^{n} (a_i x_i)^T \left(\sum_{i=1}^{n} b_i x_ix_i^T\right)^{-1/2}\right]\nonumber\\
\label{eq:trace}
\end{eqnarray}
Also, the matrix
\begin{align}
   A =  \sum_{i=1}^n\left[ a_ix_i- \sum_{i=1}^{n} \frac{(a_i x_i)}{n}\right] \left[ a_ix_i- \sum_{i=1}^{n} \frac{(a_i x_i)}{n}\right]^T & =\sum_{i=1}^n a_i^2 x_ix_i^T -\frac{1}{n} \sum_{i=1}^{n} (a_i x_i) \sum_{i=1}^{n} (a_i x_i)^T  \nonumber
\end{align}
is positive definite.
Plugging the preceding display back to \eqref{eq:trace}, we have 
\begin{align*}
 & \frac{1}{n}\tr\left[ \left(\sum_{i=1}^{n} b_i x_ix_i^T\right)^{-1/2} \sum_{i=1}^{n} (a_i x_i) \sum_{i=1}^{n} (a_i x_i)^T \left(\sum_{i=1}^{n} b_i x_ix_i^T\right)^{-1/2}\right]  \\
 & \leq \tr\left[ \left(\sum_{i=1}^{n} b_i x_ix_i^T \right)^{-1/2} \sum_{i=1}^n a_i^2 x_ix_i^T \left(\sum_{i=1}^{n} b_i x_ix_i^T\right)^{-1/2}\right]  \\
 & = \tr\left[ \left(\sum_{i=1}^{n} b_i x_ix_i^T\right)^{-1} \sum_{i=1}^n a_i^2 x_ix_i^T \right] \\
 & \leq \max_i \frac{a_i^2}{b_i}\tr\left[\left(\sum_{i=1}^{n} b_i x_ix_i^T\right)^{-1} \left(\sum_{i=1}^{n} b_i x_ix_i^T\right)\right] \\
 & = p \max_i \frac{a_i^2}{b_i}
\end{align*} 
\end{proof}

\section{Some Useful Quantities}\label{sec.quantities}
Here, we present the detailed expressions of quantities used in \Cref{sec:theory}. Set $\widehat{\Sigma}_{12} = \{(1, r_{i1}r_{i2}\sigmatru)^\T; (r_{i1}r_{i2}\sigmatru, 1)^\T\}$ and $\mu_1^* = x_i^\T \beta_1^*$, $\mu_2^* = x_i^\T \beta_2^*$. We have,
\begin{equation}
    R^{12}_1  = \lim_{n \rightarrow \infty} \frac{1}{n} \sum_{i=1}^n
    \Cov_{\theta^*}\left(\frac{\partial \ell_{i}^1}{\partial \beta_1}, \frac{\partial \ell^{12}_{i}}{\partial \sigma_{12}} \right) = 0
\end{equation}
where
\begin{align*}
	\Cov_{\theta^*} \left(\frac{\partial \ell_{i}^1}{\partial \beta_1}, \frac{\partial \ell^{12}_{i}}{\partial \sigma_{12}}\right) &= E_{\theta^*}\left(\frac{\partial \ell_{i}^1}{\partial \beta_1} \frac{\partial \ell^{12}_{i}}{\partial \sigma_{12}}\right)\\
	 & = E_{\theta^*}\left[\frac{\phi(r_{i1} \mu^*_1)}{\Phi(r_{i1} \mu^*_1)}r_{i1} \times \frac{\phi_{\widehat{\Sigma}_{12}}}{\Phi_{\widehat{\Sigma}_{12}}}r_{i1}r_{i2} x_i\right] \\
	 & = E_{\theta^*}\left[\frac{\phi(r_{i1} \mu^*_1)}{\Phi(r_{i1} \mu^*_1)}\frac{\phi_{\widehat{\Sigma}_{12}}}{\Phi_{\widehat{\Sigma}_{12}}}r_{i2} x_i\right] \\
	 & = \frac{\phi(\mu^*_1)}{\Phi(\mu^*_1)}\phi_{[1, \sigma_{12}^*; \sigma_{12}^*, 1]}(\mu^*_1,\mu^*_2)x_i -
	 \frac{\phi(\mu_1^*)}{\Phi(\mu_1^*)}\phi_{[1, -\sigma_{12}^*; -\sigma_{12}^*, 1]}(\mu_1^*,-\mu_2^*)x_i \nonumber \\
	 & + \frac{\phi(-\mu_1^*)}{\Phi(-\mu_1^*)}\phi_{[1, -\sigma_{12}^*; -\sigma_{12}^*, 1]}(-\mu_1^*,\mu_2^*)x_i - \frac{\phi(-\mu_1^*)}{\Phi(-\mu_1^*)}\phi_{[1, \sigma_{12}^*; \sigma_{12}^*, 1]}(-\mu_1^*,-\mu_2^*)x_i \\
	 & \overeq{a} 0 
\end{align*}
The equation (a) in the above display is because of $\phi_{[1, \sigma_{12}; \sigma_{12}, 1]}
(\mu_1,\mu_2) = \phi_{[1, -\sigma_{12}; -\sigma_{12}, 1]}(\mu_1,-\mu_2)$. Similarly, $R^{12}_2=\lim_{n \rightarrow \infty} \frac{1}{n}\Cov_{\theta^*}\left(\frac{\partial \ell^2_{i}}{\partial \beta_2}, \frac{\partial \ell^{12}_{i}}{\partial \sigma_{12}}\right) =0$.  For $V^{12}$, we have
\begin{equation}
    V^{12}  = \lim_{n \rightarrow \infty} \frac{1}{n} \sum_{i=1}^n
    \Cov\left[\frac{\phi(r_{i1}\mu_1^*)}{\Phi(r_{i1}\mu_1^*)}r_{i1}, \frac{\phi(r_{i2}\mu_2^*)}{\Phi(r_{i2}\mu_2^*)}r_{i2}\right]
\end{equation}

where
\begin{align*}
	 & \Cov\left[\frac{\phi(r_{i1}\mu_1^*)}{\Phi(r_{i1}\mu_1^*)}r_{i1}, \frac{\phi(r_{i2}\mu_2^*)}{\Phi(r_{i2}\mu_2^*)}r_{i2}\right] =
	            E\left[\frac{\phi(r_{i1}\mu_1^*)}{\Phi(r_{i1}\mu_1^*)} \frac{\phi(r_{i2}\mu_2^*)}{\Phi(r_{i2}\mu_2^*)} r_{i1} r_{i2}\right] \\
	           &= \sum_{r_{i1}=-1,1}\sum_{r_{i2}=-1,1} r_{i1} r_{i2} \frac{\phi(r_{i1}x_i^T\beta^*_1)}{\Phi(r_{i1}x_i^T\beta^*_1)} \frac{\phi(r_{i2}x_i^T\beta^*_2)}{\Phi(r_{i2}x_i^T\beta^*_2)}
	           \Phi_{\widehat{\Sigma}_{12}}(r_{i1}x_i^T\beta^*_1, r_{i2}x_i^T\beta^*_2 ),
\end{align*}

with $\widehat{\Sigma}_{12} = \{(1, r_{i1}r_{i2}\sigma^*_{12})^\T;(r_{i1}r_{i2}\sigma^*_{12}, 1)^\T\}$.
Now we turn to $Q^{12}_1$.We have
$$
Q^{12}_1 =  \lim_{n \rightarrow \infty} \frac{1}{n} \sum_{i=1}^n
-E_{\theta^*}\left(\sndpartial{\ell^{12}_{i}}{\sigma_{12}}{\beta_1^T}\right) 
$$

where
\begin{align}
	& -E_{\theta^*}\left(\sndpartial{\ell^{12}_{i}}{\sigma_{12}}{\beta_1^T}\right) = \Cov_{\theta^*}\left(\frac{\partial\ell^{12}_i}{\partial\sigma_{12}}, \frac{\partial\ell^{12}_i}{\partial\beta_1}\right) \nonumber \\
	& =\Cov\left(\frac{\phi_{\widehat{\Sigma}_{12}}}{\Phi_{\widehat{\Sigma}_{12}}}r_{i1}r_{i2}, \frac{\phi(r_{i1}\mu_1^*)\Phi\left(r_{i2}\frac{\mu_2^*-\sigma_{12}^*\mu_1^*}{\sqrt{1-(\sigma_{12}^*)^2}}\right)}{\Phi_{\widehat{\Sigma}_{12}}}r_{i1} x_i^T \right) \nonumber\\
	& = \sum_{r_{i1}=-1,1}\sum_{r_{i2}=-1,1} r_{i2} 
	\phi(r_{i1}x_i^T\beta^*_1)\Phi\left(r_{i2}\frac{x_i^T\beta^*_2-\sigma_{12}^*x_i^T\beta^*_1}{\sqrt{1-(\sigma_{12}^*)^2}}\right)
	\frac{\phi_{\widehat{\Sigma}_{12}}(r_{i1}x_i^T\beta^*_1, r_{i2}x_i^T\beta^*_2 )} {\Phi_{\widehat{\Sigma}_{12}}(r_{i1}x_i^T\beta^*_1, r_{i2}x_i^T\beta^*_2  )} x_i^T
\end{align}
Similarly,
\begin{align}
& Q^{12}_2 =  \nonumber\\
    &  \lim_{n \rightarrow \infty} \frac{1}{n} \sum_{i=1}^n \sum_{r_{i1}=-1,1}\sum_{r_{i2}=-1,1} r_{i1} 
	\phi(r_{i2}x_i^T\beta^*_2)\Phi\left(r_{i1}\frac{x_i^T\beta^*_1-\sigma_{12}^*x_i^T\beta^*_2}{\sqrt{1-(\sigma_{12}^*)^2}}\right)
	\frac{\phi_{\widehat{\Sigma}_{12}}(r_{i1}x_i^T\beta^*_1, r_{i2}x_i^T\beta^*_2 )} {\Phi_{\widehat{\Sigma}_{12}}(r_{i1}x_i^T\beta^*_1, r_{i2}x_i^T\beta^*_2)} x_i^T
\end{align} 

\end{document}